\documentclass[11pt,dvipsnames]{article}

\usepackage{setspace}
\usepackage{jheppub} 
\usepackage[utf8]{inputenc}
\usepackage{epsfig}
\usepackage{bbm}
\usepackage{bbding}
\usepackage{amsfonts}
\usepackage{amssymb}
\usepackage{mathtools}
\usepackage{dsfont}
\usepackage{bm}
\usepackage{graphicx}
\usepackage{longtable}
\usepackage{pdflscape}
\usepackage{subcaption}
\usepackage{psfrag}
\usepackage{amsthm}
\usepackage[capitalise]{cleveref}
\usepackage{tikz}
\usetikzlibrary{backgrounds}
\usepackage{pgfplots}
\usepackage{pgfplotstable}
\usepackage{subfiles}
\usepackage{longtable}
\usepackage{tabularx}
\usepackage{ltablex}
\usepackage{booktabs}
\usepackage[export]{adjustbox}
\usepackage{listings}
\usepackage{multirow} 
\usepackage{comment}
\usepackage{cancel,slashed}
\usepackage{bbm}
\usepackage{graphicx}
\usepackage{latexsym}
\usepackage{transparent}
\usepackage{mathtools}
\usepackage{array}
\usepackage{makecell}
\usepackage{graphics,psfrag}
\usepackage{placeins}
\usepackage{nowidow}
\usepackage{listings}
\usepackage[normalem]{ulem}
\usepackage{environ}
\usepackage{tikz}
\usepackage{enumitem}
\usetikzlibrary{positioning,matrix,arrows.meta,calc}
\usetikzlibrary{positioning,trees,decorations.pathmorphing,decorations.markings,decorations.pathreplacing,calc,shapes,shapes.geometric,shapes.symbols,patterns,arrows}

\usetikzlibrary{fadings}
\usetikzlibrary{decorations.shapes}
\usepackage{amsmath}
\usepackage{amssymb}
\usepackage{amsthm}
\usepackage{algorithm}
\usepackage{algpseudocode}
\usepackage{placeins}

\usepackage{quiver}

\newcommand{\newshortstack}[1]
{\begingroup\renewcommand{\arraystretch}{1.1}
\ifmmode
\begin{array}{c}#1\end{array}%
\else
\begin{tabular}{c}#1\end{tabular}%
\fi
\endgroup}

\pgfplotsset{
        compat=1.9,
        compat/bar nodes=1.8,
    }

\theoremstyle{definition} 
\newtheorem{theorem}{Theorem}
\newtheorem{prop}[theorem]{Proposition}
\newtheorem{lem}[theorem]{Lemma}
\newtheorem{corollary}[theorem]{Corollary}
\newtheorem{definition}[theorem]{Definition}
\newtheorem{conjecture}[theorem]{Conjecture}

\newtheorem*{theorem*}{Theorem}
\newtheorem*{prop*}{Proposition}
\newtheorem*{lem*}{Lemma}

\crefname{theorem}{Th.}{Ths.}
\crefname{definition}{Def.}{Defs.}
\crefname{corollary}{Cor.}{Cors}
\crefname{prop}{Prop.}{Props.}
\crefname{conjecture}{Conj.}{Conjs.}
\crefname{lem}{Lem.}{Lems.}
\crefformat{section}{#2§#1#3}

\Crefname{theorem}{Th.}{Ths.}
\Crefname{definition}{Def.}{Defs.}
\Crefname{corollary}{Cor.}{Cors}
\Crefname{prop}{Prop.}{Props.}
\Crefname{conjecture}{Conj.}{Conjs.}
\Crefname{lem}{Lem.}{Lems.}

\newcommand{\be}{\begin{equation}}
	\newcommand{\ee}{\end{equation}}
\newcommand{\bea}{\begin{eqnarray}}
	\newcommand{\eea}{\end{eqnarray}}

\renewcommand{\epsilon}{\varepsilon}

\newcommand{\ben}{\begin{enumerate}}
	\newcommand{\een}{\end{enumerate}}
\newcommand{\bei}{\begin{itemize}}
	\newcommand{\eei}{\end{itemize}}
	
	\makeatletter
\tikzset{
    dot diameter/.store in=\dot@diameter,
    dot diameter=3pt,
    dot spacing/.store in=\dot@spacing,
    dot spacing=10pt,
    dots/.style={
        line width=\dot@diameter,
        line cap=round,
        dash pattern=on 0pt off \dot@spacing
    }
}

\usepackage{setspace}
\usepackage{graphicx} 
\usepackage{lmodern}
\usepackage{xspace}
\usepackage{array} 
\usepackage{etoolbox}
\usetikzlibrary{positioning,trees,decorations.pathmorphing,decorations.markings,decorations.pathreplacing,calc,shapes,shapes.geometric,shapes.symbols,patterns,arrows}
	
\tikzset{decorate sep/.style 2 args=
{decorate,decoration={shape backgrounds,shape=circle,shape size=#1,shape sep=#2}}}
\usepackage{tikz}
\usepackage{latexsym}
\usepackage{graphicx}
\usepackage{svg}
\usepackage{tikz-3dplot}
\usetikzlibrary{arrows.meta}
\usetikzlibrary{patterns}

\definecolor{bluetto}{HTML}{0088ff}
\definecolor{snowfl}{HTML}{00ace6}
\definecolor{dgreen}{HTML}{298A08}
\usepackage[most]{tcolorbox}

\newcommand{\cO}{\mathcal{O}}

\newcommand{\cE}{\mathcal{E}}

\newcommand{\cN}{\mathcal{N}}

\usepackage{setspace}
\makeatletter
\renewcommand*\env@matrix[1][\arraystretch]{%
  \edef\arraystretch{#1}%
  \hskip -\arraycolsep
  \let\@ifnextchar\new@ifnextchar
  \array{*\c@MaxMatrixCols c}}
\makeatother

\usepackage{hyperref}

\begin{document}

	\pagestyle{plain}

	\makeatletter
	\@addtoreset{equation}{section}
	\makeatother
	\renewcommand{\theequation}{\thesection.\arabic{equation}}
	\pagestyle{empty}

\vspace{2cm}

\begin{center}
\phantom{a}\\
\vspace{0.8cm}
\scalebox{0.90}[0.90]{{\fontsize{24}{30} \bf{Holes in Calabi--Yau Effective Cones}}}\\
\end{center}

\vspace{0.4cm}
\begin{center}
\scalebox{0.95}[0.95]{{\fontsize{12}{30}\selectfont  
Naomi Gendler,$^{a}$
Elijah Sheridan,$^{b}$
Michael Stillman,$^{c}$ and
David H. Wu$^{a}$
}}
\end{center}

\begin{center}
\vspace{0.25 cm}

\textsl{$^{a}$Jefferson Physical Laboratory, Harvard University, Cambridge, MA 02138, USA}\\  
\textsl{$^{b}$Department of Physics, Cornell University, Ithaca, NY 14853, USA}\\
\textsl{$^{c}$Department of Mathematics, Cornell University, Ithaca, NY 14853, USA}\\

  \vspace{1.1cm}
	\normalsize{\bf Abstract} \\[8mm]
\end{center}

\begin{center}
	\begin{minipage}[h]{15.0cm}

Motivated by their role in non-perturbative potentials in string theory, we study divisors in effective cones of Calabi--Yau threefolds. We give examples of geometries for which some divisor classes in the effective cone are not themselves effective: i.e., they have no global sections. We call these non-holomorphic divisor classes ``holes,'' and characterize their behavior in an ensemble of toric hypersurface Calabi--Yau threefolds. We prove some necessary and sufficient conditions for the existence of holes, show consequences of holes that follow from the minimal model program, and demonstrate that a class of holes come in semigroups (with this class conjectured to constitute all holes). Furthermore, we provide moduli-dependent bounds on the volumes of four-cycles representing holes.

\end{minipage}
\end{center}
	\newpage
	\setcounter{page}{1}
	\pagestyle{plain}
	\setcounter{footnote}{0}
\tableofcontents
        \newpage
\section{Introduction}

In many theories, extended objects generically give rise to important non-perturbative physics. In supersymmetric string compactifications, extended objects wrapped on closed cycles in compact extra dimensions are solitonic states; when the cycle is holomorphic, these are Bogomol'nyi–Prasad–Sommerfeld (BPS) states, which give rise to well-controlled non-perturbative phenomena. In particular, BPS states have masses or tensions given by exact supersymmetric quantities readily computable from the geometry of the cycles, while the masses or tensions of non-BPS states remain elusive. Therefore, to characterize the spectrum in a given compactification, one must determine which states are BPS and which are not.

An important class of examples is the set of effective quantum gravitational theories arising from compactifications of type IIB string theory on Calabi--Yau orientifolds. These theories provide a testing ground for connecting string theory with real-world phenomenology and have been extensively studied (see, e.g.,~\cite{Cicoli:2012sz, Gendler:2023kjt, Demirtas:2021nlu, McAllister:2024lnt,Sheridan:2024vtt, Reece:2025thc,Maharana:2012tu,Marchesano:2024gul,vandeHeisteeg:2023uxj,Jain:2025vfh,Mehta:2021pwf,McAllister:2025qwq}). Much of the low-energy physics can be determined by the superpotential $W$ and the K\"ahler potential $\mathcal{K}$ of the underlying theory. A crucial task in characterizing the set of $\mathcal{N}=1$ effective theories is to determine what non-perturbative effects produce corrections to $W$ and $\mathcal{K}$.\footnote{The role of such potentials along with their non-perturbative effects in quantum gravity have been analyzed recently in light of the Swampland program via new holographic constraints \cite{Bedroya:2025fie,Bedroya:2025ltj}. To this end, there has also been significant recent progress made in studying the non-geometric flux vacua of type IIB string theory~\cite{Chen:2025rkb,Mohseni:2025tig,Becker:2022hse,Becker:2024ijy,Bardzell:2022jfh,Rajaguru:2024emw}.} One such set of corrections is that of Euclidean D3-branes wrapping four-cycles in the Calabi--Yau threefold, giving rise to instantons in the 4d $\cN=1$ theory. A necessary (but certainly not sufficient) condition for a particular Euclidean D3-brane to contribute to $W$ is that it be BPS. Geometrically, such contributions arise only when the Euclidean D3-branes wrap holomorphic four-cycles in the compact manifold. On the other hand, Euclidean D3-branes that wrap non-holomorphic four-cycles are non-BPS instantons and can contribute only to the K\"ahler potential. The task of categorizing divisor classes into holomorphic and non-holomorphic is therefore a crucial step in studying the four-dimensional physics arising from compactifications of string theory on Calabi--Yau threefolds. The purpose of this work is to make progress on this objective.

Given a Calabi--Yau threefold $X$, holomorphic (i.e., effective) divisors live in a cone, $\mathcal{E}_X \subset H^2(X, \mathbb{R})$. In this work we study the divisor classes in $\mathcal{E}_X \cap H^2(X,\mathbb{Z})$ that are not themselves effective. We refer to such classes as \textit{holes}. We provide examples of holes, and find that they admit nice structures. While the existence of holes has been previous noted,\footnote{For example, holes are accounted for in \cite{lazarsfeld2017positivity} with the notion of the ``semigroup of a divisor'' (not to be confused with the semigroups of holes we will discuss later). In particular, holes are those divisors whose semigroups are neither $\{0\}$ nor $\mathbb{Z}$ (see Def. 2.1.1 in \cite{lazarsfeld2017positivity}).} to our knowledge this is the first systematic study of their properties.

A primary focus will be placed on Calabi--Yau threefolds constructed as hypersurfaces in toric varieties from the Kreuzer--Skarke database~\cite{Kreuzer:2000xy}. In this context, we will predominantly focus on the subset of $\mathcal{E}_X$ that is generated by effective, or holomorphic, divisors inherited from the ambient toric variety $V$. We refer to the generating set of effective divisors that are inherited from effective divisors on the ambient toric variety as \textit{prime toric divisors}. We denote the cone that is generated by prime toric divisors as $\mathcal{E}_V$. 
Previous studies employing the Kreuzer--Skarke database~\cite{Gendler:2023kjt, Demirtas:2021nlu, McAllister:2024lnt, Sheridan:2024vtt,Gendler:2024adn,Cheng:2025ggf, Mehta:2021pwf} have often assumed that the most dominant four-dimensional instanton contributions come from Euclidean D3-branes wrapped on prime toric divisors, while \textit{a priori} any element of the Hilbert basis of $\mathcal{E}_V$ could furnish a leading instanton contribution.\footnote{The true leading instantons should arise from the Hilbert basis of $\mathcal{E}_X$, rather than $\mathcal{E}_V$ (the latter cone being a subset of the former), but $\mathcal{E}_X$ is difficult to compute, especially for $h^{1,1} \gtrsim 4$, so we largely content ourselves with $\mathcal{E}_V$ in this work.}
In \cref{fig:additional elements}, we illustrate how the number of Hilbert basis elements of a given effective cone greatly outnumbers the prime toric divisors, motivating the scrutiny of the assumed dominance of prime toric divisors.
Hence, one of the main goals of this work is to study the elements of the Hilbert basis of $\cE_V$ that are not prime toric divisors, which we will call \textit{non-trivial Hilbert basis elements}. In particular, we will compute their holomorphicity (determining which physical quantities they contribute to) and estimate their volumes (determining the magnitude of their contributions).

However, before restricting to toric hypersurfaces, we begin by presenting some general results on holes. In particular, we demonstrate that for any smooth Calabi--Yau threefold, holes cannot occur in a distinguished subcone of $\mathcal{E}_X$ (the cone of big movable divisors, introduced in \cref{sec:review-geometry} and related to holes by \cref{th:big_movable_CY}) and that a class of holes come in semigroups (\cref{cor:semigroup}, which holds for non-movable holes, though we discuss how the physics-inspired \cref{conj:nonbig_movable_CY} of \cite{Katz:2020ewz} implies that holes must be non-movable). We also use methods from the minimal model program to show that holes correspond to nef holes on singular birational Calabi--Yau varieties (\cref{th:CY_MMP_holes}), and can indicate the presence of torsion on these birational models (\cref{prop:torsion}). Finally, for toric hypersurface Calabi--Yau threefolds in particular, we demonstrate that non-trivial Hilbert basis elements are holes if they lie in the interior of $\mathcal{E}_V$ (\cref{prop:big}).

To complement these formal results, we conduct a computational study of Calabi--Yau threefolds in the Kreuzer--Skarke database. To achieve this, we use both \verb+CYTools+~\cite{Demirtas:2022hqf} and \verb+Macaulay2+~\cite{M2} to evaluate holomorphicity by computing line bundle cohomology using the methods described in \cref{sec:direct_comp}. In particular, we develop Python code that extends \verb+CYTools+ to compute if a divisor class on a Calabi--Yau threefold from the Kreuzer--Skarke database is effective: this code is publicly available, and can be found in the ancillary data attached to this paper. Empirically, we find that \textit{no} non-trivial Hilbert basis elements are themselves effective, extending \cref{prop:toric_big_movable} and motivating \cref{conj:non_trivial_hb_are_holes}. This is striking: effectiveness frequently differs between divisors on $V$ and on $X$ --- e.g., there are countless examples where $\mathcal{E}_X$ is strictly larger than $\mathcal{E}_V$~\cite{Demirtas:2018akl} --- so \textit{a priori}, there is no reason why non-trivial Hilbert basis elements should never be effective on the Calabi--Yau threefold.

Finally, we consider the volumes of the non-holomorphic four-cycles representing non-trivial Hilbert basis elements. Although in principle, these volumes can be directly computed from the Calabi--Yau metric, this metric is notoriously difficult to calculate in practice~\cite{Yau1978,Yau1990,Yau1991}, though progress has recently been made numerically~\cite{Ashmore:2019wzb, Douglas:2020hpv, Larfors:2021pbb, Gerdes:2022nzr, Butbaia:2024xgj}. Additionally, attempts have been made to bound volumes of non-holomorphic cycles with arguments from physics~\cite{Demirtas:2019lfi, Long:2021lon}. In this work, we study a method to put rigorous bounds on the volumes of these non-holomorphic four-cycles living in effective cones. We demonstrate the method in two explicit examples and find that there are regions in moduli space where the non-holomorphic cycle is among one of the smallest volumes, compared to the prime toric divisors. In particular, we exhibit limits in K\"ahler moduli space where the volumes of non-holomorphic cycles are exactly determined. Bounds on these volumes are important because they determine the action of instantons contributing to the K\"ahler potential and play a role in the Weak Gravity Conjecture~\cite{Arkani-Hamed:2006emk}.

In sum, in this work we come to understand the basic structure of holes in Calabi--Yau effective cones. In the toric hypersurface context in particular, our results on the non-holomorphicity of non-trivial Hilbert basis elements supports the continuation of the practice of excluding them from supersymmetric quantities, such as the superpotential. Furthermore, the methods we exhibit for bounding non-holomorphic four-cycle volumes allow one to assess the importance of non-trivial Hilbert basis elements relative to prime toric divisors in unprotected quantities, such as the K\"ahler potential.

The paper is organized as follows. In \cref{sec:review}, we review the pertinent facts about divisors in Calabi--Yau threefolds and toric varieties, defining all relevant concepts. We also explain the relevance of these geometric objects for physics in the context of string/M-theory. In \cref{sec:holeprops}, we present some formal results on the existence and structure of holes in Calabi--Yau threefold effective cones; the results of this section are summarized at its beginning, and the methods used to prove them are, for the most part, not strictly required to understand the remaining text, though they aid in the interpretation of some subsequent examples. 
In \cref{sec:direct_comp}, we lay out an algorithm for directly computing the number of global sections that a particular divisor has on a toric hypersurface. In \cref{sec:ksholes}, we use this algorithm to systematically scan Calabi--Yau geometries from the Kreuzer--Skarke database for holes in effective cones, and present some illustrative examples. In \cref{sec:vol_bounds}, we describe a method for providing upper and lower bounds on non-holomorphic cycles, and implement this method in particular geometries. Finally, we conclude and discuss future directions in \cref{sec:conclusions}.

\section{Review: Divisors in Calabi--Yau Threefolds and Toric Varieties} \label{sec:review}

\subsection{Geometry}\label{sec:review-geometry}

In this section, we review the relevant aspects of divisors in Calabi--Yau threefolds, defining all relevant quantities and pointing out the salient aspects related specifically to Calabi--Yau threefolds constructed as hypersurfaces in toric varieties.

\paragraph{Divisor Generalities.} 
\begin{definition}
    Let $X$ be a normal variety. A (Weil) \textit{divisor} $D$ in $X$ is a finite integral linear combination of \textit{prime} divisors, or closed integral subschemes of codimension one.
\end{definition}
Addition of divisors can be thought of as the union operation. If $A$ is a divisor and $D$ a prime divisor, we let $\nu_D(A)$ --- the \textit{degree of vanishing} of $D$ on $A$ --- denote the coefficient of $A$ in $D$, such that $D = \sum \nu_B(D) B$, with the sum being over prime divisors $B$. A divisor is \textit{principal} if it is the zero locus of a rational function $f$ minus its poles, denoted $\mathrm{div}(f)$. The \textit{group of divisor classes}, or \textit{class group}, is the group of divisors modulo principal divisors, and its elements are known as divisor classes. For the varieties we are most interested in, smooth Calabi--Yau varieties $X$, the group of divisor classes is isomorphic to the singular cohomology group $H^2(X, \mathbb{Z})$. Consequently, we will often refer to (co)homology classes and divisor classes interchangeably. Given a divisor $D$, we will write its class as $[D]$; conversely, for a fixed class $[D]$, we let $D$ denote a generic divisor in the class. We will also often refer to divisor classes as \textit{charges}, especially when we write them in some basis, as the components of a class written in-basis are often readily interpreted as the $U(1)$ charges of an extended object wrapped on a representative of that class.

We will be especially interested in the divisor classes which are represented by actual subschemes of the variety, or non-negative linear combinations of prime divisors.
\begin{definition}\label{def:eff-comb}
    An \textit{effective} divisor $D$ in $X$ is an integral linear combination of prime divisors with non-negative coefficients. An effective divisor class is one that is represented by at least one effective divisor.
\end{definition}
That is, an effective divisor $D$ satisfies $\nu_B(D) \geq 0$ for all prime divisors $B$: we denote this by $D \geq 0$. Effective divisors and their classes are also known as \textit{holomorphic}~\cite{Witten:1996bn}. We will use these terms interchangeably. Divisors $D$ are in general associated to rank-one \textit{reflexive sheaves} $\mathcal{O}_X(D)$. If $D$ is \textit{Cartier}, meaning it is locally principal, then $\mathcal{O}_X(D)$ is a line bundle, and if $X$ is smooth then all divisors are Cartier. We will focus mostly on smooth Calabi--Yau threefolds, so divisors will most often be Cartier. However, we occasionally will have cause to consider singular varieties (e.g., the toric fourfolds of the Kreuzer--Skarke database), so we will at times work with non-Cartier divisors. This distinction will not play any significant role in our analysis, though, and we will proceed by referring to the line bundles associated to divisors with the knowledge that these are sometimes merely sheaves. Two divisors in the same class have isomorphic line bundles, meaning classes are associated to isomorphism classes of line bundles: we will let the line bundle of a divisor class just be that of some representative. We will also almost always work with \textit{$\mathbb{Q}$-factorial} varieties, where all Weil divisors $D$ are \textit{$\mathbb{Q}$-Cartier}: that is, $mD$ is Cartier for some positive integer $m$.

An alternative characterization of an effective divisor class is that the associated line bundle admits a global section. Global sections of $\mathcal{O}_X(D)$ are rational functions $\varphi$ on $X$ which induce effective divisors in the class $[D]$: i.e., they obey $\mathrm{div}(\varphi) + D \geq 0$. A local section of $\mathcal{O}_X(D)$ on some open set $U$ --- i.e., an element of $\mathcal{O}_X(D)(U)$ --- is a rational function obeying the same inequality for all prime divisors in $U$. The sheaf cohomology functors $H^i(-)$ are the right derived functors of the global sections functor, so global sections are elements of the zeroth sheaf cohomology group $H^0(X,\mathcal{O}_X(D)) := \mathcal{O}_X(D)(X)$. Thus, a divisor is effective if $h^0(X,\mathcal{O}_X(D)) := \mathrm{dim} \, H^0(X,\mathcal{O}_X(D)) \geq 1$. We will discuss the computation of sheaf cohomology of the sheaves associated to divisor classes on toric varieties in \cref{sec:direct_comp}, as this will be necessary to compute the global sections --- and thus, the effectiveness --- of divisor classes on toric hypersurface Calabi--Yau threefolds.

The group of divisor classes may have torsion, but we can quotient this out by tensoring with $\mathbb{R}$, yielding a vector space $H^2(X,\mathbb{R})$. Non-negative linear combinations of effective divisors are effective, so it is natural to consider the cone in $H^2(X,\mathbb{R})$ generated by effective divisors.
\begin{definition}
    The \textit{effective cone} of a variety $X$ is the cone generated by the classes of prime divisors of $X$ (i.e., all non-negative real linear combinations of prime divisor classes). The \textit{pseudoeffective cone} is the closure of the effective cone.
\end{definition}
The crucial distinction that we will focus on in this work is that between an effective divisor class and a divisor class contained in the effective cone, i.e., a divisor class $[D] \in H^2(X,\mathbb{R})$ such that 
\begin{align}
    [D] \in \mathcal{E}_X \cap H^2(X,\mathbb{Z})\,.
\end{align}
One upshot of this work is that one can construct smooth Calabi--Yau threefolds featuring divisor classes that live in $\mathcal{E}_X$ but are \textit{not effective}. We refer to such classes as \textit{holes}. For any $[D]$ contained in the effective cone, there is always a positive integer $m$ such that $m[D]$ is effective: merely express it as a $\mathbb{Q}$-linear combination of effective divisor classes (which can always be done) and choose $m$ such that denominators are cleared. However, $[D]$ itself need not be expressible as a non-negative $\mathbb{Z}$-linear combination of effective divisor classes: indeed, holes do not admit such a representation.

There are several other general properties of divisors that will be useful for us in this work. In particular, divisors endowed with different combinations of these properties will correspond to various subcones of the effective cone (some of which are summarized in \cref{fig:cone}). A divisor class is \textit{nef} if it has a non-negative intersection with all effective curves in $X$ and that nef divisor classes generate the closed \textit{nef cone}. More stringently, a divisor class is \textit{ample} if it has strictly positive intersection with all effective curves $C$ and all effective divisors $D_i$ in $X$
\begin{equation}
    D^3> 0\,,\qquad D^2\cdot D_i> 0\,,\qquad D\cdot C> 0\,,
\end{equation}
and such divisor classes generate the open \textit{ample} cone. In particular, the interior (closure) of the nef (ample) cone is the ample (nef) cone. Because the ampleness condition is preserved under positive rescaling of the class $[D]$, we avoid the subtlety we encountered with effectiveness:  if a divisor is nef or ample, so are all of its multiples and roots. 

The geometry of a Calabi--Yau threefold is determined in part by a choice of K\"ahler form, $J$, which fixes the volumes of all holomorphic cycles as follows.
\begin{align}
    \text{vol}(C) = \int_{[C]} J\,, \ \ \ \text{vol}(D) = \frac{1}{2}\int_{[D]} J \wedge J\,, \ \ \ \text{vol}(X) = \frac{1}{6}\int_{X} J \wedge J \wedge J\,.
    \label{eq:calibrated_vols}
\end{align}
Here, $C$ is any holomorphic curve in $X$, and $D$ is any holomorphic divisor in $X$. 
\begin{definition}
    The \textit{K\"ahler cone} $K_X$ is the cone of K\"ahler forms $J \in H^2(X,\mathbb{R})$ for which all curve volumes $\int_{C} J$ are strictly positive for all effective curves $C$.
\end{definition}
This coincides with the ample cone generated by ample divisor classes, and its closure (i.e., the K\"ahler cone along with its boundaries) is the nef cone. That is, the definition of the K\"ahler cone does not include the boundaries of the cone.

We now turn our attention to a few additional properties of divisors which are less frequently discussed, e.g., in the string compactification literature, but which will be important for us in this work.
\begin{definition}
    A divisor class $[D]$ on an $n$-dimensional variety $V$ is \textit{big} if the number of global sections of $m[D]$ as a function of $m$ are polynomial of degree $n$ for sufficiently large $m$.
\end{definition}
There are several other useful characterizations of big divisor classes. A nef divisor class $[D]$ is big if and only if the self-intersection number $[D]^n$ is positive; additionally, $[D]$ is big if and only if it falls on the interior of the pseudoeffective cone of $V$. The interior of the pseudoeffective cone is consequently sometimes called the big cone. See \S2.2 in~\cite{lazarsfeld2017positivity} for further discussion.

Nef cones encode birational geometry, a fact that will play a key role in this paper. Let us now briefly introduce some birational geometric vocabulary. Given two varieties $X, Y$ and a dense open subset $U$ of $X$, a morphism $U \to Y$ is called a \textit{rational map} $X \dashrightarrow Y$. A rational map is called a \textit{birational map} if it is an isomorphism between dense open subsets of $X$ and $Y$. Two varieties are \textit{birational} if they are related by a birational map. We say a birational map is \textit{small} if it is an isomorphism in codimension one: that is, the loci where the morphism and its inverse are not defined are both in codimension at least two. Small birational maps of smooth Calabi--Yau threefolds are flops \cite{katz1992gorenstein, Hori:2003ic}, for example. The locus where a morphism associated to a rational map is not defined can be thought of intuitively as being contracted or ``shrunk'' by that map. For example, the morphism associated to a flop is not defined on the curves that shrink to points and are then blown back up in the process of performing the flop. 
Given a birational morphism $f : X \dashrightarrow Y$ and a divisor $D$ on $X$, there is an induced \textit{pushforward} $f_* D = f(D)$ to a divisor on $Y$, and this operation lifts to divisor classes.

Birational varieties have a lot in common, especially those related by small birational maps. Crucially, varieties related by small birational maps have canonically isomorphic groups of divisor classes, with the number of global sections being preserved by this isomorphism (this fails for higher sheaf cohomology groups). Thus, for example, for a Calabi--Yau threefold $X$, we can freely evaluate the effectiveness of a divisor class $[D]$ in any other Calabi--Yau threefold $X'$ related to $X$ by flops. We often want to do this, because sometimes the computation of effectiveness is easier on a particular such $X'$, such as the one where $[D]$ is nef. We will discuss and exploit this at length in \cref{sec:holeprops}: in particular, we will employ the philosophy of the minimal model program from algebraic geometry and exploit that even birational maps which are not small can preserve properties like effectiveness in some cases (see \cref{th:MMP_global_sections}). The nef cones of varieties related by small birational maps sew together along their faces to form a larger cone. The union of all nef cones (i.e., closures of K\"ahler cones) of Calabi--Yau threefolds related by flops is the closure of the \textit{extended K\"ahler cone} $\mathcal{K}_X$.
In this way, the K\"ahler moduli parameterize the birational geometry of a Calabi--Yau threefold. The extended K\"ahler cone is a particular case of the algebraic geometric notion of the movable cone.
\begin{definition}
    An effective divisor class $[D]$ is \textit{movable} if its basepoint locus (i.e., the common zero-locus of all sections of the associated line bundle) has codimension at least two: that is, its basepoint locus is not a divisor. The cone generated by movable divisors is called the \textit{movable cone}.
\end{definition}
Equivalently, the movable cone consists of the divisor classes $[D]$ with no codimension-one \textit{stable base locus}, or base locus which persists for $m[D]$ for all $m \geq 0$. Intuitively, because the representatives of a movable divisor class are obligated to contain only codimension $\geq 2$ subvarieties, any fixed representative can be rather freely ``moved'' elsewhere in the variety by replacing it with a different representative of the same class: in particular, the deformations of any representative cover the variety.

As we've discussed, divisor classes in the effective cone need not be effective, while this problem doesn't occur for nefness and amplitude\footnote{The quality of being ample.} because these latter properties are invariant under overall rescaling while the former is sensitive to it. The movable property is also sensitive to scaling --- divisor classes in the movable cone need not be movable --- but movability, per se, actually will not play an important role for us: we care about divisor classes which belong to the movable cone because we will be focusing on toric varieties and Calabi--Yau threefolds in this work, for which we have \cref{th:movable_is_nef_somewhere}. This is the result that the cone generated by effective movable divisors is the union of the nef cones of all varieties related to $V$ by small birational maps (i.e., those that do not contract/shrink any divisors). Indeed, as we noted, the movable cone of a Calabi--Yau threefold is its extended K\"ahler cone $\mathcal{K}_X$. We will repeatedly exploit that divisor classes in the movable cone are nef in some birational model, but whether or not a class itself is literally movable will not be relevant. Consequently, we will perform an abuse of terminology and refer to any class in the movable cone as movable. 
\begin{figure}
    \centering
    \begin{tikzpicture}[scale=2.5, every node/.style={font=\small}]
      \def\extfactor{1.18} 
      \def\s{1}
    
      \definecolor{movable}{RGB}{67,129,193}
      \definecolor{nonmovable}{RGB}{232, 116, 97}
      \definecolor{nonbig}{RGB}{13, 27, 30}
      \definecolor{nonbigV}{RGB}{168, 180, 165}
    
      \foreach \x in {0,1,2}{
        \foreach \y in {0,1,2}{
          \coordinate (v\x\y) at ({\x*\s},{\y*\s});
        }
      }
    
      \coordinate (O) at (-0.9,-0.9);
      \fill (O) circle (1.5pt) node[below left] {$0$};
    
      \begin{scope}[on background layer]
        \foreach \x in {v00,v01,v02,v10,v20}{
          \draw[->, thick, black!70] (O) -- ($(O)!\extfactor!(\x)$);
        }
      \end{scope}
    
      \begin{scope}[on background layer]
        \foreach \x in {v12,v22,v21,v11}{
          \draw[->, dotted, thick, black!70] (O) -- ($(O)!\extfactor!(\x)$);
        }
      \end{scope}

      \fill[white] (v00) -- (v10) -- (v11) -- (v01) -- cycle; 
      \fill[nonmovable]        (v00) -- (v10) -- (v11) -- (v01) -- cycle; 
    
      \fill[white] (v10) -- (v20) -- (v21) -- (v11) -- cycle;
      \fill[movable] (v10) -- (v20) -- (v21) -- (v11) -- cycle;
    
      \fill[white] (v01) -- (v11) -- (v12) -- (v02) -- cycle;
      \fill[nonmovable] (v01) -- (v11) -- (v12) -- (v02) -- cycle;
      \fill[pattern=north east lines, pattern color=black, opacity=0.5]  (v01) -- (v11) -- (v12) -- (v02) -- cycle;  
    
      \fill[white] (v11) -- (v21) -- (v22) -- (v12) -- cycle;
      \fill[movable] (v11) -- (v21) -- (v22) -- (v12) -- cycle;
      \fill[pattern=north east lines, pattern color=black, opacity=0.5] (v11) -- (v21) -- (v22) -- (v12) -- cycle;  
    
      \draw[very thick, nonbig] (0,0) rectangle (2*\s,2*\s);
      \draw[very thick, dashed, nonbigV] (0,\s) rectangle (2*\s,2*\s);
    
      \begin{scope}[shift={(3,0.6)}]
        \draw[nonbig, line width=1.3pt] (0,1.375) -- (0.45,1.375);
        \node[right] at (0.55,1.375) {$\partial \mathcal{E}_X$ (not big on $X$)};
    
        \draw[nonbigV, dashed, line width=2pt] (0,0.775) -- (0.45,0.775);
        \node[right] at (0.55,0.775) {$\partial \mathcal{E}_V$ (not big on $V$)};
      
        \draw[fill=movable] (0,0) rectangle (0.45, 0.35);
        \node[right, text width=2.5cm] at (0.55, 0.175) {$\mathcal{K}_X$ (movable divisors)};
    
        \draw[fill=nonmovable] (0,-0.6) rectangle (0.45,-0.25);
        \node[right, text width=3cm] at (0.55,-0.425) {$\mathcal{E}_X \setminus \mathcal{K}_X$ (non-movable divisors)};
    
        \draw[fill=white, pattern=north east lines, pattern color=black, opacity=0.5] (0,-1.2) rectangle (0.45,-0.85);
        \node[right] at (0.55,-1.025) {$\mathcal{E}_V$ (effective on $V$)};

      \end{scope}
    
      \draw[->, thick, black!70] (O) -- ($(O)!\extfactor!(v01)$);
      \draw[->, thick, black!70] (O) -- ($(O)!\extfactor!(v00)$);
      \draw[->, thick, black!70] (O) -- ($(O)!\extfactor!(v10)$);

    \end{tikzpicture}
    \caption{3D cartoon of the effective cone $\mathcal{E}_X$ of a Calabi--Yau hypersurface $X$ in a toric variety $V$. In blue is the extended K\"ahler cone $\mathcal{K}_X$, containing the movable divisors on $X$; its complement, the non-movable divisors, is in red. Hashed is the subcone $\mathcal{E}_V$ of the $\mathcal{E}_X$ inherited from effective divisors of the ambient variety $V$. The extended K\"ahler cone itself decomposes as a union of K\"ahler/nef cones (see, e.g., Fig. 1 in~\cite{Gendler:2022ztv}). The (open) cone of big divisors on $X$ is the interior of $\mathcal{E}_X$, and the (open) cone of divisors inherited from big divisors on $V$ is the interior of $\mathcal{E}_V$. The boundaries of these open cones are the divisors that are effective but not big on $X$ and $V$, respectively.}
    \label{fig:cone}
\end{figure}
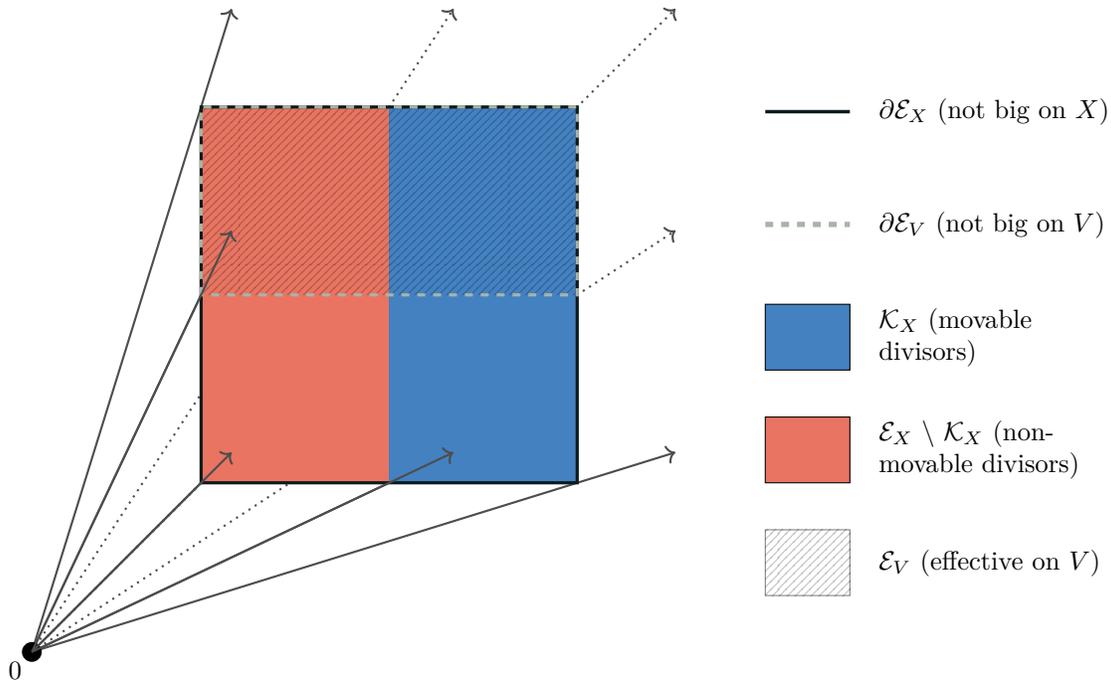

At this point, we have introduced many properties that divisors may have: to aid intuition and conceptualization, we present a cartoon of the effective cone $\mathcal{E}_X$ of a toric hypersurface Calabi--Yau threefold in \cref{fig:cone} to depict how divisors with different properties organize themselves. In particular, we show how the movable cone, or extended K\"ahler cone $\mathcal{K}_X$, is a subcone of $\mathcal{E}_X$, while the big divisors constitute the interior of the effective cone, with the boundary consisting of the non-big divisors. We will elaborate further on this figure upon explaining some specifics of toric hypersurface Calabi--Yau threefolds, but the features discussed here are universal to all Calabi--Yau threefolds.

\paragraph{Toric Varieties and Calabi--Yau Hypersurfaces.} We now turn to a review of the specific details of Calabi--Yau threefolds constructed as hypersurfaces in toric varieties, which will be our focus in this investigation. For a more general review of toric geometry we refer the reader to~\cite{cls, Hori:2003ic}: in this section, we will only recall basic results directly relevant for our purposes. 

We will restrict our attention to toric varieties arising from \textit{fans}, or sets of pointed cones closed under taking intersections and faces. These fans live in $N_\mathbb{R} := N \otimes \mathbb{R}$ for $N$ a lattice, and an important role is played by the minimal generators $u_i \in N$ which generate the one-dimensional semigroups given by intersecting individual one-cones (or rays) of the fan with $N$. For our purposes, it is important that to a fan $\Sigma$ and its associated toric variety $V$ we can associate a GLSM charge matrix $Q$~\cite{Witten:1993yc} and a Stanley--Reisner (SR) ideal $I$, which together determine $\Sigma$ and $V$. The rows of the GLSM charge matrix furnish a $\mathbb{Z}$-basis for the linear relations among the $u_i$, and we will define the SR ideal shortly.

The Batyrev construction of Calabi--Yau threefolds~\cite{batyrev1993dualpolyhedramirrorsymmetry} utilizes the fact that a generic anticanonical hypersurface in a toric variety with big nef anticanonical class and sufficiently mild singularities is a smooth Calabi--Yau threefold. In particular, the toric fans for such varieties are induced by \textit{fine}, \textit{regular}, and \textit{star} triangulations of four-dimensional reflexive polytopes $\Delta^\circ \subset N$. We denote the toric variety defined by such a fan as $V$, and the generic anticanonical hypersurface therein as $X$. In this work, we will restrict ourselves to \textit{favorable} reflexive polytopes: those for which every two-face with interior points is dual to a one-face with no interior points. This restriction results in all K\"ahler moduli of the Calabi--Yau hypersurface (i.e., all of $H^2(X, \mathbb{Z}) \cong H_4(X, \mathbb{Z})$) being inherited from the ambient variety. We do not expect significant qualitative differences in hypersurfaces originating from non-favorable polytopes, but leave this investigation for future work. 

An important thread throughout this work will be the distinction between objects that live in $V$ and those that live in $X$. We will denote objects that live in $V$ with hats: i.e., the effective cone of $V$ is $\hat{\mathcal{E}}_V$, and divisors in $V$ are written as $\hat{D}$, while the effective cone of $X$ is $\mathcal{E}_X$ and divisors in $X$ are written as $D$. 

In general, properties of the toric variety $V$ are much more straightforward to derive than the analogous properties of $X$. In particular, for favorable four-dimensional reflexive polytopes $\Delta^\circ \subset N$, the associated toric varieties contain $h^{1,1} + 4$ prime torus-invariant divisors, one for each non-zero point $u_i \in \Delta^\circ$ not interior to a facet (the divisors corresponding to points in facets do not intersect the hypersurface $X$). Here, $h^{1,1}$ is a Hodge number of the anticanonical Calabi--Yau hypersurface $X$. These divisors are known as the \textit{prime toric divisors}, $\hat{D}_i$, with $i \in \{1,\ldots, h^{1,1}+4\}$. Their classes generate all divisor classes on $V$ over $\mathbb{Z}$ and generate the effective cone $\hat{\mathcal{E}}_V$ over $\mathbb{R}_{\geq 0}$. In fact, they generate all effective divisor classes over $\mathbb{Z}_{\geq 0}$: i.e., every holomorphic divisor class $[\hat{D}]$ in $V$ can be written as
\begin{align}
    [\hat{D}] = \sum_i a_i [\hat{D}_i], \ \ \ a_i \in \mathbb{Z}_{\geq 0}.
\end{align}
so any divisor class in the effective cone that is not generated over $\mathbb{Z}_{\geq 0}$ by prime toric divisors is not effective on $V$ (i.e., is a hole on $V$). 

As an example of the computability of toric varieties, global sections $H^0(V,[\hat{D}])$ of a divisor class are conceptually straightforward to enumerate. All global sections of all line bundles are generated by monomials in the Cox ring, or the homogeneous coordinate ring, $S = \mathbb{C}[x_1, \dots, x_{h^{1,1}+4}]$: in particular, the prime toric divisor $\hat{D}_i$ is the zero-locus of the global section $x_i$ of $\mathcal{O}_X(\hat{D}_i)$. More generally, the torus-invariant divisor $\hat{D} = \sum_i a_i \hat{D}_i$ is given by the zero-locus of the monomial $\prod_i x_i^{a_i}$, which we denote by the shorthand $x^a$. The global sections of a class $[\hat{D}]$ are generated by the monomials associated to all torus-invariant divisors in that class. We will discuss this more systematically in \cref{sec:direct_comp} when we explain how to compute the higher line bundle cohomologies of a class $[\hat{D}]$, but for now we will note that the global sections of a class $[\hat{D}]$ with representative $\hat{D} = \sum_i a_i \hat{D}_i$ are in bijection with the lattice points of the \textit{Newton polytope} $\Delta_{\hat{D}}$ in the dual lattice $M = N^*$, defined as follows
\begin{equation}
    \label{eq:newton_poly}
    \Delta_{\hat{D}} = \{ m \in M \; | \; \langle m, u_i \rangle + a_i \geq 0 \}\,.
\end{equation}
In particular, $m \in \Delta_{\hat{D}}$ corresponds to the monomial $\prod_i x_i^{\langle m, u_i \rangle + a_i}$. The inequalities merely impose that all exponents of the monomials are non-negative, to avoid poles, and the addition of $a_i$ in each exponent ensures that the monomial belongs to the line bundle $\mathcal{O}_V(\hat{D})$ (monomials of the form $\prod_i x_i^{\langle m, u_i \rangle}$ belong to the trivial bundle).

With the Cox ring in hand, we can comment further on the GLSM charge matrix and Stanley--Reisner ideal. The GLSM charge matrix fixes a basis for the class group of divisors and maps a vector $a \in \mathbb{Z}^{h^{1,1} + 4}$, which we identify with the torus-invariant divisor $\hat{D} = \sum_i a_i \hat{D}_i$, to the class $[\hat{D}]$. The charge matrix is non-unique precisely due to our freedom to choose a basis for the group of divisor classes (indeed, a fixed choice of matrix determines a basis). In particular, the columns of the charge matrix are the classes of the associated prime toric divisors, so the effective cone is generated by these columns. For a fan $\Sigma$, the SR ideal $I$ is defined by
\begin{equation}
    I = \left( \prod_{\rho \in \tau} x^\rho \; \Big| \; \tau \notin \Sigma \right)\,.
\end{equation}
Toric varieties are quotients of a complex vector space, with some linear subspaces removed, by a subgroup of a torus (in particular, the dual group of the class group of divisors): the union of the vanishing loci of the generators of the SR ideal gives the linear subspaces removed before taking the quotient. The SR ideal will also play a central role in our computation of line bundle cohomology in \cref{sec:direct_comp}.

Holomorphic divisors $\hat{D} \subset V$ often descend to holomorphic divisors $D \subset X$ through intersection with $X$. That is
\begin{align}
    D = \hat{D} \cap X\,,
\end{align}
is often a divisor on $X$ (though edge cases can occur where this intersection is all of $X$). In particular, all effective divisor classes furnish at least one such representative, so any effective class $[\hat{D}]$ on $V$ induces an effective class $[D]$ on $X$. In particular, this means that $\mathcal{E}_X$ contains all non-negative linear combinations of the classes $[D_i] = [\hat{D}_i \cap X]$; as an abuse of notation, we will call both $D_i$ and $\hat{D}_i$ \textit{prime toric divisors}, with the hat clarifying which space we are referring to. The cone $\mathcal{E}_V$ generated by the $[D_i]$ furnishes a subcone of $\mathcal{E}_X$, which we denote $\mathcal{E}_V$ (with no hat). This is the part of $\mathcal{E}_X$ that is inherited from the ambient toric variety. Importantly, however, the set of holomorphic divisors on $X$ in general contains divisor classes which are \textit{not} effective on $V$. In~\cite{Demirtas:2018akl}, such divisors were referred to as \textit{autochthonous} divisors, and we utilize the same terminology here.
\begin{definition}
    Let $V$ be a suitable toric variety and $X$ a smooth Calabi--Yau hypersurface therein. An \textit{autochthonous} divisor class on $X$ is a class $[D]$ such that
    \begin{align}
        h^{0}(V, \mathcal{O}_V(\hat{D})) = 0 \ \ \text{but} \ \ h^{0}(X, \mathcal{O}_X(D)) >0\,.
    \end{align}
\end{definition}
Therefore, in general we have that $\mathcal{E}_V \subset \mathcal{E}_X$. Note further that in principle, there could be autochthonous divisors in $\mathcal{E}_V$: it is possible that there exist divisor classes in the effective cone of $V$ that are \textit{not} effective on $V$ but are indeed holomorphic in $X$. However, as we will discuss in \cref{sec:ksholes}, we do not find any examples of this phenomenon.

This discussion is illustrated through our cartoon of the effective cone $\mathcal{E}_X$ of a toric hypersurface Calabi--Yau threefold in \cref{fig:cone}. In particular, we see there how a subcone $\mathcal{E}_V$ of $\mathcal{E}_X$ is inherited from the ambient variety: its complement $\mathcal{E}_X \setminus \mathcal{E}_V$ is exactly the set of autochthonous divisors. 

A useful notion that we will refer to repeatedly in this work is that of the \textit{Hilbert basis} of a cone:
\begin{definition}
    The \textit{Hilbert basis} $\mathcal{H}(\mathcal{C})$ of a cone $\mathcal{C} \subset \mathbb{R}^n$ is the minimal set of integral vectors $\{\mathbf{v}_i\}$ needed to generate every integral point in $\mathcal{C}$ over $\mathbb{Z}_{\geq 0}$. That is, $\mathcal{H}(\mathcal{C})$ is the minimal set of $\{\mathbf{v}_i\}$ such that
    \begin{align}
        \mathbf{p} = \sum_i c_i \mathbf{v}_i\,, \ \ c_i \in \mathbb{Z}_{\geq 0}.
    \end{align}
    for any $\mathbf{p} \in \mathcal{C} \, \cap \, \mathbb{Z}^n$.
\end{definition}
Note that the Hilbert basis of the effective cone of a toric variety is not the same as the set of its prime toric divisors. In particular, neither is a subset of the other in general. In the context of Calabi--Yau hypersurfaces in toric varieties, we refer to elements of the Hilbert basis of $\mathcal{E}_V$ that are not prime toric divisors as \textit{non-trivial Hilbert basis elements}. A main purpose of this work is to argue that non-trivial Hilbert basis elements are holes. Indeed, in \cref{sec:ksholes} we search through all non-trivial Hilbert basis elements of toric hypersurface Calabi--Yau threefolds from the Kreuzer--Skarke database with $h^{1,1} \leq 6$ (and randomly sample many polytopes with larger $h^{1,1}$), finding none of them to be effective (i.e., autochthonous). In \cref{sec:holeprops}, we will prove that this must be true for non-trivial Hilbert basis elements on the strict interior of $\mathcal{E}_V$ (\cref{prop:toric_big_movable}), leaving the more general statement as a conjecture in \cref{conj:non_trivial_hb_are_holes}.

We conclude by noting that the K\"aher cone of the ambient toric variety provides some information that descends to the K\"ahler cone of the hypersurface Calabi--Yau. In particular, letting $K_V$ denote the cone generated by classes $[D]$ induced by ambient classes $[\hat{D}]$ that are ample,
\begin{align}
    K_V \subset K_X.
\end{align}
Intuitively, all K\"ahler forms assigning positive volumes on the ambient variety will continue to do so upon restriction to any hypersurface (including $X$) but need not exhaust the K\"ahler forms assigning positive volume on any given hypersurface. 

\paragraph{Hypersurface Line Bundle Cohomology.} Let us recall how the line bundle cohomology of divisors on $X$ which are inherited from $V$ can be computed from the line bundle cohomology of $V$. We presently take the computation of toric line bundle cohomology for granted: this will be discussed in detail in \cref{sec:direct_comp}.

Fix a divisor class $[\hat{D}]$ in $V$ descending to $[D]$ on $X$. We will focus on a Calabi--Yau hypersurface given by the zero locus of a global section $G$ of the anticanonical class $-\hat{K}$ of $V$, but what follows holds just as well for any other hypersurface. Letting $\iota$ denote the inclusion of $X$ into $V$, we have the following short exact sequence of sheaves, known as the Koszul sequence.
\begin{equation}
    0 \to \mathcal{O}_V(\hat{D}+\hat{K}) \xrightarrow{F} \mathcal{O}_V(\hat{D}) \to \iota_*\mathcal{O}_X(D) \to 0 \, .
\end{equation}
By applying the right derived functor $H^{\bullet}(\text{--})$ to this Koszul sequence, we arrive at a long exact sequence
\begin{align}
    \begin{split}
        0 & \to H^0(V,\mathcal{O}_V(\hat{D} + \hat{K})) \xrightarrow{F} H^0(V, \mathcal{O}_V(\hat{D})) \xrightarrow{\alpha} H^0(X,\mathcal{O}_X(D)) \\
        & \xrightarrow{\beta} H^1(V,\mathcal{O}_V(\hat{D} + \hat{K})) \xrightarrow{\tilde{F}} H^1(V,\mathcal{O}_V(\hat{D})) \to {\dots} \,.
    \end{split}
\end{align}
Here, $\alpha$ is the restriction of sections to $X$, $\beta$ is the connecting homomorphism arising from application of the snake lemma~\cite{weibel1994introduction}, and $F, \tilde{F}$ are maps induced by multiplication by the section $G$ defining $X$. From the long exact sequence, $H^0(X,\mathcal{O}_X(D))$ decomposes as 
\begin{equation}
    \label{eq:coker_ker}
    H^0(X,\mathcal{O}_X(D)) = \mathrm{coker} \, F \oplus \mathrm{ker} \, \Tilde{F} \, .
\end{equation}
Each summand admits an interpretation. The cokernel is the restriction of global sections $H^0(V, \mathcal{O}_V(\hat{D}))$ to $H^0(X,\mathcal{O}_X(D))$ modulo sections with an overall factor of $G$, parameterized as $F \cdot H^0(V,\mathcal{O}_V(\hat{D} + \hat{K}))$. That is, if a global section of $[\hat{D}]$ contains an overall factor of $G$, it vanishes identically upon restricting to $D$. The kernel describes global sections of $[D]$ which do not lift to $[\hat{D}]$.

In particular, we have that
\begin{equation}
    \label{eq:CY_global_section}
    h^0(X, \mathcal{O}_X(D)) = h^0(V, \mathcal{O}_V(\hat{D})) - h^0(V, \mathcal{O}_V(\hat{D}+\hat{K})) + \mathrm{dim}(\mathrm{ker}(\tilde{F})) \, .
\end{equation}
In this way, $h^0(X, \mathcal{O}_X(D))$ is fully determined by a toric calculation: toric line bundle cohomology groups and the rank of a map between such groups. We will review this calculation in \cref{sec:direct_comp}.

We now make a few useful observations. First, as noted earlier, holomorphic divisors on $V$ will always be holomorphic on $X$.\footnote{In our setting, the anticanonical class always has more than one section, so $h^0(V, \mathcal{O}_V(\hat{D})) > h^0(V, \mathcal{O}_V(\hat{D}+\hat{K}))$.} Second, we can see explicitly how non-holomorphic divisor classes $[\hat{D}]$ on $V$ may well still be holomorphic on $X$, resulting in autochthonous divisors. There will be no $\mathrm{coker} \, F$ contribution (because $H^0(V,\mathcal{O}_V(\hat{D} + \hat{K}))$ injects into $H^0(V,\mathcal{O}_V(\hat{D}))$, so the former vanishes when the latter does), but $\mathrm{ker} \, \Tilde{F}$ can still contribute. In this way, while holes on $V$ are natural candidates for holes on $X$, they may in fact become holomorphic upon restricting to $X$. Third, continuing in this direction, a sufficient (but unnecessary) condition for an inherited $D$ to be non-holomorphic is that $[\hat{D}]$ is non-holomorphic on $V$ and that $h^1(V, \mathcal{O}_V(\hat{D} + K)) = 0$, as this latter condition implies the vanishing of $\mathrm{ker} \, \Tilde{F}$. We will exploit this sufficient condition in the ensuing sections \cref{sec:holeprops} to prove additional results about holes for toric hypersurfaces, and we will further find in \cref{sec:ksholes} that this sufficient condition is realized by all non-trivial Hilbert basis elements we study. 

\subsection{Physical Interpretation}

In this section, we will provide further physical context and motivation for the study of effective divisors. We will begin with discussing their role in compactifications of M-theory, and then describe them in the context of Calabi--Yau orientifold compactifications of type IIB string theory.

Consider M-theory compactified on a Calabi--Yau threefold, $X$. The resulting effective theory is a five-dimensional theory endowed with $\mathcal{N}=1$ supersymmetry. The low-energy theory contains $h^{1,1}(X)-1$ vector multiplets with scalar components $\phi^a$, and the bosonic part of the action for the vector multiplets takes the form
\begin{equation}
    \begin{aligned}
        S &= \frac{2\pi}{\ell_5^3} \int \text{d}^5 x \sqrt{-g} \left(R - \frac{1}{2} f_{ij}  \partial_a t^i \partial_b t^j \partial_\mu \phi^a \partial^\mu \phi^b \right) \\
        &\qquad\qquad -\frac{1}{4\pi \ell_5} \int f_{ij}(\phi) F^i \wedge \star F^j-\frac{1}{24 \pi^2} \int \kappa_{ijk} A^i \wedge F^j \wedge F^k\,.
    \end{aligned}
\end{equation}
Here, $\ell_5$ is the five-dimensional Planck length, $R$ is the Ricci scalar of the metric $g$, $\kappa_{ijk}:=\int_X[D_i]\wedge [D_j]\wedge [D_k]$ are the triple-intersection numbers of $X$, and the $A^i$ are the $h^{1,1}(X)$ gauge potentials with field strengths $F^i$. Furthermore, the $\phi^a$ parameterize the hypersurface defined by
\begin{align}
   \mathcal{F} := \frac{1}{6} \kappa_{ijk} t^i t^j t^k =1\,,
\end{align}
where the $t^i$ are the K\"ahler parameters of $X$, and
\begin{align}
    f_{ij} = \partial_i \mathcal{F} \partial_j \mathcal{F} - \partial_i \partial_j \mathcal{F}\,.
\end{align}
The BPS spectrum is given by M2-branes and M5-branes wrapped on effective 2-cycles and 4-cycles, respectively. Focusing on string states constructed from M5-branes, when the 4-cycle $D$ is either ample or nef, then the resulting string in the 5d $\cN=1$ theory is either a black string~\cite{Maldacena:1997de} or a supergravity string~\cite{Katz:2020ewz}. Additionally, in the infinite-distance limit of the K\"ahler moduli space where we can encounter a \textit{fundamental} string becoming tensionless, the string is constructed from an M5-brane wrapping a 4-cycle that is non-big~\cite{Lee:2019wij}. Nevertheless, as all these string states are supersymmetric, their tension formula is exact and takes on the following form
\begin{align}
    T_{\text{BPS}} \propto \text{vol}(D)\,,
\end{align}
throughout the K\"ahler moduli space.

On the other hand, non-BPS states come from these branes wrapped on non-holomorphic cycles. Again, focusing on string states, the tensions are also given, to leading order, by the volumes of 4-cycles, but can be corrected by quantum effects. The tension of a string state coming from an M5-brane wrapped on a non-holomorphic irreducible cycle $D'$ is then
\begin{align}
    T_{\text{non-BPS}} \propto \text{vol}(D') + \text{quantum corrections}\,.
\end{align}
When the non-holomorphic cycle is reducible, the tension of the resulting string can additionally receive classical corrections, such as those arising from self-interactions.

Thus, characterizing the spectrum of BPS and non-BPS states in compactifications of M-theory on a Calabi--Yau threefold requires knowledge of
\begin{enumerate}
    \item which sites in $H^2(X, \mathbb{Z})$ are holomorphic and which are non-holomorphic, and
    \item the volumes of all cycles.
\end{enumerate}

Apart from determining the spectrum of solitonic objects in M-theory compactifications, the information above determines which instantons can give rise to corrections of $\mathcal{N}=1$ theories coming from type IIB compactifications to four dimensions. Consider type IIB string theory compactified on a Calabi--Yau (orientifold) $X$. As stated in \cref{sec:review-geometry}, we will assume that the orientifold projects out no K\"ahler moduli, so that the four-dimensional $\mathcal{N}=1$ theory contains $h^{1,1}(X)$ massless K\"ahler moduli. 

For simplicity, we assume that the complex structure moduli and axio-dilaton are stabilized at some high scale. Then, the low-energy theory is characterized by an F-term scalar potential that takes the following form.
\begin{align}
    V_F  = e^K (K^{i\bar{j}} D_i W D_{\bar{j}} \overline{W} - 3 |W|^2).
\end{align}
Here, the indices run over the $h^{1,1}(X)$ K\"ahler moduli, the K\"ahler potential $K$ is given at tree-level by
\begin{align}
    K = -2 \log(\mathcal{V})\,,
\end{align}
where $\mathcal{V}$ is the overall volume of $X$, and the superpotential $W$ takes the form
\begin{align}
    W = W_0 + \sum_{I} A_I e^{-2\pi q^I_i T^i}\,.
    \label{eq:superpotential}
\end{align}
Here, $W_0$ is a (generally complex) constant, and the $A_I$ are one-loop determinants, which can be zero if the associated instanton has more than the two universal zero modes. The $T^i$ are the K\"ahler moduli
\begin{align}
    T^i = \text{vol}(D^i) + i \int_{D^i} C_4 =: \tau^i + i \theta^i\,,
\end{align}
with $D^i$ a basis of holomorphic divisors, and $C_4$ the Ramond-Ramond four-form potential in ten dimensions, with $\tau^i$ and $\theta^i$ being four-dimensional saxion and axion fields, respectively. In \eqref{eq:superpotential}, the $q^I_i$ are instanton charges: importantly, these charges range only over sites in the $H^2(X, \mathbb{Z})$ lattice that correspond to effective divisors. This is necessary so that $W$ remains a holomorphic function of the moduli. The quantity $q^I_i T^i$ measures the volume of a holomorphic divisor with charge $q^I_i$ in $H^2(X, \mathbb{Z})$. Note that for every such holomorphic divisor, the volume is given by \eqref{eq:calibrated_vols}. 

By specifying a basis, we can write the volumes explicitly in terms of the triple intersection numbers $\kappa_{ijk}$ of $X$:
\begin{align}
    \text{vol}(D^i) = \frac{1}{2} \kappa_{ijk} t^j t^k\,,
\end{align}
where the $t^i$ specify a point in $K_X$.

\begin{figure}
    \centering
    \begin{subfigure}{.48\textwidth}
    \includegraphics[width=\linewidth]{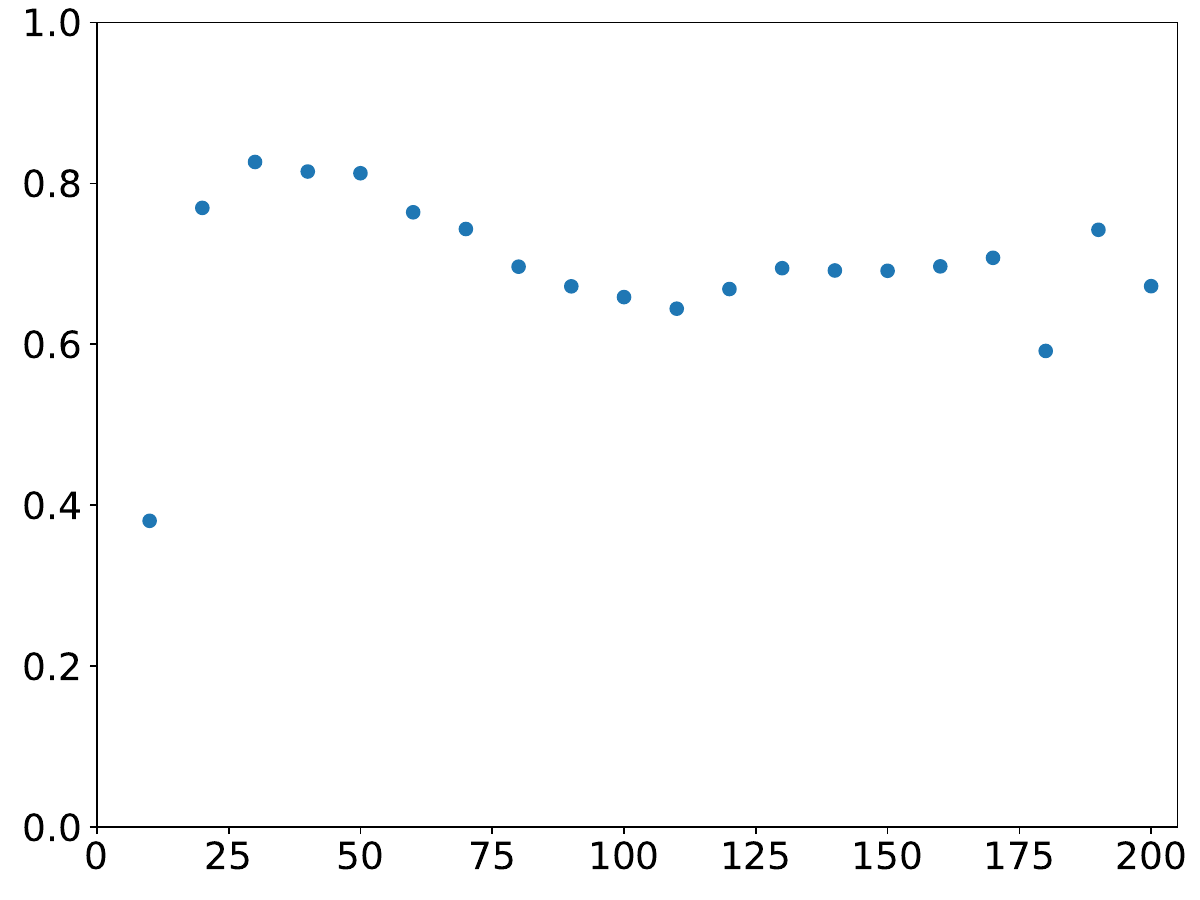}
    \begin{picture}(-210,0)\vspace*{-1.2cm}
    \put(-100,-10){\footnotesize $h^{1,1}$}
    \put(-225,65){\rotatebox{90}{\footnotesize Fraction}}
    \end{picture}\vspace{0.5cm}
    \caption{Fraction of birational classes of Calabi--Yau threefolds with non-trivial Hilbert basis elements}
    \end{subfigure}
    \hfill
    \begin{subfigure}{.48\textwidth}
    \includegraphics[width=\linewidth]{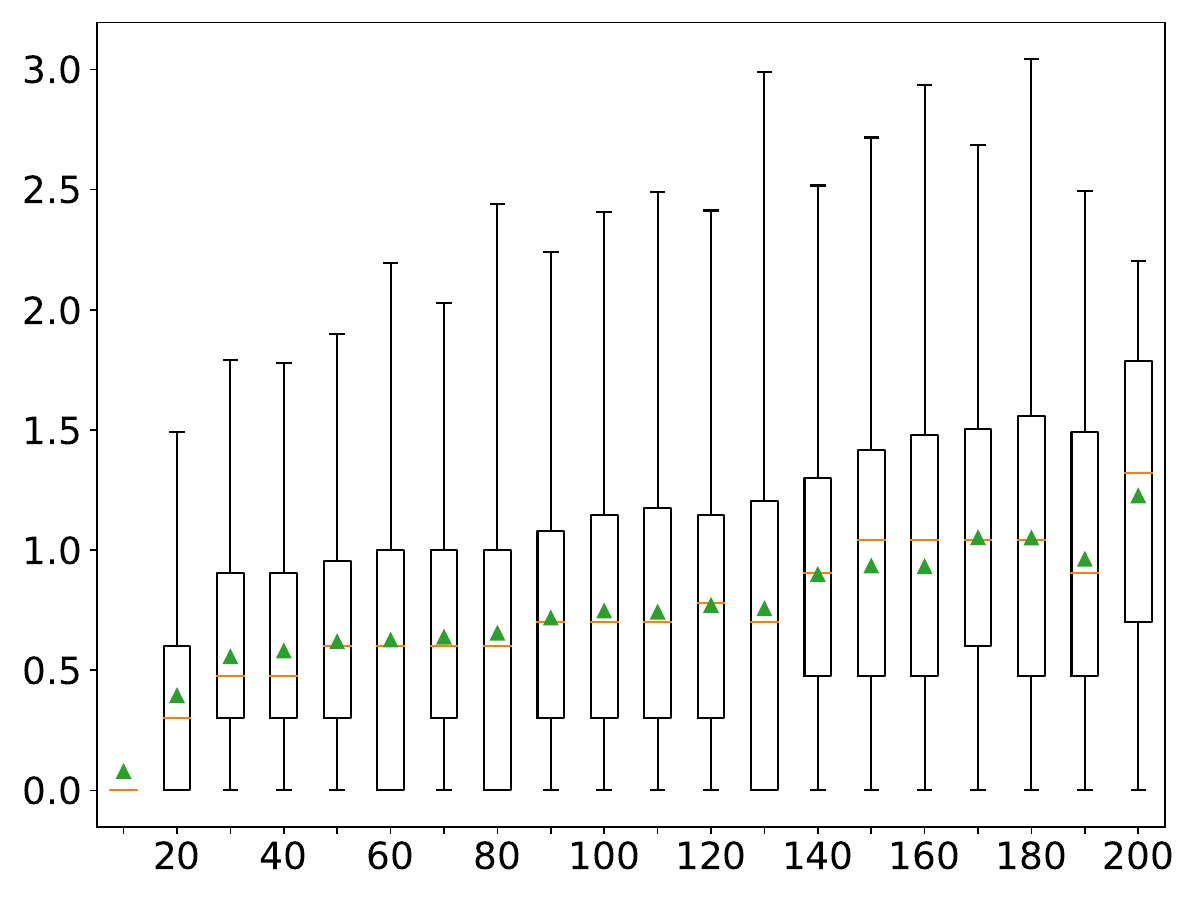}
    \begin{picture}(-200,0)\vspace*{-1.2cm}
    \put(-100,-10){\footnotesize $h^{1,1}$}
    \put(-225,10){\rotatebox{90}{\footnotesize $\log_{10}(\#\text{ of non-trivial Hilbert basis})$}}
    \end{picture}\vspace{0.5cm}
    \caption{Distribution of number of non-trivial Hilbert basis elements}
    \end{subfigure}\vspace{0.cm}
    \caption{A scan through the Kreuzer--Skarke database for Calabi--Yau threefolds with non-trivial Hilbert basis elements up to $h^{1,1}=200$. Here, for each given $h^{1,1}$, we randomly sample $\sim 10^4$ polytopes within the Kreuzer--Skarke database using CYTools~\cite{Demirtas:2022hqf}.
    In the right figure, the box indicates 50\% of the number of non-trivial Hilbert basis elements, the whiskers indicate the minimum and maximum value of the number of non-trivial Hilbert basis elements within $3/2$ of the interquartile range, the orange horizontal line indicates the median, and the green triangle indicates the mean.}
    \label{fig:additional elements}
\end{figure}
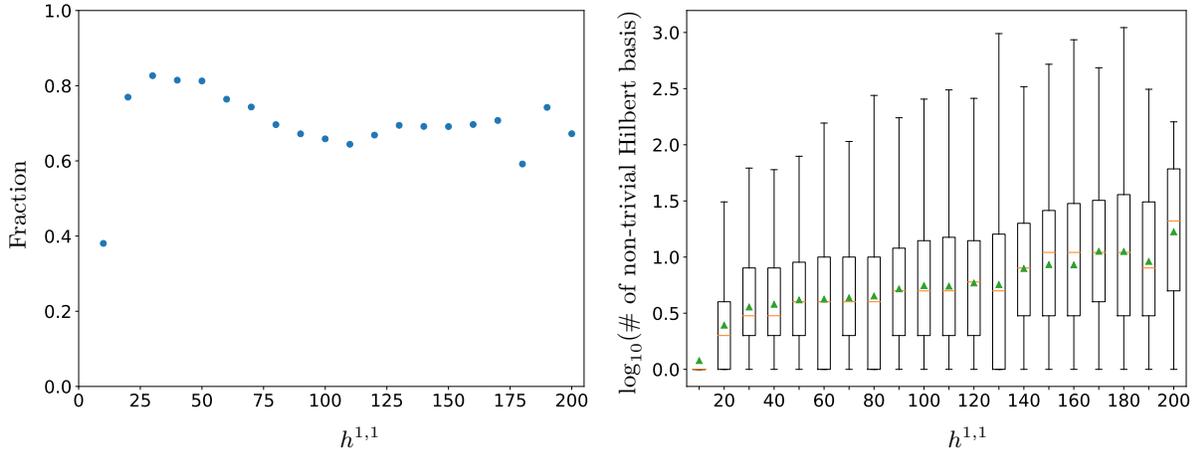
An important question in the study of $\mathcal{N}=1$ compactifications of type IIB string theory is that of which corrections to \eqref{eq:superpotential} are the leading ones. A standard assumption for phenomenological studies with toric hypersurfaces~\cite{Demirtas:2018akl, Demirtas:2021gsq, Cheng:2025ggf, Yin:2025amn, Gendler:2023kjt, Gendler:2023hwg, Gendler:2024adn, McAllister:2024lnt, Demirtas:2021ote, Demirtas:2021nlu, Mehta:2021pwf} is that the contributions from Euclidean D3-branes wrapping prime toric divisors are the most dominant contributions. However, if any non-trivial Hilbert basis elements are effective, then it is conceivable that this assumption is misguided: effective non-trivial Hilbert basis elements could easily constitute the most relevant non-perturbative effects (while divisor classes in the effective cone outside the Hilbert basis always have calibrated volumes larger than some elements of the Hilbert basis, and consequently yield subleading physical effects). The danger that this possibility presents is illustrated in \cref{fig:additional elements}, where we exhibit the statistics of non-trivial Hilbert basis elements in the Kreuzer--Skarke database. Here we see two facts: first, most Calabi--Yau threefolds contain non-trivial Hilbert basis elements, and second, the number of non-trivial Hilbert basis elements can be quite high, even $\mathcal{O}(1000)$. We will see, however, that the assumption that prime toric divisors constitute the leading contributions may indeed be justified: in \cref{sec:holeprops}, we will prove that any such non-trivial Hilbert basis elements that are interior to $\mathcal{E}_V$ are not effective (\cref{prop:big}), and in \cref{sec:ksholes} we will give empirical evidence that the same is true for non-trivial Hilbert basis elements on the boundary of $\mathcal{E}_V$, motivating \cref{conj:non_trivial_hb_are_holes}.

\section{Properties of Holes} \label{sec:holeprops}

In this section, we discuss some general properties of holes on Calabi--Yau threefolds. Our discussion will begin by considering general Calabi--Yau threefolds, though we will eventually restrict to our primary focus in \cref{sec:toric_hyper_holes}: namely, hypersurfaces in toric varieties, and their non-trivial Hilbert basis elements. This section has some mathematically technical aspects, so a reader willing to accept our main results --- summarized below --- will not need most of what is discussed for the remaining paper. However, in this section we will present some useful perspectives which will facilitate the interpretation of our examples in \cref{sec:ksholes}. 

For $X$ a smooth Calabi--Yau threefold, we show the following.
\begin{itemize}
    \item (\cref{th:big_movable_CY}) Holes cannot be big and movable, meaning they cannot belong to $\mathcal{K}_X \setminus \partial \mathcal{E}_X$ (i.e., the extended K\"ahler cone aside from its intersection with the boundary of the effective cone).
    \item (\cref{cor:semigroup}) Non-movable holes $[D]$ are not isolated, but rather come in semigroups of the form $[D] + \sum_i a_i [E_i]$ for $E_i$ the prime components of the stable base locus of $[D]$.\footnote{Note that the semigroup operation is $([D] + \sum_i a_i [E_i]) + ([D] + \sum_j b_j [E_j]) = [D] + \sum_i (a_i + b_i) [E_i]$.} 
    \item (\cref{th:CY_MMP_holes}) Using the minimal model program, holes $[D]$ correspond to nef holes $[f_* D]$ for $f : X \dashrightarrow Y$ a birational morphism to a (singular) $[D]$-minimal model $Y$.
    \item (\cref{prop:torsion}) In particular, if a multiple of $[D]$ is a non-negative linear combination of prime exceptional divisors of $f$, then $Y$ has torsion in its class group, with $[f_* D]$ being a torsion class.
    \item (\cref{prop:big}) If $X$ is a hypersurface in a toric variety $V$, then non-trivial Hilbert basis elements descending from big divisors on $V$ (i.e., those on the interior of $\mathcal{E}_V$) must be holes.
\end{itemize}
The ideas employed in this section are closely related to those of \cite{constantin_inprep}, written concurrently to this paper, which studies the structure of global sections of divisors more broadly.

\subsection{Holes and Movability}

We begin with the following result, which is straightforward to show but is nevertheless a nice starting point for our conversation.
\begin{prop}
    For $[D]$ a big divisor class in the movable cone of a Calabi--Yau threefold $X$, $[D]$ is effective. \label{th:big_movable_CY}
\end{prop}
To prove the above proposition, we will first need to establish the following result for movable divisors in Calabi--Yau threefolds. This following theorem will also be useful in showing the results for toric varieties in \cref{prop:toric_klt}.
\begin{theorem}
    If $[D]$ is a movable divisor in the effective cone of $X$ a toric variety or smooth Calabi--Yau threefold, then there exists a $[D]$-minimal model related to $X$ by small birational maps. In particular, the cone generated by effective movable divisor classes on $X$ decomposes as the union of nef cones of varieties related to $X$ by small birational maps.
    \label{th:movable_is_nef_somewhere}
\end{theorem}
\begin{proof}
    For toric varieties, this follows from Th. 15.1.10 in~\cite{cls} (in particular, see the discussion following the proof of this result), and also holds broadly for Mori dream spaces~\cite{hu2000mori}. For Calabi--Yau threefolds, this is Th. 2.3 in~\cite{kawamataCone}.
\end{proof}
That is, movable divisor classes in the effective cone can be rendered nef by a series of small birational maps (e.g., on a Calabi--Yau threefolds, by flops). We are now prepared to prove the original result.
\begin{proof}[Proof of \cref{th:big_movable_CY}]
    By \cref{th:movable_is_nef_somewhere}, because we assume $[D]$ is movable, we can render it nef by flopping $X$ some number of times, yielding $X'$. Global sections (and thus effectiveness) as well as bigness are preserved under flops. The holomorphic Euler characteristic of a divisor $[D]$ inside a Calabi--Yau threefold $X'$ is 
    \begin{equation}
        \chi(\mathcal{O}_{X'}(D)) = \sum_{i=0}^3 (-1)^i h^i(X', \mathcal{O}_{X'}(D)) \, .
    \end{equation}
    By employing the Hirzebruch--Riemann--Roch theorem~\cite{Hirzebruch1966}, we can write that 
    \begin{equation}
        \chi(\mathcal{O}_{X'}(D)) = \frac{1}{6}D^3 + \frac{1}{12}c_2(X') \cdot D \, .
    \end{equation}
    The first term is non-zero because $[D]$ is big and nef, and the second term is non-negative by Miyaoka's theorem~\cite{miyaoka1987chern}. Finally, by the Kawamata--Viehweg vanishing theorem \cite{kawamata1982generalization, viehweg1982vanishing}, $h^i(\mathcal{O}_{X'}(D)) = 0$ for $i > 0$ because $[D] - [K] = [D]$ is big and nef on $X'$ (the canonical class $[K]$ of $X$ being trivial). All together, we have
    \begin{equation}
        h^0(\mathcal{O}_X(D)) = h^0(\mathcal{O}_{X'}(D)) = \chi(\mathcal{O}_{X'}(D)) > 0 \, ,
    \end{equation}
    concluding the proof.
\end{proof}
For the case of (very) ample divisors in Calabi--Yau threefolds, from the attractor mechanism of black strings in M-theory, these divisors are naturally expected to be effective~\cite{Maldacena:1997de}.

What about non-big movable divisors (i.e., divisors on the boundary of both the effective and movable cones)? We are free to iteratively flop $X$ to $X'$ such the non-big movable divisor $D$ in question is nef, but Kawamata--Viehweg is unavailable. 
Previously, the following physics-informed conjecture has been posed.\footnote{Because nef divisors are semiample on smooth Calabi--Yau threefolds (a corollary of log abundance for threefolds \cite{keel1994log}), nef and non-big divisor classes $[D]$ define fibrations $f : X \to B$ with $\dim X > \dim B$ such that $D = f^* A$ for $A$ ample on $B$ and $h^0(X, \mathcal{O}_X(D)) = h^0(B, \mathcal{O}_B(Y))$ (see, e.g., Lemma 2.1.13 in \cite{lazarsfeld2017positivity} and the subsequent discussion of semiample fibrations). Thus, a counterexample of \cref{conj:nonbig_movable_CY} exists if and only if there is a smooth fibered Calabi--Yau threefold whose base features an ample, non-effective divisor. Note that the base would not be Calabi--Yau and need not even be smooth: indeed, two of the smooth Calabi--Yau threefolds in the Krezuer--Skarke database with $h^{1,1} = 2$ are elliptic fibrations over the singular weighted projective space $\mathbb{P}_{112}$ (these arise from the polytopes with IDs $22$ and $33$, respectively). In fact, it is not difficult to find varieties with ample, non-effective divisors: we showcase a simple weighted projective space example in \cref{sec:minimal_model}, which subsequently appears in \cref{sec:h11=4}. One could imagine attempting to engineer a counterexample to \cref{conj:nonbig_movable_CY} by constructing fibrations over such bases; the challenge is achieving a smooth fibration.}
\begin{conjecture}[\S3.1 in \cite{Katz:2020ewz}]
    Divisor classes contained in the cone of movable divisors on smooth Calabi--Yau threefolds are effective. 
    \label{conj:nonbig_movable_CY}
\end{conjecture}
For discussion related to this conjecture in the math literature, see~\cite{lazarsfeld2017positivity, Oguiso}. 
Here, we will review the interesting case where $[D]^2 = 0$ and $c_2(X) \cdot [D] = 0$ for which effectiveness is currently only motivated by physical arguments~\cite{Katz:2020ewz}. To start, let us consider M5-branes wrapping these divisors in smooth Calabi--Yau threefolds. The resulting object in the 5d $\cN=1$ theory is a magnetic monopole string preserving at least 4 chiral supercharges on its 2d string worldsheet. By computing the gauge anomaly inflow along the 2d string worldsheet, we can deduce the right-moving and left-moving central charge of the 2d CFT~\cite{Maldacena:1997de} 
\begin{equation}
    c_R=6D^3+\frac12 c_2(X) \cdot D - c_R^{\rm com}\,,\qquad c_L=6D^3+c_2(X) \cdot D-c_L^{\rm com}\,,
\end{equation}
where $c_{R,L}^{\rm com}$ are integral constants counting the center-of-mass modes in each right- and left-moving sector. We can also compute the levels of the Kac--Moody current algebras of the symmetries in the 2d CFT from the following data
\begin{equation}
    k_R=c_2(X) \cdot D\,,\qquad k_{ij} =D \cdot D_i \cdot D_j\,,
\end{equation}
where $k_R$ is the level associated to the $SU(2)_R$ symmetry while the levels associated to e.g., the remaining abelian global symmetries can be found upon diagonalizing $k_{ij}$. For the former case where $D^3=0$ and $D^2=0$,~\cite{Katz:2020ewz} argued that when combining the results of~\cite{Oguiso} along with the completeness of spectrum hypothesis~\cite{Polchinski:2003bq,Banks:2010zn}, $D$ is expected to be effective. For the latter case when $D^3=0$ and $c_2 \cdot D=0$, we have $c_R=c_L=0$ and $k_R=0$. Then, this implies that the 2d CFT has no gravitational anomaly, no right-moving current, and no $SU(2)_R$ anomaly. However, these properties could only lead to supersymmetry enhancements, due to the absence of chiral anomalies. Namely, the resulting 2d CFT may either be enhanced to $\cN=(4,4)$ or even $\cN=(8,8)$.
Hence, from the worldsheet perspective, branes wrapped on these divisors can only be more supersymmetric, and are thus also expected to be effective. 

\subsection{Holes Have Friends}
\label{sec:friends}

In this section, we show that non-movable holes are not isolated, but rather come in semigroups. We begin with the following technical result, which characterizes a situation in which local sections lift to global sections.
\begin{lem}
    \label{lem:lift_to_global}
    Let $X$ be a normal projective variety with a Weil divisor $D$ satisfying $[D] \in \mathcal{E}_X$. If $U \subset X$ is an open subset such that the prime codimension-one components of $X \setminus U$ are contained in the stable base locus of $[D]$, then $H^0(X, \mathcal{O}_X(D)) \cong \mathcal{O}_X(D)(U)$.
\end{lem}
\begin{proof}
    Let $E_1, \dots, E_n$ denote the prime codimension-one components of $X \setminus U$. Certainly $H^0(X, \mathcal{O}_X(D)) \subset \mathcal{O}_X(D)(U)$, as any rational function $f$ obeying $\mathrm{div}(f) + D \geq 0$ on all prime divisors will do so for all prime divisors on $U$. To show the other inclusion, we will assume by way of contradiction that $\psi \in \mathcal{O}_X(D)(U) \setminus H^0(X, \mathcal{O}_X(D))$ and use $\psi$ to construct a global section of some multiple of $[D]$ which does not contain some $E_j$ in its base locus. By construction, the divisor $D_\psi := \mathrm{div}(\psi) + D$ has a pole, but it can only be on one of the $E_i$.

    Choose $m > 0$ such that $m[D]$ is effective, and let $D_\varphi := \mathrm{div}(\varphi) + mD \geq 0$ (i.e., $\varphi \in H^0(X, \mathcal{O}_X(mD))$). It will suffice to choose $k,\ell \in \mathbb{Z}_+$ such that the divisor 
    \begin{equation}
        A = k(mD_\psi) + \ell(D_\varphi) = \mathrm{div}(\psi^{mk}\varphi^{\ell}) + m(k+\ell)D
    \end{equation}
    with class $m(k+\ell)[D]$ is effective and does not contain some $E_j$. Note that because $A$ is the sum of a multiple of $D_\varphi$ and an effective divisor, it too can only have poles on the $E_i$. 
    
    To this end, define a function $g : [0,1] \to \mathbb{R}$ by\footnote{One can naturally think of this as the smallest degree of vanishing attained by the $\mathbb{R}$-divisor 
    \begin{equation}
        (1 - t)(m\mathrm{div}(\psi) + mD) + t(\mathrm{div}(\varphi) + mD) = (1 - t)m\mathrm{div}(\psi) + t\mathrm{div}(\varphi) + mD
    \end{equation} across all of the $E_i$, as a function of $t$.}
    \begin{equation}
        g(t) = \mathrm{min}_i\Big((1-t)\nu_{E_i}(mD_\psi) + t \, \nu_{E_i}(\mathrm{div}(D_\varphi)\Big).
    \end{equation}
    We note that $g(0) < 0$, because $\psi^m$ has a pole on some $E_i$. Additionally, $g(1) > 0$, as each $E_i$ belongs to the base locus of $mD$, so the degree of vanishing of $D_\varphi$ on each $E_i$ is at least one. Thus, there is a unique rational $s \in (0,1)$ such that $g(s) = 0$. Pick an integer $c$ that clears denominators such that $k = c(1-s), \ell = cs$ are integral. For these choices of $k, \ell$, the divisor $A$ defined above is effective and doesn't contain some $E_j$. Indeed, as noted, for effectiveness we only need to check for poles on the $E_i$, and $\nu_{E_i}(A) \geq \mathrm{min}_i(\nu_{E_i}(A)) = c g(s) = 0$, meaning $A \geq 0$ but also that $\nu_{E_j}(A) = 0$ for some $j$.
\end{proof}
\begin{corollary}
    \label{cor:general_sublattice}
    Let $X$ be a normal projective variety with a Weil divisor $D$ satisfying $[D] \in \mathcal{E}_X$. If $E$ is a prime divisor contained in the stable base locus of $[D]$, then $H^0(X, \mathcal{O}_X(D)) = H^0(X, \mathcal{O}_X(D + E))$.
\end{corollary}
\begin{proof}
    Certainly $H^0(X, \mathcal{O}_X(D)) \subset H^0(X, \mathcal{O}_X(D + E))$ because global sections of $D$ obey a strictly stronger inequality. To see the converse, note that if $\varphi$ is a rational function satisfying $\mathrm{div}(\varphi) + D + E \geq 0$, restricting the inequality to $X \setminus E$ reveals that $\varphi \in \mathcal{O}_X(D)(X \setminus E)$, meaning we can apply \cref{lem:lift_to_global} and conclude that $\varphi \in H^0(X, \mathcal{O}_X(D))$. 
\end{proof}
\begin{corollary}
    Let $X$ be a Calabi--Yau threefold with a non-movable hole $[D]$ whose stable base locus has prime components $E_1, \dots, E_n$. Then $h^0(X, \mathcal{O}_X(D + \sum_i a_i E_i)) = h^0(X, \mathcal{O}_X(D))$ for any $a_i \geq 0$. Moreover, if \cref{conj:nonbig_movable_CY} holds, then this result holds for any hole.
    \label{cor:semigroup}
\end{corollary}
\begin{proof}
    The first statement follows directly from \cref{cor:general_sublattice}. The second statement uses that \cref{conj:nonbig_movable_CY} implies that all holes are non-movable.
\end{proof}
This result entails that non-movable holes must come in non-trivial semigroups with elements of the form $[D] + \sum_i a_i [E_i]$, for non-negative $a_i$, with the semigroup operation being $([D] + \sum_i a_i [E_i]) + ([D] + \sum_i b_i [E_i]) = [D] + \sum_i (a_i + b_i)[E_i]$. In the cases we will study, the $[E_i]$ will be prime toric divisors. The natural ``origin'' $[D]$ of the semigroup is the divisor class $[D]$ such that subtracting off any of the $[E_i]$ yields a class that lacks $[E_i]$ in its stable base locus.

To construct the semigroup associated to a hole, one must understand the stable base locus of that hole. This can be studied directly by computing the global sections of the divisor and its multiples, but it can also be accessed by understanding the birational geometry of the variety, as we will discuss in the next section. For Calabi--Yau threefolds, this birational geometry can then be understood using, e.g., the methods of~\cite{Gendler:2022ztv}.

\subsection{Holes and the Minimal Model Program}

\label{sec:minimal_model}

In this section, we employ a useful modern method in birational geometry --- the minimal model program --- to interpret holes on Calabi--Yau threefolds. In particular, we illustrate that holes on a smooth Calabi--Yau threefold $X$ correspond to nef holes on a singular Calabi--Yau threefold birational to $Y$. For a class of holes, $Y$ has torsion in its class group, precisely due to the presence of the hole. Before exhibiting these results, we first take some time to introduce the minimal model program. A more comprehensive review of the subject can be found, for example, in \cite{KM}

An important problem in algebraic geometry is the classification of varieties up to birational equivalence. Any time one wants to classify equivalence classes of any sort, it is useful to identify some kind of normal form or canonical representative of each class: this is exactly how the birational classification problem was approached. We are actually already equipped to see the basic idea: recall from \cref{th:movable_is_nef_somewhere} that for some varieties, the movable cone decomposes into nef cones, so that any movable divisor class $[D]$ is usually nef for exactly one variety in the birational class (though it can be that $[D]$ falls on the intersection of more than one nef cone, in which case $[D]$ is nef on each associated variety). Performing birational operations to make $[D]$ nef is essentially the process of iteratively blowing down curves contained in $[D]$ --- i.e., curves in the basepoint locus of $[D]$, which have negative intersection pairing with $[D]$ --- and replacing them with curves transversely intersecting $[D]$ (i.e., pairing positively with $[D]$). Importantly, for the divisor classes that are effective but not movable, this same process can be performed: it just requires the contraction of divisors in the basepoint locus in addition to curves.

The idea of the minimal model program is to represent a birational equivalence class by the (almost) unique representative for which the canonical class $[K]$ is nef --- a minimal model. This isn't always possible, for example if $[K]$ isn't within the effective cone. However, in some sense the most common varieties are those for which $[K]$ is not only in the effective cone but is in fact big. These are the varieties of ``general type.''\footnote{For example, a complete intersection hypersurface in $\mathbb{P}^n$ is general type if and only if it is of degree $> n+1$: that is, there are only finitely many degrees yielding hypersurfaces that are not general type.} It's worth noting that compact projective toric varieties, on the other hand, are never of general type.

The minimal model program is important for us because nef divisors are very well-behaved --- for example, they obey vanishing theorems like Kawamata--Viehweg --- and some properties of divisors that we are interested in are preserved under birational maps, meaning we are free to perform the computation on any birational representative we choose. A convenient choice will indeed be to choose a representative where $[D]$ is nef: i.e., to perform the minimal model program on $[D]$. This is known as running the \textit{$[D]$-minimal model program} (i.e., the original minimal model program is the $[K]$-minimal model program). While the original minimal model program is trivial for Calabi--Yau threefolds --- the canonical class is trivial, so all representatives of a birational class are minimal --- the $[D]$-minimal model program for general $[D]$ will prove interesting and useful.

The existence of $[D]$-minimal models for general varieties is a decidedly open problem, due to the difficulty of showing both that the requisite birational maps exist and that the $[D]$-minimal model program converges (i.e., terminates in finitely many steps). Fortunately, strong results have been shown in the settings we care about in this paper. First, we have the following result for Mori dream spaces, of which toric varieties are a special case.
\begin{theorem}[\cite{hu2000mori}, Prop. 1.11]
    If $[D]$ is an effective divisor class on a Mori dream space $X$, then a $[D]$-minimal model $Y$ exists: that is, we have a birational map $f : X \dashrightarrow Y$ such that $[f_* D]$ is nef on $Y$.
\end{theorem}
In addition to this, considerable effort has been invested into the construction of \textit{log minimal models} for pairs $(X,D)$ for $X$ a variety and $D$ a $\mathbb{Q}$-divisor ---  a formal $\mathbb{Q}$-linear combination of divisors --- with mild singularities, which in our earlier language are $([K] + [D])$-minimal models for $[K]$ the canonical class of $X$. If $[K] = 0$ --- as is the case for Calabi--Yau varieties --- then a log minimal model for $(X,D)$ is a $[D]$-minimal model. In three dimensions, we have the following result.
\begin{theorem}
    \label{th:cy_log_minimal_model}
    If $[D]$ is a divisor class in the effective cone of a smooth Calabi--Yau threefold $X$, then there exists a $[D]$-minimal model $Y$: that is, we have a birational map $f : X \dashrightarrow Y$ such that $[f_* D]$ is nef on $Y$.
\end{theorem}
\begin{proof}
    It suffices to let $[D]$ be effective: if $[D]$ were a hole, we can express it as a $\mathbb{Q}$-linear combination of effective divisor classes and clear denominators, yielding an effective divisor, and minimal models are invariant under overall rescaling. For sufficiently small non-negative rational $\epsilon$, $(X, \epsilon D)$ is a \textit{klt pair}\footnote{See \cref{sec:toric_hyper_holes} for a brief elaboration on klt pairs.} (this is a well-known result for effective divisors on smooth varieties: see, e.g., Lemma 4.2 in \cite{mckernan2010mori}). Because $[D]$-minimal models on Calabi--Yau threefolds $X$ coincide with log minimal models for pairs $(X, \epsilon D)$, it suffices to show that log minimal models exist for klt pairs on threefolds, which is true (see, the introduction of \cite{birkar2010existence}). 
\end{proof}

We comment in passing that the theorems we've reviewed here can be readily understood as \textit{finiteness} results. The existence of birational maps needed to construct a $[D]$-minimal model often relies on certain nef divisor classes being \textit{semiample}, which is a finiteness condition (in particular, finite generation of the section ring of the divisor). Indeed, Mori dream spaces are so nice because their entire Cox rings are finitely generated, and the Kawamata--Morrison cone conjecture \cite{AST_1993__218__243_0,kawamataCone} can be understood as the conjecture that Calabi--Yau manifolds are Mori dream space {``}`up to the action of birational automorphism groups.'{''} The termination of a $[D]$-minimal model program in finitely many steps is also immediately a finiteness condition. We mention this because finiteness is also an important principle in quantum gravity~\cite{Vafa:2005ui,Hamada:2021yxy,Grimm:2020cda,Tarazi:2021duw,Kim:2024eoa,Delgado:2024skw,Agmon:2022thq,Grimm:2021vpn,Grimm:2023lrf}, and it would be worthwhile to understand if and/or how these algebraic geometric notions of finiteness connect to finiteness in quantum gravity. Indeed, for example, the role of semiamplitude in quantum gravity was already stressed in \cite{Katz:2020ewz}.

As we've alluded to, the $[D]$-minimal model program is useful to us because it preserves the global sections of $[D]$.
\begin{theorem}
    \label{th:MMP_global_sections}
    If $X$ is a normal projective $\mathbb{Q}$-factorial variety with divisor $D$ such that a $[D]$-minimal model $Y$ exists, then the birational map $f : X \dashrightarrow Y$ obeys $h^0(X, \mathcal{O}_X(D)) = h^0(Y, \mathcal{O}_Y(f_* D))$ and $[f_* D]$ is nef.
\end{theorem}
\begin{proof}[Sketch of proof]
    This is explicitly stated, e.g., in the introduction of \cite{kaloghiros2016finite}, and more precisely follows from the first part of Remark 2.4 alongside the result from Th. 5.2 that the steps of the $[D]$-MMP are $[D]$-nonpositive.
\end{proof}
\begin{corollary}
    \label{th:CY_MMP_holes}
    If $X$ is a smooth Calabi--Yau threefold with hole $D$, then there is a birational map $f : X \dashrightarrow Y$ to a $[D]$-minimal model $Y$ satisfying $h^0(X, \mathcal{O}_X(D)) = h^0(Y, \mathcal{O}_Y(f_* D))$ with $[f_* D]$ nef.
\end{corollary}
\begin{proof}
    This follows from \cref{th:cy_log_minimal_model} together with \cref{th:MMP_global_sections}.
\end{proof}
While we do not present a complete proof of \cref{th:MMP_global_sections} in this work, thanks to the results we introduced in \cref{sec:friends} we can quickly prove a result that we believe captures much of the intuition. In particular, the $[D]$-MMP essentially proceeds by iteratively contracting the prime components of the stable base loci of $[D]$ (see, e.g., Lem. 4.1 in \cite{kaloghiros2016finite} or part 3 of Lemma 3.6.5 in \cite{birkar2010existence}). Thus, the task is roughly to show that these contractions preserve global sections of $[D]$, which we show as follows.
\begin{prop}
    Let $f : X \to Y$ be a birational morphism between normal projective varieties and let $D$ be a divisor with class in the effective cone such that the stable base locus of $[D]$ contains the codimenson-one exceptional locus of $f$. Then $h^0(X, \mathcal{O}_X(D)) = h^0(Y, \mathcal{O}_Y(f_* D))$.
    \label{prop:birational_preserve}
\end{prop}
\begin{proof}
    Let $U \subset X$ be a dense open such that $f|_U$ is an isomorphism. We then have the following diagram.
    \[\begin{tikzcd}
        {H^0(X, \mathcal{O}_X(D))} && {\mathcal{O}_X(D)(U)} \\
        \\
        {H^0(Y, \mathcal{O}_Y(f_* D))} && {\mathcal{O}_Y(f_* D)(f(U))}
        \arrow["\cong"{description}, tail reversed, from=1-1, to=1-3]
        \arrow["\cong"{description}, tail reversed, from=1-3, to=3-3]
        \arrow["\cong"{description}, tail reversed, from=3-1, to=3-3]
    \end{tikzcd}\]
    The top equivalence follows from \cref{lem:lift_to_global}. The right equivalence follows because $f|_U$ is an isomorphism, so the sheaves are isomorphic. The bottom equivalence follows because $Y \setminus f(U)$ has codimension $\geq 2$, so local sections on $f(U)$ lift to $U$ as a consequence of normality.
\end{proof}
Applying \cref{th:CY_MMP_holes} to holes $[D]$ allow us to pass to a $[D]$-minimal model where $[D]$ remains a hole, but is now additionally nef. This may seem to be in tension with \cref{th:big_movable_CY}: at least, if $[D]$ were big, then it may seem that the pushforward would be obligated to be effective on the $[D]$-minimal model $Y$. However, \cref{th:big_movable_CY} can be evaded by $Y$ not being smooth: indeed, all Calabi--Yau varieties related to a smooth Calabi--Yau threefold $X$ by a non-small birational map must be singular. 
In this way, we see that holes are intimately related to the existence of a singular $Y$ birationally related to the original $X$ featuring non-effective nef divisors $[D]$.\footnote{In fact, something a little stronger is true: $Y$ must not be smoothable by complex structure deformation, as generic deformations cannot increase the number of sections nor change nefness, so nef holes $[D]$ must remain nef and non-holomorphic under deformation. We note that a general, non-nef hole $[D]$ can cease to be a hole under generic deformation, because the effective cone may shrink.}

This perspective provides a natural interpretation for why the prime components $E_i$ of the stable base locus of a hole $[D]$ should generate the associated semigroup of holes. The pushforward $f_*$ to a $[D]$-minimal model $Y$ maps both $[D]$ and $[D] + \sum_i a_i [E_i]$ ($a_i \geq 0$) to $[f_* D]$, because the $[E_i]$ are exceptional divisors of $f$. In particular, $Y$ is a minimal model for any non-negative linear combination of this form, so it preserves the global sections of each of these divisors, so each divisor $[D] + \sum_i a_i [E_i]$ must be a hole. If we were to allow the $a_i$ to be negative, on the other hand, $Y$ may no longer be a minimal model, and correspondingly the $[E_i]$ may no longer belong to the stable base locus, as discussed in \cref{sec:friends}. The ``origin'' of the semigroup is then interpreted as essentially the divisor $[D]$ from which no more factors of $[E_i]$ can be subtracted without changing the minimal model of the resulting divisor.

Moreover, as we alluded to at the end of the previous section, studying holes using birational geometry approach also furnishes a method for computing these semigroup generators. That is, if we don't want to attempt to compute the stable base locus by brute force --- by explicitly computing the global sections of $m[D]$ for $m > 1$ --- it suffices to understand the $[D]$-minimal model, and in particular the divisors contracted along the way. On toric varieties, the prime exceptional divisors are always prime toric divisors, and we will only be studying examples of Calabi--Yau threefolds whose birational geometry is inherited from the ambient variety, such that their prime exceptional divisors are also prime toric.

Additionally, we gain insight into the number of generators these semigroups have: it is the difference in $h^{1,1}$ between the original Calabi--Yau threefold $X$ and any $[D]$-minimal model,\footnote{To be precise, it is the difference in Picard number, or rank of the Picard group of Cartier divisors, as $h^{1,1}$ is more challenging to define for singular varieties, which the $[D]$-minimal model will be for non-movable $[D]$ on a Calabi--Yau variety.} or the number of prime exceptional divisors contracted by the birational map $f$ to the $[D]$-minimal model. We study an example in \cref{sec:h113 interior} where the structure of holes is more complicated than a single semigroup generated by the exceptional divisors of $f$ because there can be more than one nef hole on the $[D]$-minimal model. However, we will empirically find a decomposition into such semigroups is possible in general.

It's worth emphasizing an interesting special case that we will see in our later examples, in which holes on $X$ are the hallmark of torsion in the class group of a blowdown of $X$. We recall that a subgroup $H \subset G$ is saturated if $g^n \in H$ implies $g \in H$.
\begin{prop}
    \label{prop:torsion}
    Let $f : X \rightarrow Y$ be a birational morphism between normal projective varieties with prime exceptional divisors $E_1, \dots, E_n$. Then
    \begin{equation}
        \label{eq:quotient}
        \mathrm{Cl}(Y) \cong \mathrm{Cl}(X) / \bigoplus_i \mathbb{Z} [E_i]
    \end{equation}
    In particular, if the subgroup $H = \bigoplus_i \mathbb{Z} [E_i] \subset \mathrm{Cl}(X)$ is not saturated, then $\mathrm{Cl}(Y)$ has torsion. Moreover, divisor classes $[D] \notin H$ with multiples in $H$ are holes, and the classes $[f_* D]$ are pure torsion.
\end{prop}
\begin{proof}
    Let $U \subset X$ such that $f|_U$ is an isomorphism. Then the $E_i$ are the prime components of $X \setminus U$, and we have the following standard sequence (see, e.g., Theorem 4.0.20 in \cite{cls}), which directly implies the result follows.
    \[\begin{tikzcd}
    	{\bigoplus_i \mathbb{Z} [E_i]} && {\mathrm{Cl}(X)} && {\mathrm{Cl}(U)} && 0.
    	\arrow[from=1-1, to=1-3]
    	\arrow[from=1-3, to=1-5]
    	\arrow[from=1-5, to=1-7]
    \end{tikzcd}\]
    $U$ is isomorphic to a dense open $f(U) \subset Y$ whose codimension is $\geq 2$, so $\mathrm{Cl}(U) \cong \mathrm{Cl}(Y)$, from which \cref{eq:quotient} follows.

    Let $[D]$ be as described above: $[D]$ is a hole because if it were effective, it would be contracted, by linearity of the pushforward, and thus it would belong to $H$. Instead, their pushforwards $[f_* D]$ are torsion, because if $m[D] \in H$, then $m[f_* D] = 0$.
\end{proof}
In this way, holes can be symptomatic of torsion in the birational geometry of a Calabi--Yau threefold $X$.

We expect such holes to fall on faces of the boundary of the effective cone, as we now explain. For Calabi--Yau threefolds, we expect that the exceptional divisors $[E_i]$ of a birational morphism generate extremal rays of the effective cone $\mathcal{E}_X$. This is argued, for example, in \cite{Alim:2021vhs} (in particular in \S8.1). This is relatively intuitive: it is analogous to how the generators of extremal rays of the Mori cone are the curves contracted by flops (i.e., the curves which shrink near the boundary of the K\"ahler cone). Moreover, we expect the subcones of the effective cone whose divisors $[D]$ have a common $[D]$-minimal model to be generated by a face of the moving cone and the classes of the exceptional divisors of the birational map to that minimal model.\footnote{This is well understood, for example, for Mori dream spaces: see Prop. 1.11.2 in \cite{hu2000mori}.} In this way, the exceptional divisors of a fixed birational map ought to not only generate extremal rays of the effective cone, but their non-negative span should be a face of the effective cone. In particular, the non-movable holes arising from a non-saturated subgroup of exceptional divisor classes $[E_i]$ will form semigroups (in the sense of \cref{cor:semigroup}) on the boundary of $\mathcal{E}_X$. We will corroborate this expectation in many examples in \cref{sec:ksholes}. Even in \cref{sec:h113 interior}, when we find holes on the interior of the effective cone in a geometry where the non-trivial Hilbert basis element is on the boundary of $\mathcal{E}_V$, a carefully characterization of the semigroups reveals that the semigroup associated with the non-trivial Hilbert basis element is strictly contained in the boundary.

We provide examples of $[D]$-minimal models of holes $[D]$ in explicit examples in \cref{sec:ksholes}. Construct $[D]$-minimal models is difficult in general, but we can do this in the examples we study because there, the relevant Calabi--Yau hypersurface birational geometry is inherited from the ambient toric variety: i.e., when $[\hat{D}]$ is contracted on the ambient variety, then $[D]$ is contracted on the hypersurface. The birational geometry of toric varieties is combinatorial, and discussed in, e.g., \cite{cls, MacFadden:2025ssx}, and practical aspects of verifying that divisors are contracted (i.e., shrink) along certain faces of the nef cone is discussed in, e.g., \cite{Jain:2025vfh}. In practice, we construct the birational geometry of a Calabi--Yau threefold and compute the shrinking divisors using the methods of~\cite{Gendler:2022ztv}.

To prepare for our eventual examples from the Kreuzer--Skarke database, we take some time now to consider some simple singular Calabi--Yau threefolds with non-effective nef divisors. These varieties will end up featuring in the birational geometry of examples we will consider in \cref{sec:ksholes}.

\paragraph{A torsional example: $\mathbb{P}_{11114}/\mathbb{Z}_2$.} First, consider the generic anticanonical hypersurface $X$ in $\mathbb{P}_{11114}/\mathbb{Z}_2$, with minimal generators given as the columns of the following matrix.
\begin{equation}
    \begin{pmatrix}
        -5 & 0 & 0 & 1 & 1 \\
        -1 & 0 & 1 & 0 & 0 \\
        -1 & 1 & 0 & 0 & 0 \\
        -2 & 0 & 0 & 2 & 0
    \end{pmatrix}.
\end{equation}
A charge matrix for this toric variety is given by
\begin{equation}
    \label{eq:toy_torsion_charge_matrix}
    \left(\begin{array}{cccccc}
        1 & 1 & 1 & 1 & 4 \\ \hline 0 & 0 & 0 & 1 & 1
    \end{array}\right).
\end{equation}
Here, we have adopted a convention that we will use in a few later examples. In particular, the columns of a charge matrix can be interpreted as the classes of the prime toric divisors: in the case that the class group has torsion, classes of divisors also have a torsional part, which we include in the charge matrix below the horizontal line. It is well-understood for toric varieties that class group torsion (and non-simply connectedness) arise exactly when the minimal generators of the rays of the fan do not generate the $N$ lattice, but rather a finite-index sublattice $N'$ (i.e., the minimal generators do not generate a saturated subgroup). In particular, the torsion factor is $N / N'$ (\cite{cls}, Th. 12.1.10 and Ex. 5.4.10). 

In our case, $N \cong \mathbb{Z}^4$ and $N / N' = \mathbb{Z}_2$ ($(0,0,0,1)$ doesn't belong to $N'$ but its double does). The class group of divisors on the ambient variety can thus be identified with $\mathbb{Z} \times \mathbb{Z}_2$, meaning the final row of \cref{eq:toy_torsion_charge_matrix} should be understood to take values in $\mathbb{Z}_2$. By Cor. 1.9. in~\cite{Batyrev:2005jc}, this class group torsion descends to $X$. It is clear that the pure torsion class $(0,1)$ has no global sections on the ambient toric variety, and using the methods of \cref{sec:direct_comp}, one can show that the non-big pure torsion class on $X$ also has no sections, yielding an example of a hole: in particular, a non-big, non-effective nef divisor. In \cref{sec:h11=2} we will consider a geometry $Y$ in the Kreuzer--Skarke database that is a blowup of $X$, and that $Y$ features holes $[D]$ whose $[D]$-minimal models are $X$ and which, in particular, pushforward to this hole on $X$. In particular, they realize the mechanism of \cref{prop:torsion}.

\paragraph{A non-linearly weighted projective space example: $\mathbb{P}_{2,2,3,4,11}$.} Now, consider the generic anticanonical hypersurface $X$ in $\mathbb{P}_{2,2,3,4,11}$. The minimal generators of the ambient toric fourfold are given as the columns of the following matrix
\begin{equation}
    \begin{pmatrix}
        0 & 1 & -3 & -1 & 1 \\
        0 & 1 & -2 & 1 & 0\\
        0 & 2 & 0 & -1 & 0 \\
        1 & -1 & 0 & 0 & 0
    \end{pmatrix}\,.
\end{equation}
For this toric variety, we see that the hyperplane class, represented as $1$ upon fixing the charge matrix
\begin{equation}
    \begin{pmatrix}
        2 & 2 & 3 & 4 & 11
    \end{pmatrix}\,,
\end{equation}
does not correspond to any monomial in the Cox ring: all of the homogeneous coordinates have degrees that are ``too large.'' Thus, this is a hole on the ambient toric variety: in particular, it is a non-trivial Hilbert basis element. Using the methods of \cref{sec:direct_comp} one can compute that it is also a hole on $X$. In particular, this is a big, non-effective ample divisor on the singular Calabi--Yau variety $X$. 
In \cref{sec:h11=4} we will consider a geometry $Y$ in the Kreuzer--Skarke database that is a blowup of $X$, and $Y$ will feature holes $[D]$ whose $[D]$-minimal models are $X$ and which, in particular, pushforward to this hole on $X$.

The examples of holes that we have carefully inspected in the Kreuzer--Skarke database can all be explained by mechanisms largely analogous to these two. If $D$ is a hole, we find that the $D$-minimal model is either non-simply connected/torsional with the pushforward of $[D]$ being torsion --- as characterized by \cref{prop:torsion} --- or the pushforward of $[D]$ is nef (not necessarily always ample) and non-effective on a singular CY, arising from a nef non-trivial Hilbert basis element of the ambient variety.

\subsection{Holes on Toric Hypersurfaces} \label{sec:toric_hyper_holes}

Thus far, our discussion in this section has been about general Calabi--Yau threefolds. At this point we restrict to the case that our Calabi--Yau threefold $X$ is a smooth hypersurface in a simplicial toric variety $V$. In this setting, the non-trivial Hilbert basis elements (as defined in \cref{sec:review-geometry}) are natural candidates for ``primitive'' holes whose non-holomorphic representatives are capable of the smallest volumes, in the following sense. Any class in the toric effective cone $\mathcal{E}_V$ that is not a member of the Hilbert basis is a non-negative linear combination of such basis elements and thus has strictly larger volume than some basis elements. From a physics perspective, then, holes in the Hilbert basis are especially important to identify. 

In particular, we will show the following result.
\begin{prop}
    If $[D]$ is a non-trivial Hilbert basis element for $[\hat{D}]$ big, then $[D]$ is a hole: $h^0(X, \mathcal{O}_X(D)) = 0$.
    \label{prop:big}
\end{prop}
We note that it is the bigness of $[\hat{D}]$ that is important, but this is actually equivalent to bigness of $[D]$ because $\mathcal{E}_V \subset \mathcal{E}_X$.
We recall from \cref{eq:coker_ker} that $h^0(X,\mathcal{O}_X(D))$ receives contributions from $\mathrm{coker} \, F$ and $\mathrm{ker} \, \Tilde{F}$. There are no contributions from the former for non-trivial Hilbert basis elements, so for such divisors we will often exploit that it suffices (though is not necessary) to have $h^1(X,\mathcal{O}_X(D+K)) = 0$ for $\mathrm{ker} \, \Tilde{F} = 0$ (as was also mentioned in \cref{sec:review-geometry}). 

As a warm-up, we have the following result.
\begin{prop}
    If $[D]$ is a non-trivial Hilbert basis element and a hole, then $[D]$ cannot descend from a big and movable divisor on $V$.
    \label{prop:toric_big_movable}
\end{prop}
\begin{proof}
    By construction $h^0(V, \mathcal{O}_V(D)) = 0$, and we have that $h^1(V, \mathcal{O}_V(D+K)) = 0$ as a direct corollary of the toric Kawamata--Viehweg vanishing theorem (e.g., Th. 9.3.10 in~\cite{cls}), possibly replacing $V$ by one of its flips to render $[D]$ nef (using \cref{th:movable_is_nef_somewhere}). Together these imply that \cref{eq:coker_ker} is trivial.
\end{proof}
This naively implies that for a non-trivial Hilbert basis element to induce an autochthonous divisor, it must either be non-big on $V$ --- and thus fall on the boundary of the toric effective cone --- or non-movable. 

We can actually prove the stronger statement, \cref{prop:big}, because the Kawamata--Viehweg vanishing theorem admits a generalization involving so-called \textit{klt (Kawamata log terminal) pairs}, which are pairs $(V, \Delta)$ for $V$ a variety and $\Delta$ a $\mathbb{Q}$-divisor. Klt pairs play a crucial role in the minimal model program and algebraic geometry more broadly, but for our purposes, it suffices to think of them as fractional positive linear combination of divisors which behave nicely together (i.e., they do not intersect in a pathological way). Kawamata--Viehweg vanishing then holds not only for big nef divisors but for divisors that are a sum of $\Delta$ and a big nef divisor, as we will state precisely in \cref{th:klt_kv}.

Let's demonstrate how we can apply this. First, note that non-trivial Hilbert basis elements are ``fractional,'' in a sense we now formalize. 
\begin{lem}
    \label{lem:fraction_coefs_nontrivial_hilb}
    There is no expression of a non-trivial Hilbert basis element $[D]$ as a non-negative linear combination of classes of prime toric divisors $[D] = \sum_\rho a_\rho [D_\rho]$ such that any $a_\rho \geq 1$. 
\end{lem}
\begin{proof}
    Let a fixed $a_\rho$ be greater than or equal to $1$. Then $[D] = ([D] - [D_\rho]) + [D_\rho]$, so $[D]$ couldn't have belonged to the Hilbert basis.
\end{proof}
With this in hand, we can show that big non-trivial Hilbert basis elements are only ``fractionally'' outside of the toric moving cone, and thus fall in the purview of the following klt Kawamata--Viehweg vanishing theorem. 
\begin{theorem}[Th. 1.1 in~\cite{tanakakawamata}]
    Let $(V,\hat{\Delta})$ be a klt pair and $\hat{L}$ a $\mathbb{Q}$-Cartier Weil divisor on $V$. If $\hat{L} - (\hat{K}_V + \hat{\Delta})$ is big and nef, then $h^i(V,\mathcal{O}_V(\hat{L})) = 0$ for $i > 0$.
    \label{th:klt_kv}
\end{theorem}
\begin{proof}
    This is exactly Th. 1.1 in~\cite{tanakakawamata} upon choosing $S$ in their statement to be $\mathrm{Spec}(\mathbb{C})$, a point.
\end{proof}
Fortunately, there is a simple sufficient condition for a $\mathbb{Q}$-divisor $\hat{\Delta}$ to induce a klt pair $(V,\hat{\Delta})$ for $V$ a toric variety.
\begin{prop}[Prop. 11.2.14 in~\cite{cls}]
    If $\hat{\Delta} = \sum_\rho a_\rho \hat{D}_\rho$ on a normal, $\mathbb{Q}$-factorial toric variety $V$ such that $0 \leq a_\rho < 1$, then $(V,\hat{\Delta})$ is a klt pair.
    \label{prop:toric_klt}
\end{prop}
Finally, we can prove our main result.
\begin{proof}[Proof of \cref{prop:big}]
    Let $[\hat{D}]$ be a big divisor class on $V$ which descends to a non-trivial Hilbert basis element $[D]$ on $X$. We already handled the movable case in \cref{prop:toric_big_movable}, so let $[\hat{D}]$ be non-movable. Let $\hat{D} = \sum_\rho a_\rho \hat{D}_\rho$ be a torus-invariant representative of $[\hat{D}]$ with non-negative rational coefficients. By \cref{lem:fraction_coefs_nontrivial_hilb}, these coefficients will all be strictly less than $1$. 

    We would like to show that $\hat{D}$ is the sum of a big movable divisor $\hat{P}$ and a torus-invariant divisor $\hat{N}$ whose coefficients are strictly less than one (to use in a klt pair). The Newton polytope\footnote{It's worth noting that the lattice points inside Newton polytopes constructed from $\mathbb{Q}$-linear combinations of prime toric divisors do not count global sections; this is only true for $\mathbb{Z}$-linear combinations. However, nowhere in this proof will we utilize the lattice points inside Newton polytopes for this purpose.} $\Delta_{\hat{D}}$ is bounded by affine hyperplanes: one per ray in the toric fan, or equivalently, one per homogeneous coordinate/prime toric divisor. The affine hyperplanes which do not intersect $\Delta_{\hat{D}}$ correspond to the prime toric divisors in the stable basepoint locus of $[\hat{D}]$. Because $[\hat{D}]$ is assumed not to be movable, such hyperplanes will be present: let $\hat{P} = \sum_\rho b_\rho \hat{D}_\rho$ denote the torus-invariant divisor achieved by starting with $\hat{D}$ and reducing the $a_\rho$ for each $\rho$ associated to a non-intersecting hyperplane of $\Delta_{\hat{D}}$ until that hyperplane intersects the boundary of $\Delta_{\hat{D}}$. The $b_\rho$ remain non-negative: to see this, note that $\Delta_{\hat{D}} = \Delta_{\hat{P}}$ by construction, and $0 \in \Delta_{\hat{D}}$, while if a $b_\rho$ was negative then $\Delta_{\hat{P}}$ wouldn't contain $0$. We then have that $\hat{N} = \hat{D} - \hat{P}$.\footnote{The decomposition $\hat{D} = \hat{P} + \hat{N}$ is a special case of the Zariski decomposition for Mori dream spaces \cite{hu2000mori}} 
    
    The bigness of $[\hat{P}]$ follows from $\Delta_{\hat{D}} = \Delta_{\hat{P}}$, as divisors with solid Newton polytopes are big, following \S9.3 in \cite{cls}. $\hat{P}$ is also movable, as all hyperplanes bounding its Newton polytope intersect the polytope by construction, entailing that there is no stable base locus. Finally, both $\hat{P}$ and $\hat{N}$ are positive linear combinations of prime toric divisor classes with coefficients strictly less than one, because $0 \leq b_\rho \leq a_\rho < 1$. 

    Because $[P]$ is movable, by \cref{th:movable_is_nef_somewhere} it is nef on some $V'$ (with anticanonical hypersurface $X'$) related to $V$ by flips, which preserve global sections (and bigness) on $V$ and $X$. The pair $(V', \hat{N})$ is then klt by \cref{prop:toric_klt}. We can now apply \cref{th:klt_kv} to our klt pair. Choose $[\hat{L}] = [\hat{D}] + [\hat{K}]$: then $[\hat{L}] - ([\hat{K}] + [\hat{N}]) = [\hat{D}] + [\hat{K}] - ([\hat{K}] + [\hat{N}]) = [\hat{P}]$ is indeed nef and big. Thus, $h^1(V', \mathcal{O}_{V'}(\hat{D} + \hat{K})) = h^1(V', \mathcal{O}_{V'}(\hat{L})) = 0$, so \cref{eq:coker_ker} is trivial and $h^0(X,\mathcal{O}_{X}(D)) = h^0(X',\mathcal{O}_{X'}(D)) = 0$.
\end{proof}

\section{Computation of $h^{0}(X,\mathcal{O}_X(D))$} \label{sec:direct_comp}

As we've seen, the star of the show for us is the quantity $h^0(X,\mathcal{O}_X(D))$, as it determines the effectiveness/holomorphicity of $D$. A natural question to ask is how we can compute this cohomology for a fixed Calabi--Yau threefold hypersurface $X$ and divisor classes $[D]$. We recall that we expressed this cohomology in \cref{eq:CY_global_section}. From \cref{sec:review-geometry}, we already know how to compute global sections of divisors on a toric variety using the Newton polytope from \cref{eq:newton_poly}. It therefore suffices to understand how to compute higher line bundle cohomology groups on toric varieties, as well as the ranks of maps between such groups. In this present section, we will explain how to perform these computations using the basic building blocks of a toric variety presented in \cref{sec:review-geometry}: the Cox ring $S$, the GLSM charge matrix $Q$, and the toric fan/SR ideal. This approach is directly informed by the \verb+StringTorics+~package~\cite{StringToricsSource} in \verb+Macaulay2+~\cite{M2}. This material is also reviewed in Sec. B.1 of \cite{Demirtas:2019lfi}.

\paragraph{Computing $h^i(V,\mathcal{O}_V(\hat{D}))$.} We now briefly summarize how to compute line bundle cohomology on toric varieties. Methods for this have long been known~\cite{mustata_local_2000, eisenbud_coho, maclagan_multigraded_2004} and have been implemented, for example, in the \verb+NormalToricVarieties+~package~\cite{NormalToricVarietiesSource} in \verb+Macaulay2+~\cite{M2}. We take direct inspiration from the approach introduced in~\cite{Blumenhagen:2010pv, Demirtas:2019lfi} (see~\cite{jow_cohomology_2011, Rahn:2010fm, cohomCalg:Implementation} for follow-ups and~\cite{Blumenhagen:2010ed} for applications).

Recall that we can identify monomials $x^a \in S$ ($a \in \mathbb{Z}^{h^{1,1}+4}$) in the homogeneous coordinate ring and the associated torus-invariant divisor $\hat{D} = \sum_i a_i \hat{D}_i$ with a divisor class, namely the class $[\hat{D}]$ of $\hat{D}$. Recall further that a GLSM charge matrix $Q$ sets a basis for the group of divisor classes, and in this basis the class of $x^a$ is merely $Qa$. In this way, $S$ is graded by the group of divisor classes, and the global sections of $\mathcal{O}_V(\hat{D})$ are generated by the monomials in $S$ with grading $[\hat{D}]$: more generally, the local sections of $\mathcal{O}_V(\hat{D})$ --- those sections only defined on proper open subsets away from their poles --- are generated by monomials in the ring of fractions $\mathrm{Frac}(S)$ with grading $[\hat{D}]$. Local sections of $[\hat{D}]$ will represent elements of the higher line bundle cohomology $H^i(V, \mathcal{O}_V(\hat{D}))$, so we are motivated to spend some time studying them.

We can count local sections of $\mathcal{O}_V(\hat{D})$ which have poles on specific prime toric divisors using polytopes, generalizing the Newton polytope introduced in \cref{sec:review-geometry} to count global sections. Consider picking some subset of the prime toric divisors $\hat{D}_{j_1}, \dots, \hat{D}_{j_k}$: we can then parameterize this subset by a vector $\beta$ of length $h^{1,1} + 4$ valued in $\{0, 1\}$ such that $\beta_i = 1$ if $i \in \{j_1, \dots, j_k\}$ and $0$ otherwise. The local sections of $\hat{D}$ with poles on exactly the divisors $\hat{D}_{j_1}, \dots, \hat{D}_{j_k}$ are generated by monomials counted by the following polytope $\Delta_{\hat{D},\beta}$.
\begin{equation}
\label{eq:Hi Mlattice}
    \Delta_{\hat{D}, \beta} = \left\{ m \in M \; \Bigg| \; \langle m, u_i \rangle + a_i 
    \begin{cases} \geq 0 & \beta_i = 0 \\ < 0 & \beta_i = 1 \end{cases} \right\}.
\end{equation}
These lattice points $m \in M$ are mapped to local sections as $m \mapsto \prod_i x_i^{\langle m, u_i \rangle + a_i}$, just as before. In this way, the inequalities are defined merely to enforce that $x_i$ has a negative power if we want a pole on the associated divisor $\hat{D}_i$ (which is cut out by $x_i = 0$). For example, the Newton polytope $\Delta_{\hat{D}}$ is $\Delta_{\hat{D},(0, \dots, 0)}$, whose sections have poles nowhere. We will find it useful to define $|\beta|$ to be the numbers of $1$s in $\beta$: i.e., the number of prime toric divisors upon which the associated local sections have poles. It will also be convenient to define a map $\mathrm{neg} : \mathbb{Z}^{h^{1,1} + 4} \to \{0,1\}^{h^{1,1} + 4}$ given as $\mathrm{neg}(a)_i = 1$ if $a_i < 0$ and $0$ otherwise. In particular, the set of monomials $x^a$ that take the form $x^{\langle m, \hat{D} \rangle}$ for some $m \in \Delta_{\hat{D}, \beta}$ are exactly those monomials with grading $[\hat{D}]$ and $\mathrm{neg}(a) = \beta$. The key idea is that the extent to which a monomial $x^a$ contributes to line bundle cohomology depends only on $\mathrm{neg}(a)$. 

For each fixed $\beta$, the monomials counted by $\Delta_{\hat{D},\beta}$ contribute to $H^i(V, \mathcal{O}_V(\hat{D}))$ some number of times, determined by the combinatorics of the fan $\Sigma$. We can interpret $\Sigma$ as a simplicial complex; let $\Sigma_\beta$ then be the restriction of the simplicial complex induced by the toric fan to the rays which are non-zero in $\beta$. It is then the reduced cohomology $\tilde{H}^{i-1}(\Sigma_\beta,\mathbb{C})$ of $\Sigma_\beta$ that counts the contributions of $\Delta_{\hat{D},\beta}$ to $H^i(V, \mathcal{O}_V(\hat{D}))$. 

This technically suffices to compute the higher cohomologies, but from a computational standpoint, it is useful to identify simpler necessary conditions for a given $\tilde{H}^{i-1}(\Sigma_\beta,\mathbb{C})$ to be non-zero --- i.e., for a given $\Delta_{\hat{D},\beta}$ to contribute --- such that we don't need to evaluate the simplicial cohomology for the exponentially many $\beta$. We will employ one such condition. Let $x^{\mathrm{SR}_1}, \dots, x^{\mathrm{SR}_\ell}$ generate the Stanley--Reisner (SR) ideal. The monomials are square-free, so the exponent vector $\mathrm{SR}_j$ takes its values in $\{0,1\}$. Let $R = \{\mathrm{SR}_1, \dots, \mathrm{SR}_\ell\}$ be the set of exponent vectors of generators of the SR ideal, and let $P(R)$ be the power set of $R$. Then $\Delta_{\hat{D},\beta}$ can only contribute to $H^i(V, \mathcal{O}_V(\hat{D}))$ if $x^\beta = \mathrm{lcm}(\{x^{\mathrm{SR}_{i_1}}, \dots, x^{\mathrm{SR}_{i_k}}\})$ such that $|\beta| - k = i$. For example, $\Delta_{\hat{D},\beta}$ contributes to $H^1(V, \mathcal{O}_V(\hat{D}))$ if $x^\beta$ is the degree of a single SR generator with two homogeneous coordinates, or a least common multiple of two SR generators with a total of three unique homogeneous coordinates, or of three SR generators with four unique coordinates, and so on.

We are now prepared to state a general formula for line bundle cohomology~\cite{eisenbud_coho,maclagan_multigraded_2004,Blumenhagen:2010pv}.
\begin{equation}
    H^i(V, \mathcal{O}_V(\hat{D})) \cong \bigoplus_{\substack{\beta \text{ s.t. } x^\beta = \mathrm{lcm}(S) \\ \text{for some }S \in P(R)\\ |\beta| - |S| = i}} \bigg[\bigoplus_{m \in \Delta_{D,\beta}} \mathbb{C} \cdot \prod_i x_i^{\langle m, u_i \rangle + a_i}\bigg]\otimes \tilde{H}^{i-1}(\Sigma_\beta,\mathbb{C})\,.
\end{equation}
We stress that this formula holds even without the condition on $\beta$, but we gain computationally efficiency by imposing that condition first. We also emphasize that we are principally interested in the cases $i = 0, 1$, so $\tilde{H}^{0}(\Sigma_\beta,\mathbb{C})$ is the only relevant non-trivial simplicial cohomology we need to compute, which is merely $\mathbb{C}^{k-1}$ for $k$ the number of connected components of $\Sigma_\beta$. 

We comment in passing that for the case of $i = 0$, then the only relevant $\beta$ is $(0, \dots, 0)$ and the reduced cohomology is trivial. The computation of $H^0(V,\cO_V(\hat{D}))$ therefore reduces to counting the monomials in $\Delta_{\hat{D}} = \Delta_{\hat{D}, (0, \dots, 0)}$, as we know from \cref{sec:review-geometry}.

In the ancillary files, we have provided a python implementation in which to calculate the number of monomials for a given $\beta$, we make use of the Coin-or Branch and Cut solver within the \verb+PuLP+ integer programming python library \cite{mitchell2011pulp} alongside the \verb+HPolytope+ class in \verb+CYTools+ (built on \verb+PPL+ \cite{BagnaraHZ08SCP}) to solve for all such points satisfying \cref{eq:Hi Mlattice}. 

\paragraph{Computing $\mathrm{Ker}(\tilde{F})$.}

From the previous section, we now know how to compute the cohomology groups constituting the domain and codomain of the $\Tilde{F}$ featuring in \cref{eq:CY_global_section}; we now discuss how to construct $\tilde{F}$.

In the basis fixed by the charge matrix $Q$, the anticanonical line bundle can be expressed as
\begin{equation}
    -[\hat{K}] = Q \cdot \underbrace{(1, \dots, 1)^\intercal}_{h^{1,1}(V)+4}\,.
\end{equation}
From the previous section, we know how to compute $H^1(V,\cO_V(\hat{D}+\hat{K}))$ and $H^1(V,\cO_V(\hat{D}))$. The map $\Tilde{F}$ between these cohomologies is the one induced by multiplication by $G$, the defining polynomial for the anticanonical hypersurface. Let us now set up some notation.
\begin{itemize}
    \item Let $e_{\beta,\mu}$ be a basis for $\tilde{H}^{0}(\Sigma_\beta,\mathbb{C})$.
    \item In particular, let the $e_{\beta,\mu}$-coefficient of $e \in \tilde{H}^{0}(\Sigma_\beta,\mathbb{C})$ be $\mathrm{coef}_\mu(e)$, such that $e = \sum_\mu \mathrm{coef}_\mu(e) e_{\beta,\mu}$.
    \item Let $x^{a_i} \otimes e_{\mathrm{neg}(a_i),\mu}$ be a basis for $H^1(V,\cO_V(\hat{D}+\hat{K}))$.
    \item Let $x^{b_j} \otimes e_{\mathrm{neg}(b_j),\nu}$ be a basis for $H^1(V,\cO_V(\hat{D}))$.
    \item Let $G = \sum_k \delta_k x^{c_k}$ for $x^{c_k}$ the global sections of $-\hat{K}$.
\end{itemize}
We would like to compute the matrix elements $\tilde{F}_{i\mu, j\nu}$ of $\Tilde{F}$ from $x^{a_i} \otimes e_{\mathrm{neg}(a_i),\mu}$ to $x^{b_j} \otimes e_{\mathrm{neg}(b_j),\nu}$. The map on the Cox ring monomial part of $x^{a_i}$ is induced by multiplication by $G$: for $\tilde{F}_{i\mu, j\nu}$ to be non-trivial, then, it is necessary that there is some $c_k$ such that $x^{a_i}x^{c_k} = x^{b_j}$ (i.e., $a_i + c_k = b_j$). If it exists, such a $c_k$ will of course be unique; in such a case, note that $\mathrm{neg}(a) \geq \mathrm{neg}(b)$ (as multiplication by a global section can only remove poles in a monomial). This means we have an inclusion $\Sigma_{\mathrm{neg}(b)} \to \Sigma_{\mathrm{neg}(a)}$ which induces a map $\tilde{H}_{i-1}(\Sigma_{\mathrm{neg}(b)},\mathbb{C}) \to \tilde{H}_{i-1}(\Sigma_{\mathrm{neg}(a)},\mathbb{C})$ on simplicial homology, whose dual on simplicial cohomology we denote $\pi_{ab} : \tilde{H}^{i-1}(\Sigma_{\mathrm{neg}(b)},\mathbb{C}) \to \tilde{H}^{i-1}(\Sigma_{\mathrm{neg}(a)},\mathbb{C})$. In the case of interest, $i = 1$, the simplicial homology map merely sends a connected component of $\Sigma_{\mathrm{neg}(b)}$ to the one containing it in $\Sigma_{\mathrm{neg}(a)}$, and the simplicial cohomology map is the induced map between the dual vector spaces. In summary, then,
\begin{equation}
    \tilde{F}_{i\mu, j\nu} =
    \begin{cases}
        \delta_k \cdot \mathrm{coef}_\nu(\pi_{ab}(e_{\mathrm{neg}(a),\mu})) & \text{if } x^{c_k}x^{a_i} = x^{b_j} \,,\\
    0 &\text{otherwise}\,.
    \end{cases}
\end{equation}
Here, for our purposes, we will be performing the calculations at generic points in the complex structure moduli. Hence, we can take $\delta_k$ to be a random real number.
Note, at special points in the complex structure moduli, which corresponds to specially tuned values of $\delta_k$, divisors may acquire additional global sections. In particular, non-effective divisors may become effective.

Now, having constructed $\tilde{F}$, we can readily compute its rank, furnishing the final ingredient needed to compute $h^0(X,\cO_X(D))$ using \cref{eq:coker_ker}.

\paragraph{Example.}

To demonstrate the above algorithm, let us consider the example of an autochthonous divisor in a Calabi--Yau threefold with $h^{1,1}=2$ and $h^{2,1}=86$. In particular, we'll consider the anticanonical hypersurface in $\mathbb{P}^3 \times \mathbb{P}^1$ (i.e., this geometry also belongs to the CICY database of complete intersection Calabi--Yau threefolds in products of projective spaces). $\mathbb{P}^3 \times \mathbb{P}^1$ is a fourfold in the Kreuzer--Skarke database, and the vertices of the associated 4D reflexive polytope are given as follows.
\begin{equation}
    \begin{pmatrix}
        0&0&0&0&-1&1\\
        -1&0&0&1&0&0\\
        -1&0&1&0&0&0\\
        -1&1&0&0&0&0
    \end{pmatrix}\,.    
\end{equation}
Then, a choice of GLSM matrix is
\begin{equation}
    \begin{pmatrix}
        1&1&1&1&0&0\\
        0&0&0&0&1&1
    \end{pmatrix}\,.
\end{equation}
From here, we can see that the toric effective cone is generated by $(1,0)$ and $(0,1)$. Therefore, all effective divisors inside $V$ take on the following form
\begin{equation}
    [\hat{D}]=a\begin{pmatrix}
       1\\
       0
    \end{pmatrix}+b
    \begin{pmatrix}
        0\\
        1
    \end{pmatrix}\,,
\end{equation}
where $a,b\in\mathbb{Z}_{\geq 0}$.
Now, instead, consider the divisor $[\hat{D}]=(4,-1)$ which is clearly not effective in $V$. However, computing the elements inside $H^1(V,\mathcal{O}_V(\hat{D}+\hat{K}))$ we find that they are generated by the monomials $x_5^{-2}x_6^{-1}, x_5^{-1}x_6^{-2}$ where the corresponding $x^\beta$ is $x_5x_6$. Then, we have $\Sigma_{\beta}=\{\{5\},\{6\}\}$ such that $\mathrm{dim} \, \tilde{H}^0(\Sigma_{\beta})=1$. Additionally, we find that $H^1(V,\mathcal{O}_V(\hat{D}))$ is trivial. Therefore, the kernel of $\tilde{F}$ is two-dimensional, which implies that $h^0(X,\mathcal{O}_X(D)) = 2 > 0$. Hence, $D$ is an autochthonous divisor.

\section{Holes in Kreuzer--Skarke} \label{sec:ksholes}

\noindent Now we are in a position to compute line bundle cohomologies inside Calabi--Yau threefolds. To this end, we utilize \verb+CYTools+~\cite{Demirtas:2022hqf} to extract both toric variety and Calabi--Yau threefold data from the Kreuzer--Skarke database~\cite{Kreuzer:2000xy}. Then, using the algorithm presented in the previous section and implemented in the ancillary files submitted alongside this paper, we will determine the effectiveness of divisors in many Calabi--Yau threefolds. For convenience and reproducibility, we work in the default basis chosen deterministically by \verb+CYTools+ for the $N$ lattice and group of divisor classes. 

We will identify holes in Calabi--Yau effective cones with the following strategy. For a fixed polytope in the Kreuzer--Skarke database, we can compute the non-trivial Hilbert basis elements (elements of the Hilbert basis of the toric effective cone that are not themselves classes of prime toric divisors), which are strong candidates for holes in the effective cone of the Calabi--Yau hypersurface. For each such divisor class $[D]$, we compute $h^0(X, \mathcal{O}_X(D))$ to check for its effectiveness. From here, we probe \cref{cor:semigroup} by evaluating $h^0(X, \mathcal{O}_X(D'))$ for $[D']$ given by a sum of $[D]$ and a non-negative linear combination of prime toric divisor classes. 

Our empirical findings are as follows. 
\begin{itemize}
    \item Beyond verifying the result that big non-trivial Hilbert basis elements must be holes (\cref{prop:big}), we find that their non-big counterparts are also always holes.
    \item In particular, recalling \cref{eq:CY_global_section}, for all non-trivial Hilbert basis elements $[D]$ appearing in Calabi--Yau threefolds with $h^{1,1}\leq 6$ within the Kreuzer--Skarke datalist, we find that $H^1(V, \mathcal{O}_V(D + K_V)) = 0$, hence implying that $\mathrm{Ker} \, \Tilde{F}$ vanishes. 
    \item The vast majority of non-trivial Hilbert basis elements are not big on $V$ (i.e., belong to the boundary of $\mathcal{E}_V$).
    \item In verifying \cref{cor:semigroup}, we find that holes nearly always come in \textit{codimension-one} families.
    \item The majority of non-trivial Hilbert basis elements arise for Calabi--Yau hypersurfaces that are not generic in complex structure.\footnote{In \cref{sec:review-geometry} we mentioned the favorability condition for a four-dimensional reflexive polytope, ensuring that $H^{1,1}(X)$ was fully inherited from the ambient variety. If the same condition holds for the dual polytope --- ``dual favorability'' --- then $H^{2,1}(X)$ is fully inherited from the sections of the anticanonical bundle: this duality is a consequence of mirror symmetry in the Kreuzer--Skarke database. We enforce that our polytopes are favorable but not dual-favorable, so we encounter many hypersurfaces that are not generic in their complex structure.} That is, it seems that when tuning the complex structure of a CY, the effective cone grows but often not all classes now in the effective cone are themselves rendered effective. In particular, we have no examples of a non-trivial Hilbert basis element on the \textit{interior} of $\mathcal{E}_V$ for a hypersurface at generic complex structure.
\end{itemize}
In particular, we take the first two findings above as strong evidence for the following conjecture.
\begin{conjecture}
    \label{conj:non_trivial_hb_are_holes}
    All non-trivial Hilbert basis elements on smooth Calabi--Yau threefolds arising from the Krezuer--Skarke database are holes.
\end{conjecture}
We also note that non-trivial Hilbert basis elements that belong to $\partial\mathcal{E}_V$ may also belong to $\partial\mathcal{E}_X$ --- which is more difficult to check, and which we do not do systematically in this work --- and when they do, by \cref{prop:torsion} they must be holes and are moreover indicative of class group torsion in the birational geometry of $X$, from \cref{prop:torsion}. Given that holes do predominantly fall in $\partial\mathcal{E}_V$, it would be interesting to understand whether the abundance of holes we find in \cref{fig:additional elements} (especially at large $h^{1,1}$) does indeed correspond to an abundance of class group torsion in the birational geometry of toric hypersurface Calabi--Yau threefolds.

We now present our findings in greater detail, presenting some illustrative examples as well as more quantitative results of our scan.

\subsection{Holes on Boundaries of Toric Effective Cones}

The first few occurrences of holes (examined in the order of increasing $h^{1,1}$) appear on the boundaries of the toric effective cones: i.e., they do not descend from big classes on the ambient variety. Here, we will demonstrate this phenomenon in two explicit examples. As they fall on the boundary, we must explicitly computationally check whether these divisors are holes using the methods of \cref{sec:direct_comp}.  

\subsubsection{Example 1: $h^{1,1}=2$.}
\label{sec:h11=2}

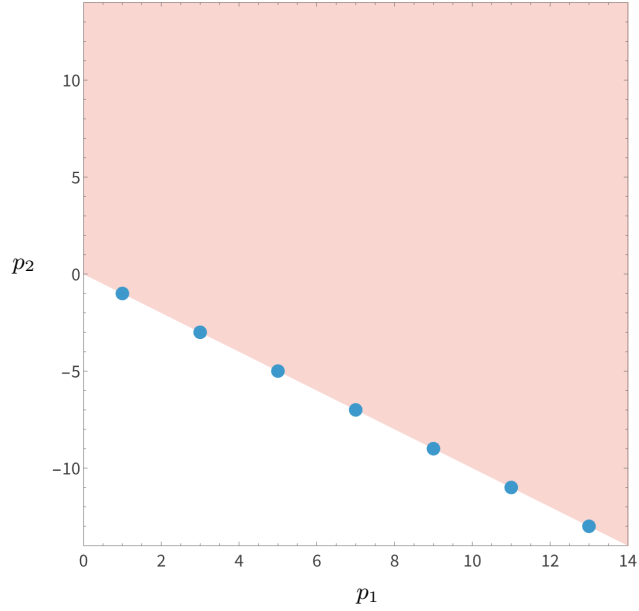
\begin{figure}
    \centering
    \includegraphics[width=0.5\linewidth]{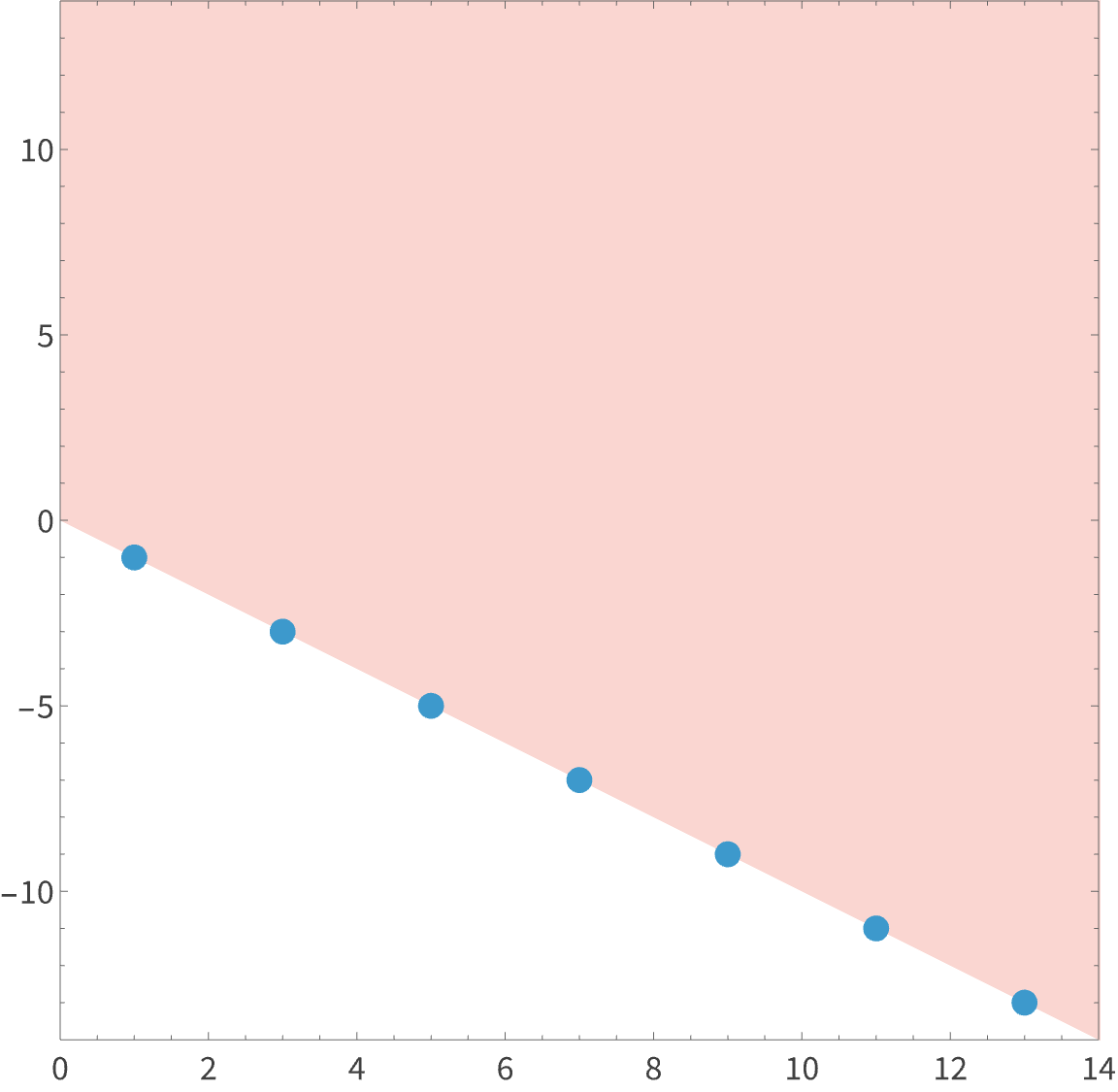}
    \begin{picture}(0,0)\vspace*{-1.2cm}
        \put(-110,-10){\footnotesize $p_1$}
        \put(-240,115){\footnotesize $p_2$}
    \end{picture}\vspace*{0.5cm}
    \caption{The toric effective cone $\mathcal{E}_{V_{2,106}}$ for the Calabi--Yau hypersurface $X_{2,106}$ (red) with non-effective divisor classes in the effective cone --- holes --- in blue. The variables $p_1, p_2$ denote the components of divisor classes when expressed in the basis of \cref{eq:X2106 GLSM}.}
    \label{fig:effective cone X2106}
\end{figure}

At $h^{1,1}=2$, we have a single Calabi--Yau threefold in the Kreuzer--Skarke database which exhibits a hole. This variety is a hypersurface in a toric variety $V_{2,106}$ arising from a fine, regular, and star triangulations (FRST) of the reflexive polytope with lattice points
\begin{equation}
\begin{pmatrix}
-5 & 0 & 0 & 1 & 1 & 1 \\
-1 & 0 & 1 & 0 & 0 & 0 \\
-1 & 1 & 0 & 0 & 0 & 0 \\
-2 & 0 & 0 & 0 & 2 & 1
\end{pmatrix}\,.
\end{equation}
A choice of GLSM for the toric variety is given by
\begin{equation}
\label{eq:X2106 GLSM}
\begin{pmatrix}
    1&1&1&3&0&2\\
    0&0&0&1&1&-2
\end{pmatrix}\,.
\end{equation}
The Kreuzer--Skarke database is canonically ordered for each value of $h^{1,1}$, so for convenience we will share the index of the polytopes we present as examples. That is, the index $n$ polytope for a fixed $h^{1,1}$ is the $(n+1)$th favorable polytope at that $h^{1,1}$ in the Kreuzer--Skarke ordering.\footnote{The Kreuzer--Skarke database is available \href{http://hep.itp.tuwien.ac.at/~kreuzer/CY/}{here}.} For this polytope, that index is 19, and the Hodge numbers of the Calabi--Yau hypersurface are $(h^{1,1},h^{2,1})=(2,106)$, so we will refer to it as $X_{2,106}$
The generators of the extremal rays of the effective cone in this choice of basis are
\begin{equation}
\begin{pmatrix}
    1\\
    -1
\end{pmatrix}\,,
\qquad \begin{pmatrix}
    0\\
    1
\end{pmatrix}\,,
\end{equation}
and the Hilbert basis matches with the set of extremal ray generators. However, one of the elements of the Hilbert basis (in particular, an extremal ray generator) is not a prime toric divisor: only one of its multiples is. Thus, this Hilbert basis element is non-effective on the ambient variety, and a cohomology computation demonstrates that it is additionally not effective on $X_{2,106}$. That is,
\begin{equation}
    h^0(X_{2,106},\mathcal{O}_{X_{2,106}}((1,-1))=0\,.
\end{equation}
Furthermore, this non-effectiveness continues to hold for all odd integer multiples of this divisor
\begin{equation}
    h^0(X_{2,106},\mathcal{O}_{X_{2,106}}((2n+1,-(2n+1)))=0\,.
\end{equation}
This family of non-effective divisors in the effective cone lie strictly on the boundary of the effective cone and can be written as 
\begin{equation}
    [D_a] =
    \begin{pmatrix}
        1\\
        -1
    \end{pmatrix}
    +a
    \begin{pmatrix}
        2\\
        -2
    \end{pmatrix}\,,
\end{equation}
where $a\in\mathbb{Z}_{\geq 0}$. The second divisor in the above equation is a prime toric divisor, namely the last column appearing in the GLSM. Thus, the non-effective divisors form a codimension-one semigroup in the toric effective cone. The effective cone along with the non-effective divisors for this Calabi--Yau threefold is illustrated in \cref{fig:effective cone X2106}.

This hole falls in the purview of \cref{prop:torsion}. Indeed, letting $[D] = (1, -1)$ denote the hole, the class $2[D]$ is a prime toric divisor that is contracted by a birational map, and the subgroup of the class group it generates is not saturated, because it does not contain $[D]$. On the ambient variety, contracting $2[\hat{D}]$ yields the toric variety $\mathbb{P}_{11114}/\mathbb{Z}_2$, the singular non-simply connected toric fourfold we considered in \cref{sec:friends}. The divisorial contraction of the hypersurface is inherited from the toric one, and yields an anticanonical hypersurface $Y$ in $\mathbb{P}_{11114}/\mathbb{Z}_2$, which inherits the $\mathbb{Z}_2$ torsion of the ambient variety. The pushforward of $[D]$ on $Y$ is the pure torsion class, which is a nef hole. It is also clear why $2[D]$ generates the semigroup of holes: it is the contracted divisor --- such that the entire semigroup is mapped to the nef hole on $Y$ by the pushforward --- and also the prime component of the stable base locus of $[D]$, following the discussion of \cref{sec:friends} and \cref{sec:minimal_model}.

\subsubsection{Example 2: $h^{1,1}=3$.}
\label{sec:h113 boundary}

\begin{table}
    \centering
    \begin{tabular}{c|c|c}
    ID & $(h^{1,1},h^{2,1})$ & $[D]$ \\ 
    \hline
       0 & (3,43) & $(1,-1,1)$ \\ 
       1 & (3,45) & $(1,-1,0),(1,0,-1)$ \\ 
       3 & (3,51) & $(-1,1,1)$ \\ 
       8 & (3,59) & $(1,-1,1)$ \\ 
       11 & (3,63) & $(1,-1,1)$ \\ 
       105 & (3,83) & $(1,-1,1)$ \\ 
       123 & (3,89) & $(-1,2,2)$ \\ 
       130 & (3,93) & $(1,-1,1)$ \\ 
       134 & (3,95) & $(-1,4,-1)$ \\
       139 & (3,99) & $(1,-1,1)$ \\
       166 & (3,105) & $(1,-1,0)$ \\
       170 & (3,105) & $(1,-2,-2)$  \\
       232 & $(3,165)$ & $(3,-1,1)$ \\ 
    \end{tabular}
    \caption{Non-trivial Hilbert basis elements in the Kreuzer--Skarke database at $h^{1,1}=3$. In particular, we find a total of 13 toric Calabi--Yau threefolds with non-trivial Hilbert basis elements lying on the boundary of the toric effective cone. We explicitly compute that all of these classes are holes, using the methods of \cref{sec:direct_comp}. Additionally, each of these geometries are non-generic in complex structure aside from $(3,51)$, the example we present in \cref{sec:h113 boundary}. Here, $[D]$ indicates the non-trivial Hilbert basis element in the toric effective cone.}
    \label{tab:h11=3}
\end{table}

Now, the Calabi--Yau threefolds with $h^{1,1}=3$ in the Kreuzer--Skarke database with non-trivial Hilbert basis elements also strictly feature them along the boundary of the toric effective cone. Here, we tabulate all such examples in \cref{tab:h11=3} and work out an explicit example. This will be the only example we discuss which is generic in its complex structure.

Consider the Calabi--Yau threefold obtained as a hypersurface in a toric variety from an FRST of the reflexive polytope with lattice points
\begin{equation}
    \begin{pmatrix}
        1&-1&0&0&-2&0&2\\
        0&0&1&0&-1&0&-1\\
        0&0&0&1&-1&0&0\\
        0&0&0&0&0&1&-1
    \end{pmatrix}\,.
\end{equation}
We can choose the following GLSM charge matrix for this toric fourfold.
\begin{equation}
    \label{eq:351_charge_matrix}
    \begin{pmatrix}
        1 & 0 & 1 & 1 & 0& -2&  0\\
         0  &1 & 0  &1  &1&  2 & 0\\
         0  &0 & 0  &0  &0  &1 & 1
    \end{pmatrix}\,.
\end{equation}
This Calabi--Yau threefold has the Hodge numbers $(h^{1,1},h^{1,2})=(3,51)$, so we will refer to it as $X_{3,51}$, with the ambient variety being $V_{3,51}$.

The toric effective cone of divisors, which we denote as $\mathcal{E}_{V_{3,51}}$, is generated by the columns of the following matrix.
\begin{equation}
\mathcal{E}_{V_{3,51}} = \text{cone}\begin{pmatrix}
0 & -2 & 1 & 0\\
1&2&0&0\\
0&1&0&1
\end{pmatrix}\,.
\end{equation}
Now, the Hilbert basis of $\mathcal{E}_{V_{3,51}}$ is
\begin{equation}
    \mathcal{H}(\mathcal{E}_{V_{3,51}}) =\left\{\begin{pmatrix}
        -2\\2\\1
    \end{pmatrix},\begin{pmatrix}
        -1\\1\\1
    \end{pmatrix},\begin{pmatrix}
        0\\0\\1
    \end{pmatrix},\begin{pmatrix}
        0\\1\\0
    \end{pmatrix},\begin{pmatrix}
        1\\0\\0
    \end{pmatrix}\right\}\,.
\end{equation}

\begin{figure}
    \centering
    \includegraphics[width=0.5\linewidth]{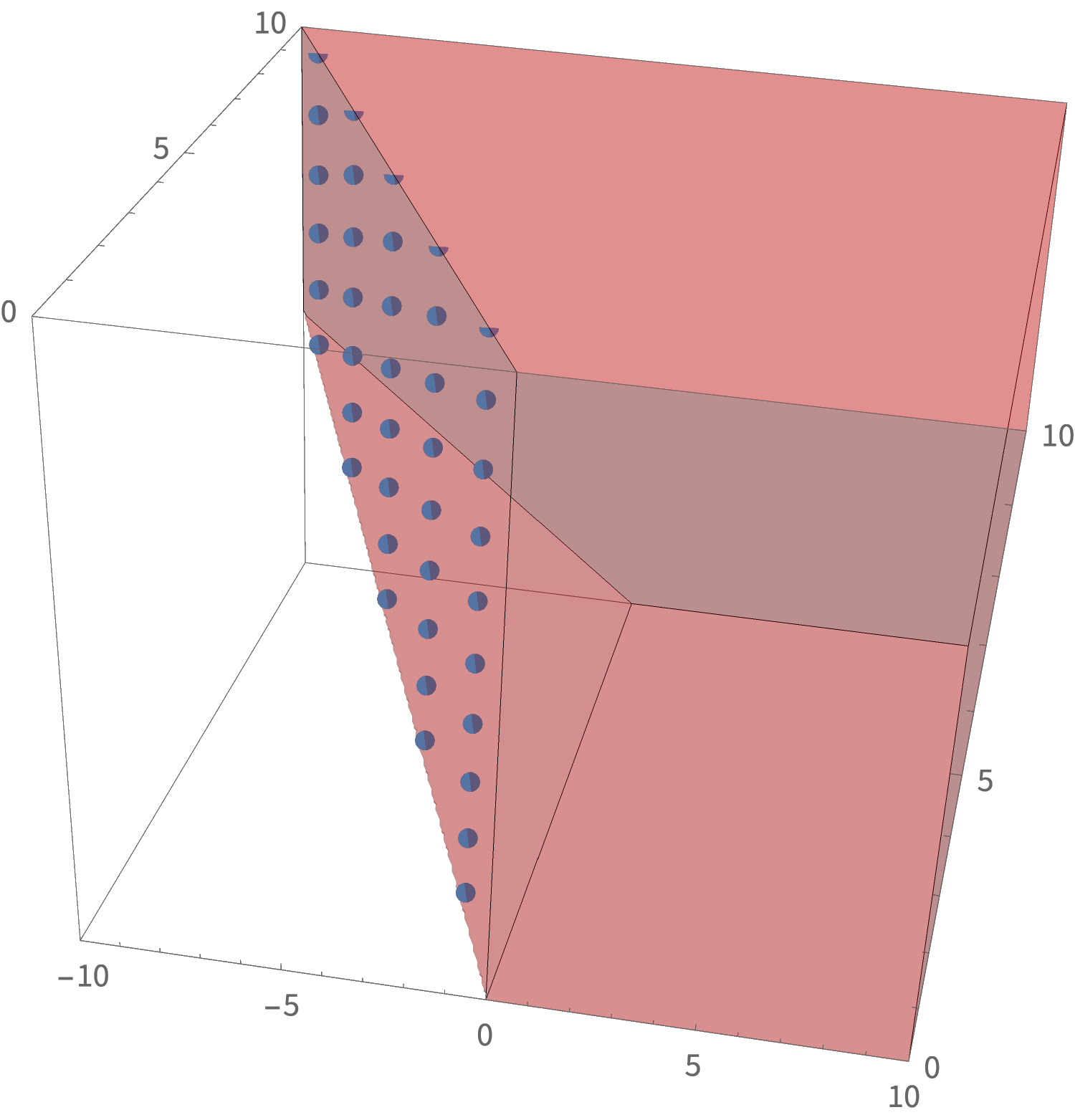}
    \begin{picture}(0,0)\vspace*{-1.2cm}
        \put(-125,0){\footnotesize $p_1$}
        \put(-215,200){\footnotesize $p_2$}
        \put(-240,90){\footnotesize $p_3$}
    \end{picture}\vspace*{0.5cm}
    \caption{The toric effective cone $\mathcal{E}_{V_{3,51}}$ for the Calabi--Yau hypersurface $X_{3,51}$ (red) with holes in blue. The variables $p_1, p_2, p_3$ denote the components of divisor classes when expressed in the basis of \cref{eq:351_charge_matrix}.}
    \label{fig:effective cone X351}
\end{figure}
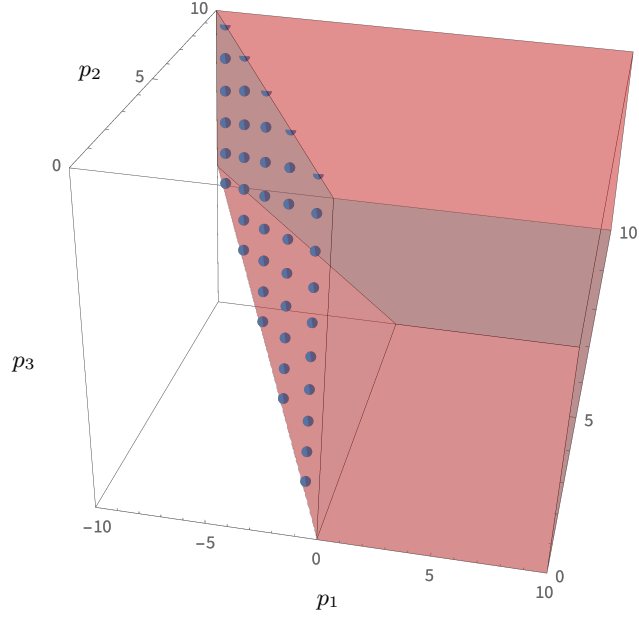

The non-trivial Hilbert basis element $[D] = (-1,1,1)$ has the zeroth cohomology
\begin{equation}
    h^{0}(X_{3,51},\mathcal{O}_{X_{3,51}}([D]))=0\,.
\end{equation}
Therefore, this divisor $[D]$ is a hole. By computing $h^0$ for more divisors in the effective cone, we have found more non-effective divisors living on the boundary of the effective cone. This has been illustrated in \cref{fig:effective cone X351}. As can be seen, the general pattern for non-effective holes are 
\begin{equation}
    [D_{m,n}] =\begin{pmatrix}
        -(2n+1)\\
        2n+1\\
        m
    \end{pmatrix}\,,
\end{equation}
where $n,m\in\mathbb{Z}_+$. The other sites along this boundary are effective as they are formed via non-negative combinations of effective divisors, i.e.,
\begin{equation}
    [D^{\rm eff}_{m,n}] = n\cdot\begin{pmatrix}
        -2\\
        2\\
        1
    \end{pmatrix}+(m-n)\cdot\begin{pmatrix}
        0\\
        0\\
        1
    \end{pmatrix}\,.
\end{equation}

\subsection{Holes in Interiors of Toric Effective Cones}

We've now seen that non-trivial Hilbert basis elements for Kreuzer--Skarke geometries with $h^{1,1} \leq 3$ always belong to the boundary of the effective cone, and in the examples we've studied, the associated semigroup predicted by \cref{cor:semigroup} has been contained in that boundary, as anticipated by the discussion at the end of \cref{sec:minimal_model}. We are now motivated to study holes on the interior of the toric effective cone. Can non-trivial Hilbert basis elements on the boundary give rise to semigroups which include holes on the interior? And, by going to larger $h^{1,1}$, can we find non-trivial Hilbert basis elements in the interior? 

To understand these questions, we study two examples with holes on the interior of the toric effective cone. The first involves a collection of holes populating both the boundary and the interior which naively all appear to arise from a boundary non-trivial Hilbert basis element: after studying the birational geometry of this example, we will interpret these holes as decomposing into a countable disjoint union of semigroups, with each semigroup contained either strictly on the boundary or strictly within the interior. This substantiates our expectation from \cref{sec:minimal_model} that the semigroup of non-movable holes in faces of $\partial \mathcal{E}_V$ remain in those faces. The second example will feature a non-trivial Hilbert basis element which itself lies on the interior of the toric effective cone: in studying its birational geometry, we will encounter the weighted projective space hypersurface we first introduced at the end of \cref{sec:minimal_model}.

\subsubsection{Example 3: $h^{1,1}=3$.}
\label{sec:h113 interior}

\begin{figure}[tp!]
    \centering
    \includegraphics[width=0.5\linewidth]{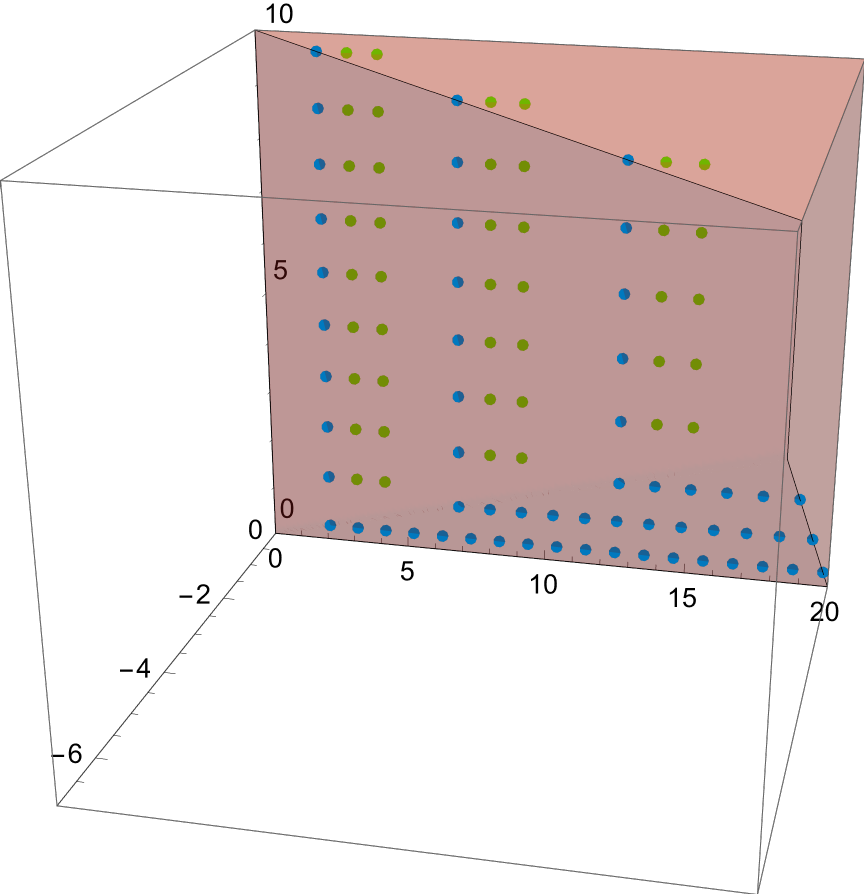}
    \begin{picture}(0,0)\vspace*{-1.2cm}
        \put(-225,10){\footnotesize $p_1$}
        \put(5,75){\footnotesize $p_2$}
        \put(-165,240){\footnotesize $p_3$}
    \end{picture}\vspace*{0.5cm}
    \caption{The toric effective cone $\mathcal{E}_{V_{3,165}}$ for the Calabi--Yau hypersurface $X_{3,165}$ (red) with holes on the boundary and on the interior of $\mathcal{E}_{V_{3,165}}$ in blue and green, respectively. The variables $p_1, p_2, p_3$ denote the components of divisor classes when expressed in the basis of \cref{eq:GLSM_3_165}.}
    \label{fig:effective cone X3165}
\end{figure}
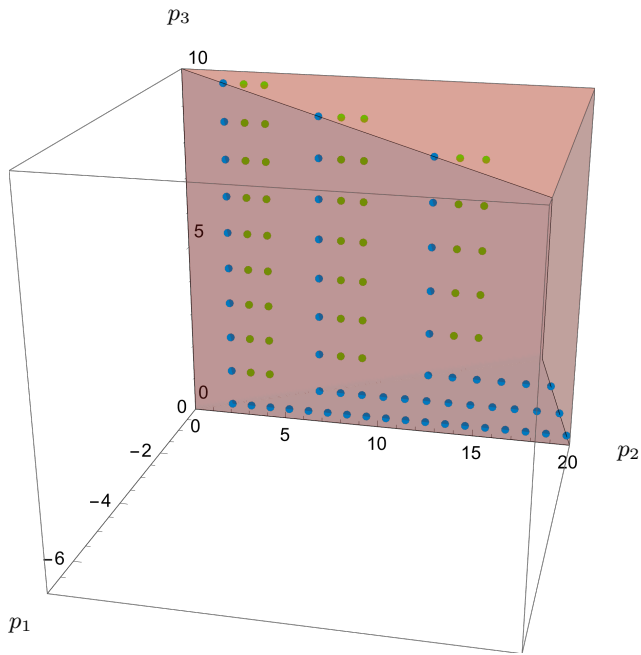

Consider the Calabi--Yau threefold which has index 232 with $(h^{1,1},h^{2,1})=(3,165)$, namely the last entry in \cref{tab:h11=3}. We will refer to this Calabi--Yau threefold as $X_{3,165}$ and its ambient space as $V_{3,165}$. For this example, the Calabi--Yau threefold is obtained as a hypersurface in a toric variety from an FRST of the reflexive polytope with lattice points
\begin{equation}
    \begin{pmatrix}
        -9&0&0&1&1&1&-3\\
        -6&0&0&0&2&1&-2\\
        -1&0&1&0&0&0&0\\
        -1&1&0&0&0&0&0
    \end{pmatrix}\,.
\end{equation}
We can choose the following GLSM charge matrix for this toric fourfold.
\begin{equation}
    \begin{pmatrix}
        \label{eq:GLSM_3_165}
        1&1&1&3&0&6&0\\
        0&0&0&1&1&-2&0\\
        0&0&0&1&0&2&1
    \end{pmatrix}\,.
\end{equation}
The toric effective cone of divisors has extremal ray generators given by the columns of the following matrix
\begin{equation}
    \mathcal{E}_{V_{3,165}} = \mathrm{cone}
    \begin{pmatrix}
        1&0&0&3\\
        0&1&0&-1\\
        0&0&1&1
    \end{pmatrix}\,.
\end{equation}
The Hilbert basis of the toric effective cone is just given by the extremal rays:
\begin{equation}
    \mathcal{H}(\mathcal{E}_{V_{3,165}})=\left\{
    \begin{pmatrix}
        1\\
        0\\
        0
    \end{pmatrix}\,,
    \begin{pmatrix}
        0\\
        1\\
        0
    \end{pmatrix}\,,
    \begin{pmatrix}
        0\\
        0\\
        1
    \end{pmatrix}\,,
    \begin{pmatrix}
        3\\
        -1\\
        1
    \end{pmatrix}
    \right\}\,.
\end{equation}
The non-trivial Hilbert basis element is $[D]=(3,-1,1)$ and cannot be generated by a non-negative integral combination of the prime toric divisors, namely 
\begin{equation}
    h^0(X_{3,165},\cO_{X_{2,165}}(D))=0\,.
\end{equation}
Furthermore, there are semigroups of holes in the toric effective cone of $X_{3,232}$ as illustrated by \cref{fig:effective cone X3165}. The holes in the toric effective cone can be generated as 
\begin{equation}
    \label{eq:h11_3_interior_holes}
    [D_{n,m,l}]=[D]+n
    \begin{pmatrix}
        1\\
        0\\
        0
    \end{pmatrix}+m\begin{pmatrix}
        6\\
        -2\\
        2
    \end{pmatrix}
    +l\begin{pmatrix}
        0\\
        0\\
        1
    \end{pmatrix}\,,
\end{equation}
where either $l = 0$ and $n,m \geq 0$ or $n=0,1,2$ and $m,l\geq 0$.
In particular, all divisor classes within this collection of holes with $n>0$ and $l > 0$ are in the interior of the toric effective cone, and are colored in green. However, the case with $n=0$ becomes the familiar codimension one semigroup of holes that we have found in the previous section. Additionally, there is another codimension two semigroup with $m=l=0$ that is along the boundary of the toric effective cone. 

This is the most complicated collection of holes we've seen, and it is worth repeating the analysis performed in \cref{sec:h11=2} to interpret the holes using the minimal model program. In particular, we will find that the holes in this example naturally decompose into countably many semigroups of the kind we are accustomed to: namely, a hole $[D]$ plus a non-negative linear combination of prime divisors yielding the stable base locus or, equivalently, divisors contracted in passing to the $[D]$ minimal model. This is summarized in \cref{fig:eff_cone_h11_3_MMP}, where we plot a two-dimensional slice of the effective cone and label the semigroups of holes.\footnote{One can alternatively think of this plot as the entire effective cone of $X_2$, with the points being the pushforwards of holes to $X_1$. The bottom cone is then the nef cone of $X_2$.} 

We will need to discuss two divisorial contractions of $X_{3,232}$: $X_2$, with Picard number $2$, achieved by contracting $[D_6] = 2[D]$, and $X_1$, with Picard number one, achieved by contracting $[D_6]$ and $[D_7]$. These contractions descend from the ambient variety: in particular, let us introduce the charge matrices $Q_1, Q_2$ of their ambient varieties now.
\begin{align}
    Q_2 &= \left(\begin{array}{cccccc}
        1 & 1 & 1 & 6 & 3 & 0 \\ 0 & 0 & 0 & 2 & 1 & 1 \\ \hline 0 & 0 & 0 & 1 & 1 & 0
    \end{array}\right), \\
    Q_1 &= \left(\begin{array}{cccccc}
        1 & 1 & 1 & 6 & 3 \\ \hline 0 & 0 & 0 & 1 & 1
    \end{array}\right), 
\end{align}
We use the horizontal line to separate $\mathbb{Z}$-valued and torsional charges for toric varieties with class group torsion, as introduced and discussed in the torsional example of \cref{sec:minimal_model}. Here, one can see that the contracted prime toric divisors are removed, and otherwise we have preserved the order of the prime toric divisors from \cref{eq:GLSM_3_165}. In this case, the contractions can be computed on the toric variety, 

The first birational map resembles what we saw in \cref{sec:h11=2}, in that the contracted divisor generates a non-saturated subgroup with a hole given by half the class. What is unique in this case is that instead of there being a unique nef hole on the image of the map, now we have countably infinitely many holes. Some of these holes are nef, while some are not: for those of the latter type, we will still need to contract $[D_7]$ to arrive at their minimal model.

In particular, we can compute that the nef cone of $X_2$ in this basis is generated by $(1,0)$ and $(3,1)$ (omitting the torsional part of these classes because the nef cone lies in the class group tensored with $\mathbb{R}$). The holes on $X_2$ are $(0,0,1) + a(1,0,0) + b(0,1,0)$ for $b = 0$ or $0 \leq a < 3$. At the level of the toric variety, it is clear that these classes have no global sections (they cannot be expressed as a linear combination of prime toric divisors), and this persists on the Calabi--Yau hypersurface. Those of the former type are nef on $X_2$, and each one induces a one-dimensional semigroup on $X_{3,232}$ generated by the stable base locus/contracted divisor $[D_6]$. We can collect these semigroups together into a two-dimensional semigroup by writing
\begin{equation}
    [D] + n[D_1] + m[D_6]
\end{equation}
for $n, l \geq 0$, reproducing the first case described after \cref{eq:h11_3_interior_holes}. The first two terms parameterize the nef holes on $X_2$, while the last term parameterizes their respective semigroups, generated by $[D_6]$. These are the holes on the bottom face of \cref{fig:effective cone X3165} and, analogously, corresponding to the bottom chamber of \cref{fig:eff_cone_h11_3_MMP}.

Now let us study the latter case, $0 \leq a < 3$ for $b > 0$. These holes are not nef on $X_2$, so we contract $[D_7]$ to achieve $X_1$. The effective cone and nef cone of $X_2$ are both just $\mathbb{R}_{\geq 0}$: in particular, all of the holes are nef. We can also immediately read off that the holes are $(a,1)$ for $a < 3$, as the smallest prime toric divisor with torsional part has free part $3$. Each one of the three nef holes results in a distinct two-dimensional semigroup on $X$, parameterized by $n \in \{0, 1, 2\}$ as follows.
\begin{equation}
    [D] + n[D_1] + m[D_6] + l[D_7]
\end{equation}
Here, $l, m \geq 0$. These are the holes on the upper face of \cref{fig:effective cone X3165} and, analogously, corresponding to the top chamber of \cref{fig:eff_cone_h11_3_MMP}.

\begin{figure}
    \centering
    \includegraphics[width=0.7\linewidth]{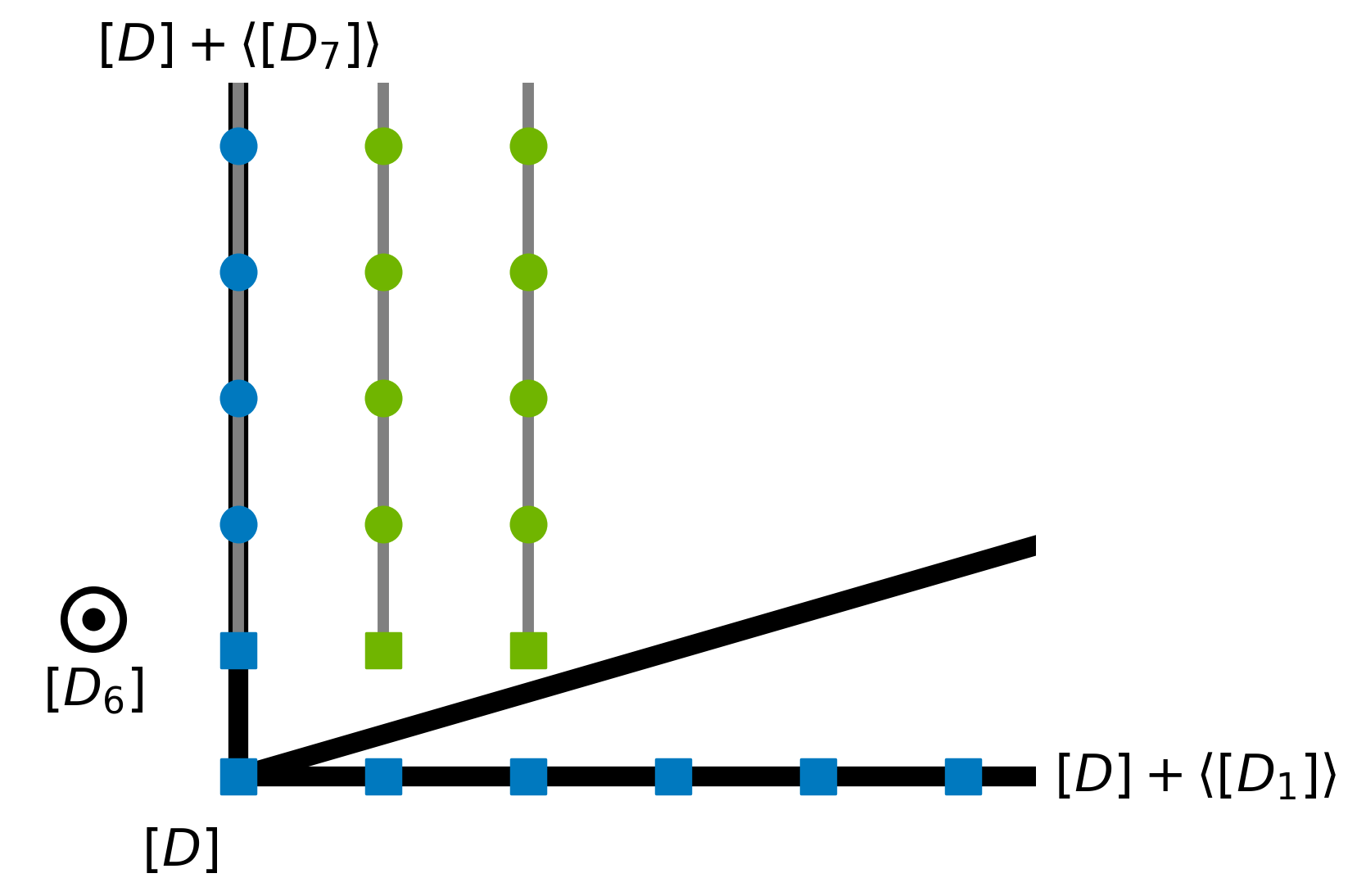}
    \caption{The semigroup structure of holes in $X_{3,232}$. We plot the slice of the effective cone given by $[D]$ plus the non-negative span of $[D_1]$ and $[D_7]$, with $[D_6]$ being the direction out of the page. All plotted lattice points are holes, and the coloring matches that of \cref{fig:effective cone X3165}. Squares denote the ``origins'' of the distinct semigroups, and gray lines connect holes belonging to the same semigroup. The black line separates holes with minimal model given by $X_2$ (below) and by $X_1$ (above). $[D_6]$ is always a semigroup generator, so each semigroup extends out of the page. We have a countable family of one-dimensional semigroups whose minimal model is $X_2$ --- corresponding to a countable family of nef holes on $X_2$ --- each generated by the exceptional divisor $[D_6]$. Additionally, we have three two-dimensional semigroups whose minimal model is $X_1$ --- corresponding to the three nef holes on $X_1$ --- each generated by the exceptional divisors $[D_6, D_7]$. Every semigroup is  strictly contained either in the boundary of the effective cone or in the interior.}
    \label{fig:eff_cone_h11_3_MMP}
\end{figure}

We have thus accounted for all the holes we originally identified in \cref{eq:h11_3_interior_holes}: they arise from an infinite number of one-dimensional semigroups --- organizable into a two-dimensional semigroup --- corresponding to nef holes on $X_2$, together with three two-dimensional semigroups corresponding to nef holes on $X_1$.

There is one final comment worth making about this example. The effective cone $\mathcal{E}_{X_{3,165}}$ admits an automorphism given, in our basis, by permutation of the final two coordinates. This is highly non-trivial, as both toric approximation to the nef cone of $X_{3,165}$ and the toric part $\mathcal{E}_{V_{3,165}}$ of the Calabi--Yau effective cone are not closed under the symmetry, meaning that there are ``non-toric'' regions in both the nef and effective cones of $X_{3,165}$. Using the methods of \cite{Gendler:2022ztv}, one can show that $\mathcal{E}_{X_{3,165}}$ is actually the cone generated by $\mathcal{E}_{V_{3,165}}$ as well as the autochthonous divisor $(6,2,-2)$ that is the image under the symmetry of the prime toric divisor $[D_6]$. In particular, then, the non-inherited region of the effective cone $\mathcal{E}_{X_{3,165}} \setminus \mathcal{E}_{V_{3,165}}$ contained the image of the family of holes we have just identified under the automorphism. These are our first example of holes outside of the toric effective cone $\mathcal{E}_V$.

Therefore, in this example, we can see holes can start to appear in the interior. In the following example, we will see that as the non-trivial Hilbert basis is in the interior, the semigroup of holes will then be entirely in the interior of the toric effective cone.

\begin{table}
    \centering
    \begin{tabular}{c|c|c}
    ID & $(h^{1,1},h^{2,1})$ & $[D]$ \\
    \hline
    851 & (4,94) & $(1, 0, -2, -1)$\\
    1056 & (4,112) & $(-1, 1, 1, 0), (5, -1, -1, 2)$\\
    1059 & (4,112) & $(-1, 1, 1, 0)$\\
    1061 & (4,112) & $(-1, 1, 1, 0)$\\
    \end{tabular}
    \caption{Non-trivial Hilbert basis elements in the Kreuzer--Skarke database at $h^{1,1}=4$ which are interior to the toric effective cone. These are holes as a consequence of \cref{prop:big}, which we verify with an explicitly computation using the methods of \cref{sec:direct_comp}. All are non-generic in complex structure.
    }
    \label{tab:h11=4}
\end{table}

\subsubsection{Example 4: $h^{1,1}=4$.} \label{sec:h11=4} 

At $h^{1,1}=4$, we encounter our first example in the Kreuzer--Skarke database where there is a non-trivial Hilbert basis element in the strict interior of $\mathcal{E}_V$ (in the previous example, the interior holes were not themselves non-trivial Hilbert basis elements). The Calabi--Yau threefolds with $h^{1,1}=4$ that exhibit this property have been listed in \cref{tab:h11=4}. Such divisors are guaranteed to be holes by \cref{prop:big}.

Let us consider the first example in \cref{tab:h11=4}. Here, the Calabi--Yau threefold obtained as a hypersurface in a toric variety from an FRST of the reflexive polytope with lattice points
\begin{equation}
    \begin{pmatrix}
        1&0&1&-3&-1&-1&0&0\\
        0&0&1&-2&1&-1&0&1\\
        0&0&2&0&-1&0&1&0\\
        0&1&-1&0&0&0&0&0
    \end{pmatrix}\,.
\end{equation}
We can choose the following GLSM charge matrix for this toric fourfold.
\begin{equation}
\label{eq:x494 GLSM}
\begin{pmatrix}
    1&0&0&0&0&1&0&1\\
    0&1&1&0&0&1&-2&0\\
    0&0&0&1&0&-3&0&-1\\
    0&0&0&0&1&-1&1&-2
\end{pmatrix}\,.
\end{equation}
This Calabi--Yau threefold has the following Hodge numbers $(h^{1,1},h^{1,2})=(4,94)$. Therefore, similar to before, we will refer to it as $X_{4,94}$, with its ambient space being $V_{4,94}$. 

Now, the effective cone of divisors has extremal ray generators given by the columns of the following matrix
\begin{equation}
    \mathcal{E}_{V_{4,94}} = \mathrm{cone}
    \begin{pmatrix}
        0&0&0&1&1\\
        -2&0&1&0&1\\
        0&1&0&-1&-3\\
        1&0&0&-2&-1
    \end{pmatrix}\,.
\end{equation}
The Hilbert basis of this effective cone is
\begin{equation}
    \mathcal{H}(\mathcal{E}_{V_{4,94}}) =\left\{
    \begin{pmatrix}
        0\\
        -2\\
        0\\
        1
    \end{pmatrix}\,,
    \begin{pmatrix}
        0\\
        0\\
        1\\
        0
    \end{pmatrix}\,,
    \begin{pmatrix}
        0\\
        1\\
        0\\
        0
    \end{pmatrix}\,,
    \begin{pmatrix}
        1\\
        0\\
        -2\\
        -1
    \end{pmatrix}\,,
    \begin{pmatrix}
        1\\
        0\\
        -1\\
        -2
    \end{pmatrix}\,,
    \begin{pmatrix}
        1\\
        1\\
        -3\\
        -1
    \end{pmatrix}
    \right\}\,.
\end{equation}
Similar to before, among the above Hilbert basis elements, we have the non-trivial Hilbert basis element $[D] = (1,0,-2,-1)$ which cannot be generated by the non-negative linear integral combinations of prime toric divisors. Indeed, the zero cohomology associated to this additional generator vanishes
\begin{equation}
    h^0(X_{4,94},\mathcal{O}_{X_{4,94}}(D))=0\,,
\end{equation}
implying its non-effectiveness. Furthermore, within $X_{4,94}$, the self-intersection of this non-trivial Hilbert basis element is $-23$. We can also compute the associated supersymmetric index (the number of real moduli of the divisor~\cite{Maldacena:1997de}) which is $d_p=-2$. Lastly, its Euler characteristic and signature are $\chi=23$ and $\sigma=-23$, respectively. Therefore, this non-trivial Hilbert basis element is indeed not nef.

From here, we can also find a plethora of additional non-effective divisors. However, the crucial difference is that for all the found holes in this example, the non-effective divisors lies strictly in the interior of the effective cone. Similar to before, the holes form a codimension-1 sub-lattice within the effective cone. This sub-lattice of non-effective divisors are generated by linear positive integral combinations of the following three lattice sites (starting with the additional generator appearing in the Hilbert basis as the origin of the sub-lattice)
\begin{equation}
\label{eq:h114-non-effective}
    [D_{n,m,\ell}] = [D] +
    n\begin{pmatrix}
        0\\
        -2\\
        0\\
        1
    \end{pmatrix}
    +m\begin{pmatrix}
        1\\
        0\\
        -1\\
        -2
    \end{pmatrix}
    +\ell\begin{pmatrix}
        1\\
        1\\
        -3\\
        -1
    \end{pmatrix}\,,
\end{equation}
where $n,m,l\in\mathbb{Z}_{\geq 0}$. Furthermore, as the above generators in \eqref{eq:h114-non-effective} all lie on the boundary of the effective cone, upon the ``shift'' by $[D]$, we have every site that lies within this semigroup is indeed strictly in the interior of the effective cone.

It is useful to understand this example too from the perspective of \cref{prop:birational_preserve}. In particular, let $[D]$ be the hole descending from $[\hat{D}]$: a $[\hat{D}]$-minimal model for the ambient variety is achieved by successively contracting the three prime toric divisors given in our basis as $(0, -2, 0, 1)$, $(1, 0, -1, -2)$, and $(1, 1, -3, -1)$. The composition of these divisorial contractions results in the $[\hat{D}]$-minimal model $\mathbb{P}_{2,2,3,4,11}$, the singular four-dimensional toric weighted projective space we introduced in \cref{sec:friends}. The divisorial contractions of the hypersurface are also inherited from the toric ones, and result in an anticanonical hypersurface $X'$ in $\mathbb{P}_{2,2,3,4,11}$. In this way, we have implemented the mechanism discussed in our second example from \cref{sec:friends} in an explicit smooth Calabi--Yau threefold $X_{4,94}$. In particular, in light of \cref{prop:birational_preserve}, this hole can be thought of as the result of having a singular birationally related geometry $X'$ which has a nef hole: in this case, an ample hole, namely the restriction of the hyperplane class from $\mathbb{P}_{2,2,3,4,11}$ which is not effective (there is no homogeneous coordinate with that degree). That is, the semigroup of holes which feature in this example all pushforward to the non-effective, nef divisor on $X'$. This is yet another explicit realization of \cref{cor:semigroup}: because the three contracted prime toric divisors lie in the stable base locus of $[D]$ (as well as the kernel of the pushforward of divisor classes to $X'$), the number of global sections of $[D]$ plus any non-negative linear combination of these exceptional prime toric divisors will be the same as the number of global sections as the hole $D$ itself. This explains why this codimension-one semigroup of divisors all constitute holes. 

\subsection{Statistics of holes when $h^{1,1} \leq 17$}

\begin{table}
    \centering
    \begin{tabular}{c|c|c}
    ID & $(h^{1,1},h^{2,1})$ & $[D]$ \\ 
    \hline
    883 & (5,57) & $(-3, 2, -2, 1, -1)$ \\ 
    3004 & (5,77) & $(1, -2, -2, 0, -1)$ \\ 
    3024 & (5,77) & $(1, 0, -2, -1, -1)$ \\ 
    4193 & (5,95) & $(1, -1, -2, -1, 1)$ \\ 
    \end{tabular}
    \caption{Non-trivial Hilbert basis elements in the Kreuzer--Skarke database at $h^{1,1}=5$. These are holes as a consequence of \cref{prop:big}, which we verify with an explicitly computation using the methods of \cref{sec:direct_comp}. All are non-generic in complex structure.
    }
    \label{tab:h11=5}
\end{table}

\begin{table}
    \centering
    \begin{tabular}{c|c|c}
    ID & $(h^{1,1},h^{2,1})$ & $[D]$ \\
    \hline
    354 & (6,42) & $(1, 1, -1, 1, 1, 1)$ \\ 
    3178 & (6,52) & $(1, -2, -3, -2, -2, -1)$ \\ 
    6636 & (6,60) & $(1, 0, -2, -1, -2, -1)$ \\
    7120 & (6,60) & $(1, -2, -1, -1, 1, -2)$ \\ 
    8126 & (6,64) & $(1, 0, -2, -2, -1, -1)$ \\ 
    8573 & (6,66) & $(-1, 0, 5, 2, 3, 2)$ \\
    9105 & (6,66) & $(1, -2, -2, 0, -2, -1)$ \\
    9129 & (6,60) & $(-1, 1, 1, 0, 1, 0)$ \\
    12486 & (6,78) & $(1, -1, -2, -1, -1, 1)$ \\ 
    \end{tabular}
    \caption{Non-trivial Hilbert basis elements in the Kreuzer--Skarke database at $h^{1,1}=6$. In particular, These are holes as a consequence of \cref{prop:big}, which we verify with an explicitly computation using the methods of \cref{sec:direct_comp}. All are non-generic in complex structure.
    }
    \label{tab:h11=6}
\end{table}

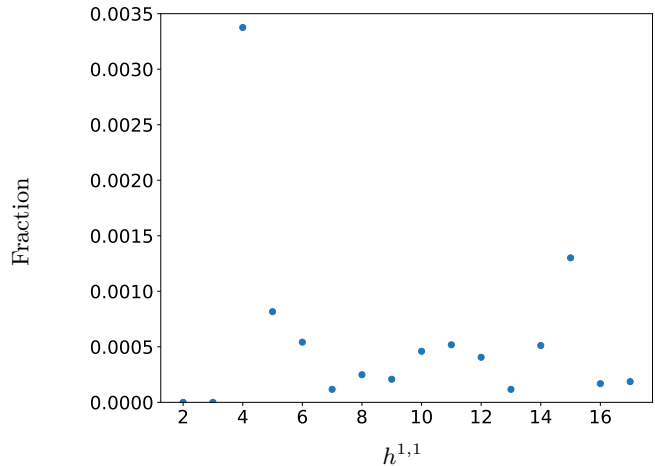
\begin{figure}
    \centering
    \includegraphics[width=0.5\linewidth]{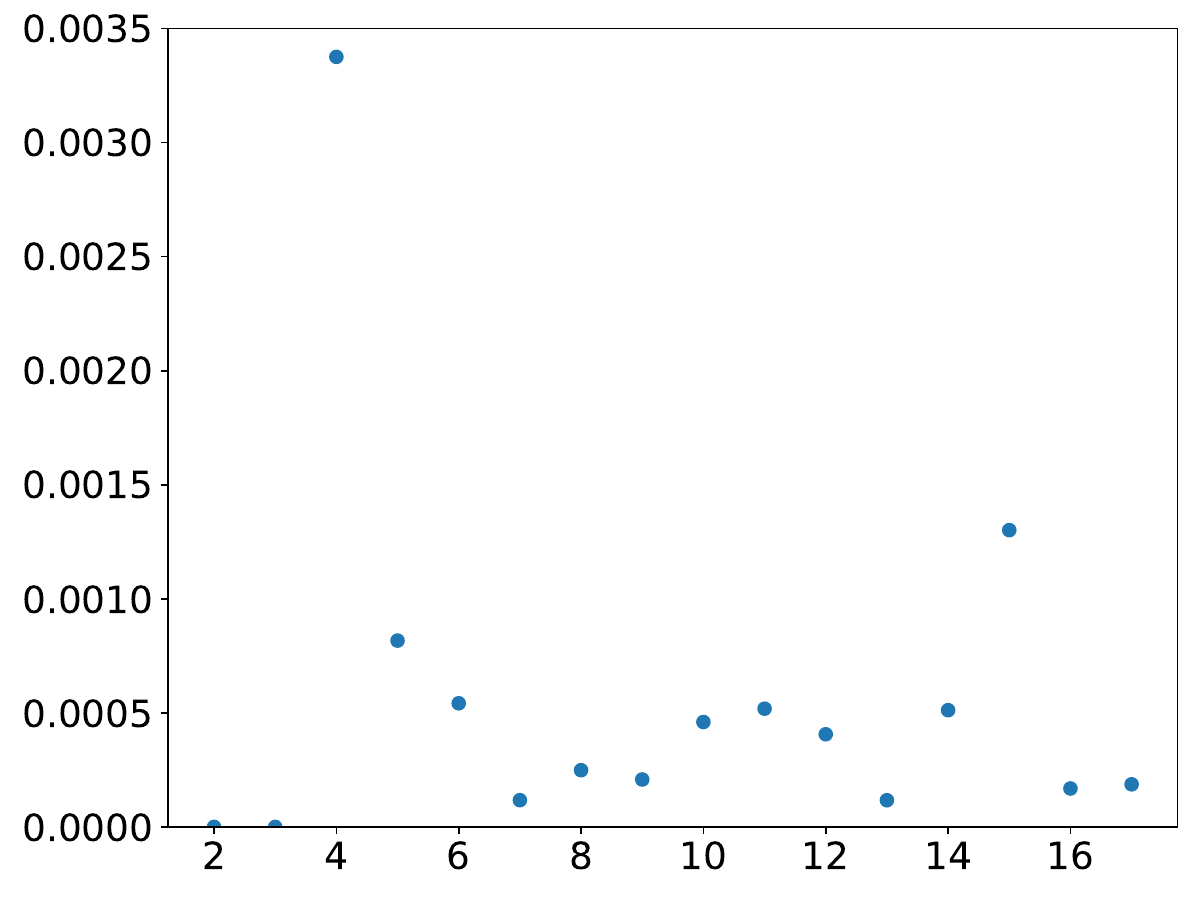}
    \begin{picture}(0,0)\vspace*{-1.2cm}
    \put(-110,-10){\footnotesize $h^{1,1}$}
    \put(-250,65){\rotatebox{90}{\footnotesize Fraction}}
    \end{picture}\vspace{0.5cm}
    \caption{The fraction of birational classes of toric hypersurface Calabi--Yau threefolds with a non-trivial Hilbert basis element in the strict interior of the toric effective cone up to $h^{1,1}=17$.}
    \label{fig:interior fraction}
\end{figure}

We have conducted a complete classification of non-trivial Hilbert basis elements appearing in the (toric) effective cone of Calabi--Yau threefolds with $h^{1,1}\leq 6$ constructed from the Kreuzer--Skarke database. There is only one geometry within the database at $h^{1,1}=2$ with a hole and this has been analyzed in \cref{sec:h11=2}. We presented the entire list of Calabi--Yau threefolds in the Kreuzer--Skarke database with holes at $h^{1,1}=3$ in \cref{tab:h11=3}. For $h^{1,1}=4$, we have a total of 88 Calabi--Yau threefolds with holes in the database. Instead of exhaustively listing all 88 examples, we have chosen to list only those with holes in the strict interior of the toric effective cone, which only occurs $h^{1,1}\geq 4$ in the Kreuzer--Skarke database. This is shown in \cref{tab:h11=4}. We can continue to push our algorithm to perform searches within the database and have done so up to $h^{1,1}=6$. The results for $h^{1,1}=5$ and $h^{1,1}=6$ are shown in \cref{tab:h11=5} and \cref{tab:h11=6}, respectively.

As the number of distinct polytopes increases drastically beyond $h^{1,1}=6$, we randomly sample $\sim 10^4$ polytopes for a given $h^{1,1}>6$ from the Kreuzer--Skarke database in \verb+CYTools+~\cite{Demirtas:2022hqf} and analyze the number of birational classes of Calabi--Yau threefolds that have a non-trivial Hilbert basis element in the interior of the toric effective cone. The fraction of such birational classes with this property among the $\sim 10^4$ sampled birational classes for a given $h^{1,1}$ up to 15 are shown in \cref{fig:interior fraction}. Consistent with the classification and observations we performed at small $h^{1,1}$ ($h^{1,1}\leq 6$), the fraction of these birational classes with this behavior is small compared to the more general cases where these additional Hilbert basis elements lie on the boundary of the toric effective cone (cf. \cref{fig:additional elements}). Nevertheless, all such additional divisors are non-effective in both the toric variety and the Calabi--Yau threefold. We would like to highlight that our algorithm for determining the effectiveness of divisors remains efficiently functional even at $h^{1,1}\sim \cO(10)$. This is largely due to the observation that all holes we encounter exhibit the property $H^1(V,\cO_V(\hat{D}+K_V))=0$. Hence, the only relevant piece is determining if $h^0(V,\cO_V(\hat{D}))\neq 0$, or in other words if a solution to $Q_V\cdot a=[\hat{D}]$ exist. 

For every hole we have examined (divisor classes with components less than $10$ in the chosen basis) appearing in Calabi--Yau threefolds within the Kreuzer--Skarke database with $h^{1,1}\leq 6$, we have found all of them to be non-effective. For those on the interior of $\mathcal{E}_V$, this is a verification of \cref{prop:big}, while for those on the boundary of $\mathcal{E}_V$ this is an important experimental result.

At last, as shown in \cref{sec:friends}, all holes come in semigroups. An interesting property of all the families is that they are mostly of codimension one in the toric effective cone. In particular, all such families of holes in the toric effective cones of Calabi--Yau threefolds with $h^{1,1}\leq 3$ are of codimension one, while for $h^{1,1}=4,5,6$, there are more than $90\%$ of birational classes of Calabi--Yau threefolds with codimension one semigroups of holes.
Additionally, for a codimenion-$n$ sublattice of holes, the non-holomorphic divisor classes can always be expressed in terms the following generators
\begin{equation}
    [D] + \sum_{i=1}^{h^{1,1}-n} a_i [D_i] \,,
\end{equation}
where $[D]$ indicates the non-trivial Hilbert basis element of the toric effective cone while the $[D_i]$'s indicate prime toric divisors. In light of \cref{cor:general_sublattice} and \cref{cor:semigroup}, it is natural to guess that at least some of the generators of a semigroup of holes with base $D$ will be a subset of the prime toric divisor classes. This is because such generators are furnished by the exceptional divisors of the birational map from the original Calabi--Yau $X$ to its $D$-minimal model, and often (though not always) these maps are restrictions of maps on the ambient space. It then suffices to note that the exceptional divisors of a toric birational map are always prime toric divisors $\hat{D}_i$. Thus, hole semigroup generators will be furnished by prime toric divisor classes in cases when the Calabi--Yau birational geometry is inherited from the ambient toric variety. The semigroups are predicted to be codimension one when there are $h^{1,1} - 1$ prime divisors in the stable base locus of $[D]$, or passing to the $[D]$-minimal model contracts $h^{1,1} - 1$ prime divisors (resulting in a variety with $h^{1,1} = 1$). However, as we saw in \cref{sec:h113 interior}, semigroups can appear to be larger when there are families of nef holes on a birational model, and this effect likely contributes in many cases.

\section{Bounds on Volumes of Holes} \label{sec:vol_bounds}
In this section, we will explain how to place upper and lower bounds on the volumes of non-holomorphic divisor classes (defined below), with specific attention towards volumes of holes. We will then give examples of explicit geometries where we implement these bounds, providing robust inequalities on the volume of a hole as a function of the K\"ahler moduli.

Consider a non-holomorphic divisor class $[D]$. An important physical quantity is $\text{vol}([D]) := \text{vol}(D_{\text{min}})$, the volume of the volume-minimizing representative $D_{\text{min}}$ of $[D]$. A formula for $\text{vol}(D_{\text{min}})$ can be simply stated: 
\begin{align}
    \text{vol}(D_{\text{min}}) = \int_{D_{\text{min}}} \sqrt{\text{det}(g_{D_{\text{min}}})}\,,
    \label{eq:metricvol}
\end{align}
where $g_D$ is the induced metric on $D$. In practice, evaluating this expression is infeasible: it requires knowing both the Calabi--Yau metric, as well as the volume-minimizing representative $D$, which requires optimization of the expression \eqref{eq:metricvol} over different representatives of $[D]$. We refer to $\text{vol}(D_{\text{min}})$ as the \textit{volume} of $[D]$. Note that this terminology is in concord with \eqref{eq:calibrated_vols}, our earlier expression for the volume of a holomorphic divisor class.

Although $\text{vol}(D_{\text{min}})$ is not computable in practice, we can bound this quantity from above and below. A lower bound is given by
\begin{align}
    \text{vol}(D_{\text{min}}) \geq \frac{1}{2} \int_{[D]} J \wedge J \,,
\end{align}
i.e., the volume of the minimum-volume representative of any non-holomorphic divisor class is at least as big as the calibrated volume. This has the physics interpretation of a BPS bound: the tension of a non-BPS string coming from wrapping an M5-brane on $[D]$ is at least as big as the would-be tension of a BPS object with the same charge (i.e., divisor class). Note that this bound has varying utility depending on the charge of the non-holomorphic divisor: for some charges, this bound can even be trivial.

Additionally, an upper bound on $\text{vol}(D_{\text{min}})$ is provided by any representative of $[D]$ whose volume is known. A straightforward set of such representatives is known as the set of \textit{piecewise-calibrated representatives}, introduced in~\cite{Demirtas:2019lfi}, and defined as follows:
\begin{align}
    D_{p.c.} = \sum_{I} c_I D_I \ \ \ \text{with} \ \ \ [D_{p.c.}] = [D]\,,
\end{align}
where $D_I$ are holomorphic divisors and the $c_I$ are integers (positive or negative). Intuitively, this means that any divisor can be written as a union of holomorphic and anti-holomorphic components. 

Since the volumes of the holomorphic divisor classes are known, the volume of the piecewise-calibrated representative $D_{p.c.}$ is
\begin{align}
    \text{vol}(D_{p.c.}) = \frac{1}{2}\sum_{I} |c_I| \int_{[D_I]} J \wedge J\,.
\end{align}
For any choices of $c_I$ such that $[D_{p.c.}] = [D]$, this provides an upper bound on the volume of $D_{\text{min}}$:
\begin{align}
    \text{vol}(D_{\text{min}}) \leq \text{vol}(D_{p.c.})\,.
\end{align}
Here, the upper bound can be understood as imposing energy conservation in this BPS and anti-BPS system. The inequality is saturated when the resulting object in the lower-dimensional theory can decay entirely into widely separated BPS and anti-BPS components, or when the non-BPS string is stable.

The \textit{best} upper-bound is then found by identifying the piecewise-calibrated representative with the smallest volume. We refer to this minimum-volume piecewise-calibrated representative as $D_{\text{p.c.}}^{\text{min}}$. For a general non-effective divisor class $[D]$, optimizing the decomposition $[D] = [A] - [B]$ such that $[A]$ and $[B]$ are effective and $\text{vol}([A]) + \text{vol}([B])$ is minimized is difficult. In our case, if $[B]$ lies in the effective cone, so does $[A] = [D] + [B]$, so we are really optimizing $[B]$ in the effective cone such that neither $[B]$ nor $[D] + [B]$ is a hole. The best upper bound will then be
\begin{equation}
    \label{eq:upper_bound_theory}
    \text{vol}([D]) \leq \text{vol}(D_{\text{p.c.}}^{\text{min}}) = \text{vol}([D] + [B]) + \text{vol}([B]) = \frac{1}{2} \int_{[D]} J \wedge J + 2 \, \text{vol}([B])\,,
\end{equation}
for $[B]$ the effective divisor class with the smallest volume such that $[D] + [B]$ is effective. In particular, our piecewise calibrated divisor is $D_{\text{p.c.}} = A - B$ for $A, B$ being any holomorphic representatives of $[D] - [B]$ and $[B]$, respectively. Equivalently, we are taking the union of the holomorphic divisor $A$ and antiholomorphic divisor $-B$. We see that $\text{vol}(D_{\text{p.c.}}^{\text{min}})$ is bounded exactly in the limit $\text{vol}([B]) \to 0$, meaning the strongest bounds arise when we can take $[B]$ to be a class which shrinks at a boundary of the extended K\"ahler cone (e.g., a generator of an extremal ray of the effective cone).

We will algorithmically compute upper bounds for volumes of holes $[D]$ as follows. We first construct a set $\mathcal{B}$ of effective divisor classes $[B]$ which yield the smallest upper bounds when substituted into \cref{eq:upper_bound_theory}. We recall from the beginning of \cref{sec:toric_hyper_holes} that effective divisors outside the Hilbert basis $\mathcal{H}(\mathcal{E}_X)$ of $\mathcal{E}_X$ must have larger volume than the calibrated volume of at least one Hilbert basis element. This motivates us to construct $\mathcal{B}$ from Hilbert basis elements, which we do as follows.
\begin{equation}
    \label{eq:set_of_small_div}
    \begin{aligned}
        \mathcal{B} &= \Big\{ [B] \; \Big| \; [B] \in \mathcal{H}(\mathcal{E}_X) \text{ such that } [B], [D] + [B] \text{ effective} \Big\}\,.
    \end{aligned}
\end{equation}
While we worked only with $\mathcal{E}_V$ in earlier sections in order to perform calculations across many geometries, we will now compute the true effective cones $\mathcal{E}_X$ as we will only present two explicit examples in this section.

We note that one could imagine, for example, that Hilbert basis elements $[B_1], [B_2]$ could not individually satisfy the criteria above to belong to $\mathcal{B}$ while some non-negative linear combination of the two does satisfy the criteria and thus could give rise to a better upper bound than other individual Hilbert basis elements. However, for simplicity, and because this is a proof of concept, we will not investigate such a possibility. With $\mathcal{B}$ in hand, we bound the volume of a hole $[D]$ from above by applying \cref{eq:upper_bound_theory} to each $[B]$ in $\mathcal{B}$. 

Summarizing these upper and lower bounds, we have that given a non-holomorphic divisor class $[D]$, the volume of its minimum-volume representative $D_{\text{min}}$ obeys
\begin{align}
    \frac{1}{2} \int_{[D]} J \wedge J\leq \text{vol}(D_{\text{min}}) \leq \text{vol}(D_{\text{p.c.}}^{\text{min}}).
\end{align}

\subsection{Revisiting Example \hyperref[sec:h11=2]{1}}

Let's consider the geometry $X_{2,106}$ from \cref{sec:h11=2}. The K\"ahler cone has a single chamber, with divisors shrinking on either boundary. Recall that this geometry has a hole with charge
\begin{align}
    [D] = \begin{pmatrix}
        1 \\ -1
    \end{pmatrix}
\end{align}
in the basis fixed by the GLSM charge matrix shown in \cref{eq:X2106 GLSM}. Specifically, the K\"ahler cone is generated by the columns of the following matrix,
\begin{align}
    \label{eq:kahler_cone}
    \begin{pmatrix}
        1 & 3 \\
        0 & 1
    \end{pmatrix}\,.
\end{align}
$\mathcal{E}_X = \mathcal{E}_V$ for this geometry, and the Hilbert basis of $\mathcal{E}_X$ includes the classes
\begin{equation}
    [D_5] = \begin{pmatrix}
        0 \\ 1
    \end{pmatrix}\,,
    \qquad \text{and} \qquad 
    [D] = \begin{pmatrix}
        1 \\ -1
    \end{pmatrix}.
\end{equation}
The effective classes $[D_6] = 2[D]$ and $[D_5]$ shrink to zero volume along the walls of the K\"ahler cone generated by the first and second columns of \cref{eq:kahler_cone}, respectively; ideally, we would use these classes to impose tight bounds on the volume of $[D]$. The shrinking of these divisors is in direct correspondence to the fact that the birational map associated to these walls contract the divisor classes $[D_6]$ and $[D_5]$, respectively.

We can bound $\text{vol}([D])$ from below using its calibrated volume,
\begin{equation}
    \label{eq:h11_2_lower}
    \int_{[D]} J \wedge J = 2t_2(2t_1 - 3t_2) \leq \text{vol}([D])\,.
\end{equation}

To bound the volume from above, we first construct the set $\mathcal{B}$ of candidate classes to employ for the piecewise calibrated representative. One can show $[D_5]$ satisfies the criteria of \cref{eq:set_of_small_div}, but $[D]$ of course does not.\footnote{One might have additionally tried to let $[B]$ be some effective multiple of $[D]$, but then $[D] + [B]$ isn't effective.} Unfortunately, then, we will not be able to exploit that $[D_6] = 2[D]$ shrinks to place a tight bound on $\text{vol}([D])$. However, we still have the unique entry of $\mathcal{B} = \{[D_5]\}$ at our disposal, which induces the piecewise representative $D_{\text{p.c.}} = D_1 - D_5$.

\begin{figure}[tp!]
    \centering
    \includegraphics[width=\linewidth]{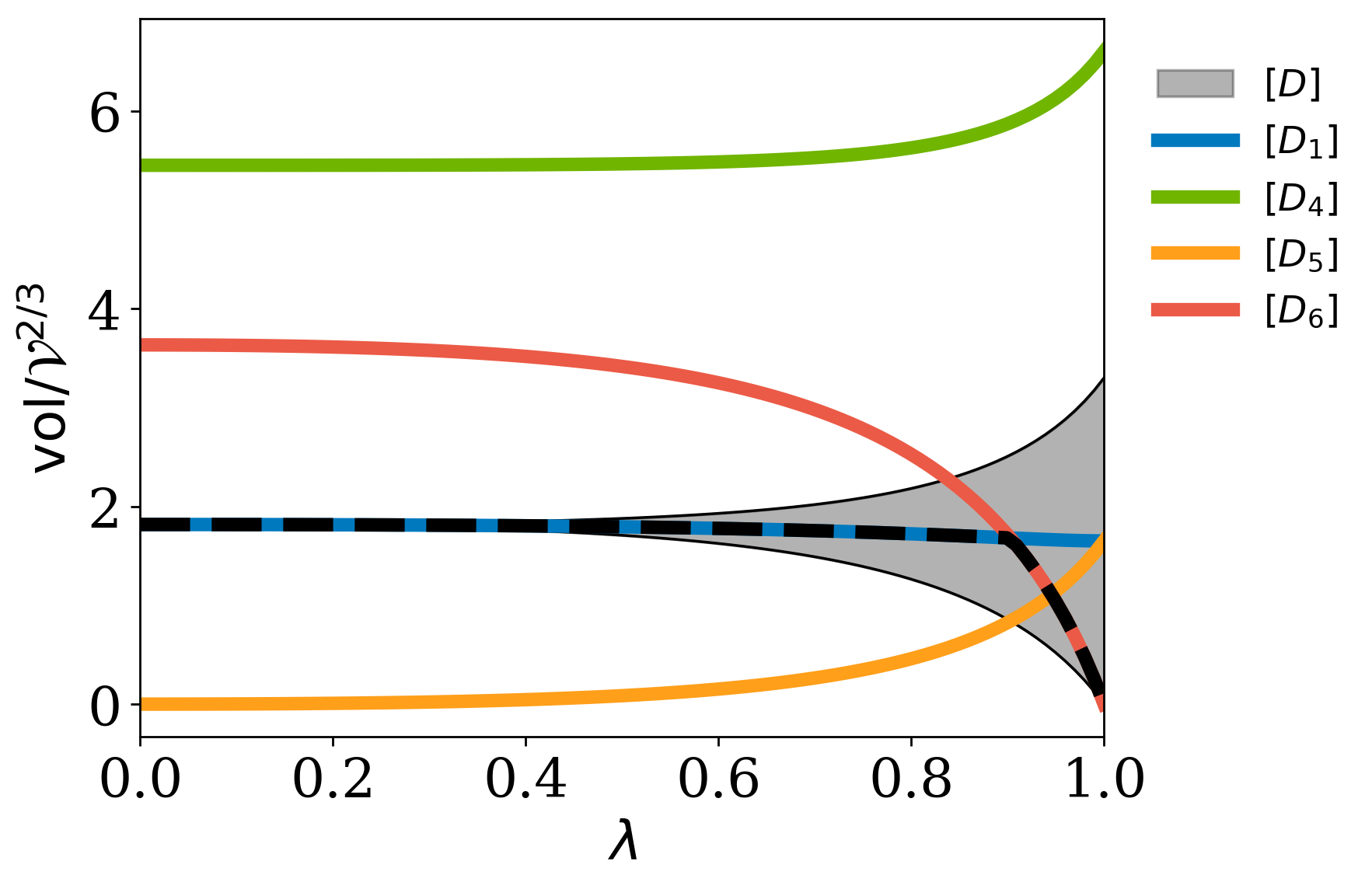}
    \caption{Bounds on the volumes of non-holomorphic cycles representing the hole presented in \cref{sec:h11=2}, compared to the volumes of prime toric divisors. Here, $D_i$ denotes the $i$th prime toric divisor in the ordering set by the columns in the GLSM from~\cref{eq:X2106 GLSM}. The volume $\mathrm{vol}([D])$ of the hole must fall below the dashed line in order for $[D]$ to yield a leading instanton contribution along some direction in the charge lattice. Bounds and volumes are plotted as a function of $\lambda$ parameterizing a 1-dimensional slice in the K\"ahler cone following \cref{eq:h11=2 line param}.}
    \label{fig:volbounds}
\end{figure}

In particular, following \cref{eq:upper_bound_theory}, the non-holomorphic divisor class $[D]$ is bounded from above as
\begin{equation}
    \label{eq:piece wise tension bounds}
    \begin{aligned}
        \text{vol}([D]) 
        &\leq \int_{[D]} J \wedge J + 2\text{vol}([D_5]) = \Big(\text{vol}([D_5]) + \text{vol}([D_6])\Big) + \text{vol}([D_1]) \\
    \end{aligned}
\end{equation}
Along the wall where $D_5$ shrinks to zero volume, we resolve $[D]$ exactly as the second upper bound collapses to the BPS lower bound.

To visualize these bounds easily as a function of the K\"ahler moduli, in \cref{fig:volbounds} we consider a one-dimensional line segment spanning between the two facets of the K\"ahler cone. We parameterize this slice as 
\begin{align}\label{eq:h11=2 line param}
    (1-\lambda) \vec{t}_1 + \lambda \vec{t}_2\,,
\end{align}
where $\lambda \in [0,1]$ and $\vec{t}_1, \vec{t}_2$ are the two columns of \cref{eq:kahler_cone}. We plot the bounds on $\text{vol}([D])$ as well as the volumes of the prime toric divisors. The figure shows the expected features: for $\lambda = 0$, where $[D_5]$ shrinks, the bound is sharp, while at $\lambda = 1$ --- where $[D_5]$ has finite volume --- the bound is weaker. Additionally, imagining these divisors $[D]$ as making contributions to a non-perturbative potential with magnitude $\sim \exp(-\text{vol}([D]))$ along the direction $[D]$, we plot the largest volume $[D]$ could have such that there exists a direction in $H_4(X)$ along which $[D]$ is the leading contribution. In particular, we find that the hole has the capacity to be a leading contribution everywhere along the slice. 

An interesting phenomenon in the spectrum of the resulting lower-dimensional quantum gravity theory emerges as we take the limit where $[D_5]$ shrinks. Let us consider the case of M-theory compactified on $X_{2,106}$. The prepotential in the resulting 5d $\mathcal{N}=1$ quantum gravity theory takes on the following form
\begin{equation}
    \mathcal{F}=\frac{1}{6}\left(t_1^3+t_1^2t_2-3t_1t_2^2+3t_2^3\right)\,.
\end{equation}
Here, the K\"ahler cone constraints are
\begin{equation}
    t_2\geq 0\,,\qquad  t_1-3t_2\geq 0\,.
\end{equation}
We can further restrict to the constant volume slice in the K\"ahler cone.
From this, we can see that both $[D_5]$ and $[D_6]$ shrink to zero volume at finite distance in this codimension 1 moduli space. These correspond to an SCFT wall and an $SU(2)$ boundary, respectively. Note that there are no additional phases associated to this Calabi--Yau threefold, therefore the vector moduli space of the 5d $\cN=1$ quantum gravity theory is finite.

For higher energy levels in the spectrum of this theory, we observe instead a degeneracy as we approach this limit, namely
\begin{equation}
    \lim_{t_2\to 6^{-1/3}} \mathrm{vol}\left(\begin{pmatrix}
        m\\
        0
    \end{pmatrix}+n\begin{pmatrix}
        0\\
        1
    \end{pmatrix}\right)=\lim_{t_2\to 6^{-1/3}}\mathrm{vol}\left(\begin{pmatrix}
        m\\
        0
    \end{pmatrix}\right)=6^{1/3}m\,.
\end{equation}
Here, while it is straightforward to see such a relation when the divisor class $(m, n)$ is effective, this relation remains true for those divisor classes that are not effective due to the inequalities \cref{eq:h11_2_lower,eq:piece wise tension bounds}. Furthermore, not all such divisors are rigid, hence these no longer have the interpretation of being a configuration of identical strings. 
However, how can it be the case that at a strongly-coupled point in the moduli space of quantum gravity we have an exact formula for the tension of infinitely many non-BPS strings? The resolution can be given in two ways: 1) the resulting non-supersymmetric configurations of strings in 5d are no longer bound states; 
or, 2) all of these string states (including those that are supersymmetric) have tensions much heavier than the species scale, which when approaching the SCFT wall takes on the following form~\cite{vandeHeisteeg:2023dlw}
\begin{equation}
    \lim_{t_2\to 6^{-1/3}}\Lambda_s\sim \frac{1}{4\cdot 6^{1/3}}\,.
\end{equation}
Therefore, while such an exact expression on the volume of the four-cycle that is the minimal volume representative of this hole is mathematically accurate in this example, we lose the physical interpretation of this exact volume being the exact tension for a non-BPS string.

\subsection{Revisiting Example \hyperref[sec:h11=4]{4}}

Now let us turn our attention to the geometry $X_{4,94}$ from \cref{sec:h11=4}. Recall that the toric effective cone $\mathcal{E}_{V_{4,94}}$ contains a lattice site in its Hilbert basis (along with a semigroup of sites) that is not effective. This divisor class, which is non-holomorphic, is given in the basis of \cref{eq:x494 GLSM} by
\begin{align}
    [D] = \begin{pmatrix}
        1 \\ 0\\-2 \\-1
    \end{pmatrix}.
\end{align}
The goal of this section is to determine the best possible bounds on the volume of the minimum-volume representative of $[D]$.

\begin{figure}[tp!]
    \centering
    \includegraphics[width=\linewidth]{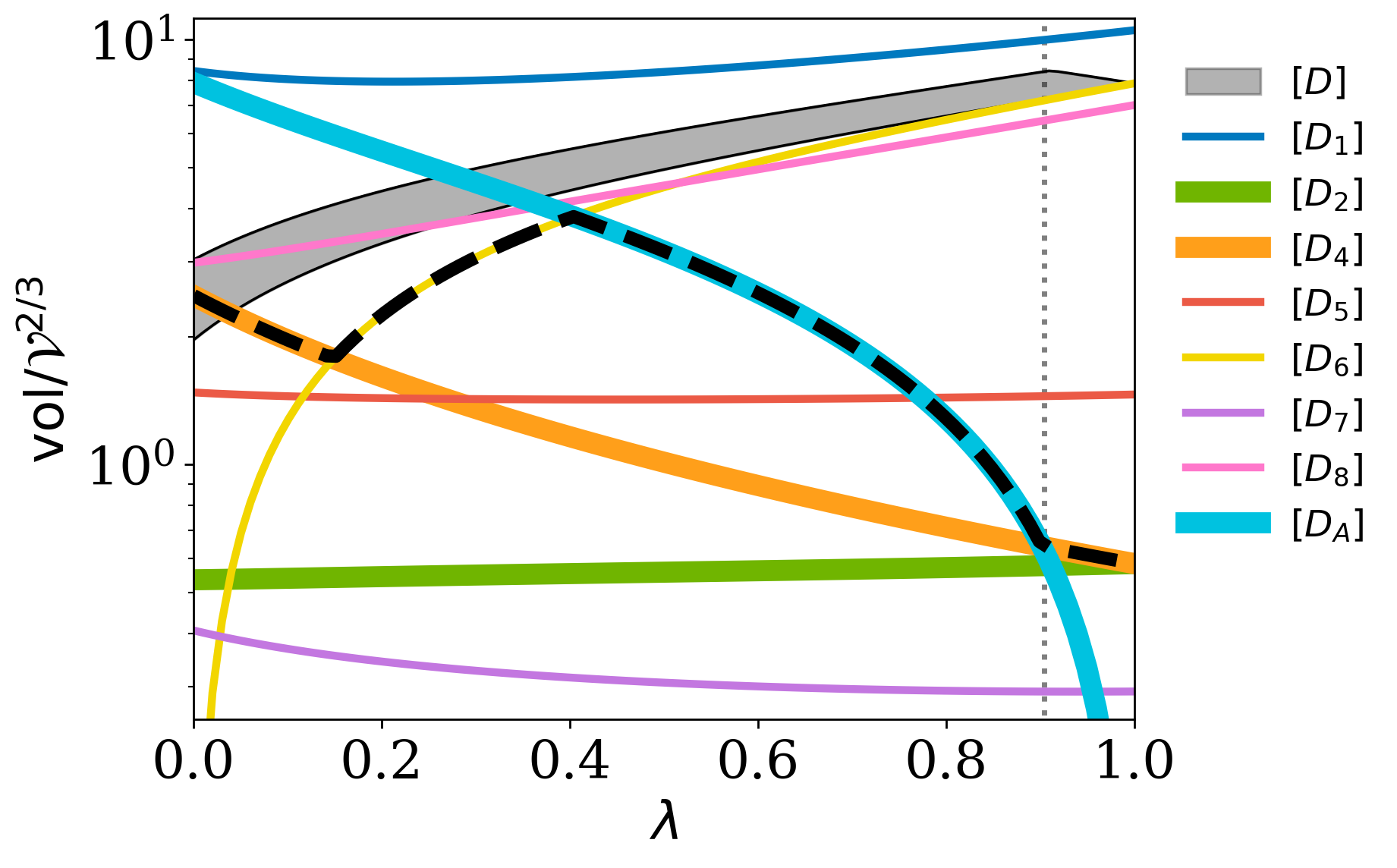}
    \caption{Bounds on the volumes of non-holomorphic cycles representing the hole presented in \cref{sec:h11=4}, compared to the volumes of prime toric divisors as well as that of the autochthonous elements of the Hilbert basis of $\mathcal{E}_X$. Here, $[D_i]$ denotes the $i$th prime toric divisor class in the ordering set by the columns in the GLSM from~\cref{eq:x494 GLSM}, while $[D_A]$ is the autochthonous divisor classes discussed in the text. The volume $\mathrm{vol}([D])$ of the hole must fall below the dashed line in order for $[D]$ to yield a leading instanton contribution along some direction in the charge lattice. 
    Thicker lines denote the divisor classes belonging to $\mathcal{B}$: i.e., those amenable to providing upper bounds on the volume of the hole. Gray vertical lines denote where the divisor class providing the best upper bound changes. Bounds and volumes are plotted as a function of $\lambda$ parameterizing a 1-dimensional slice in the K\"ahler cone following \cref{eq:h11=4 line param}.}
    \label{fig:volboundsh114}
\end{figure}

Using the methods outlined in~\cite{Gendler:2022ztv}, we can compute the extended K\"ahler cone $\mathcal{K}_{\text{ext}}$ of this geometry. In particular, for this example, the Calabi--Yau threefold admits a single non-symmetric flop transition. The fundamental domain of $\mathcal{K}_{\text{ext}}$ is thus comprised of two chambers, and has extremal rays generated by
\begin{align}
    \left\{\begin{pmatrix} 0 \\ 0 \\ 0 \\ 1 \end{pmatrix} , \begin{pmatrix}
        4 \\ 0 \\ -1 \\ -2
    \end{pmatrix}, \begin{pmatrix}
        0 \\ 1 \\ 0 \\ 0
    \end{pmatrix}, \begin{pmatrix}
        2 \\0 \\-2 \\ -1
    \end{pmatrix}, \begin{pmatrix}
        0 \\ 0 \\ 1 \\ 2
    \end{pmatrix}, \begin{pmatrix}
        2 \\ 2 \\ -1 \\-2
    \end{pmatrix}, \begin{pmatrix}
        1 \\1 \\ -1 \\-1
    \end{pmatrix}\right\}.
\end{align}

Along each of the facets of $\mathcal{K}_{\text{ext}}$, a divisor shrinks to zero volume. There are five such unique divisor classes that shrink:
\begin{align}
    [D_4] = \begin{pmatrix}
        0 \\ 0 \\1 \\0
    \end{pmatrix}\,, \ \ 
    [D_6] = \begin{pmatrix}
        1 \\ 1 \\ -3 \\-1
    \end{pmatrix}\,, \ \
    [D_7] = \begin{pmatrix}
        0 \\ -2 \\0 \\1
    \end{pmatrix}, \ \
    [D_8] = \begin{pmatrix}
        1 \\ 0 \\-1 \\-2
    \end{pmatrix}\,, \ \
    [D_A] = \begin{pmatrix}
        -1 \\ 0 \\ 3 \\ 6
    \end{pmatrix}\,. 
\end{align}
The first four are prime toric --- hence the notation --- while $[D_A]$ is autochthonous, and together with the prime toric divisor $[D_2]$ they generate $\mathcal{E}_X$. The Hilbert basis of $\mathcal{E}_X$ is then 
\begin{equation}
    \mathcal{H}(\mathcal{E}_X) = \{[D_2], [D_4], [D_6], [D_7], [D_8], [D], [D_A]\}\,,
\end{equation}
From these we can construct $\mathcal{B}$, which we find to be
\begin{equation}
    \mathcal{B} = \{[D_4], [D_5], [D_A]\}\,.
\end{equation}
This makes sense: $[D]$ clearly isn't effective, while we know from \cref{sec:h11=4} that adding any of $[D_6], [D_7], [D_8]$ to $[D]$ is also not effective. Thus, we have three divisor classes we can employ to achieve volume upper bounds. Of these, $[D_4]$ and $[D_A]$ generate extremal rays of the effective cone and shrink along a facet of the extended K\"ahler cone, and thus can be used to achieve sharp constraints on the volume of $[D]$.

In particular, we have the following constraints.
\begin{equation}
    \label{eq:bounds}
    \frac{1}{2}\int_{[D]} J \wedge J \leq \text{vol}([D]) \leq \frac{1}{2}\int_{[D]} J \wedge J + 2 \text{vol}([B])\,.
\end{equation}
Here, $[B]$ is any element of $\mathcal{B}$.

To demonstrate the utility of our bounds on the minimum volume representative of $[D]$, we again consider a one-dimensional linear slice through the extended $\mathcal{K}_{\text{ext}}$. Along this slice, we compute each of the four upper and lower bounds given by \cref{eq:bounds}. \cref{fig:volboundsh114} shows the allowed region for the volume of $[D]$, along with the volumes of the prime toric divisors, as well as the autochthonous divisor $[D_A]$ in the Hilbert basis of $\mathcal{E}_X$. Here, we have again normalized the divisor volumes by the overall volume of the Calabi--Yau, such that the quantities plotted are invariant under overall homogeneous rescalings of the K\"ahler moduli. This slice is parameterized as
\begin{align}\label{eq:h11=4 line param}
    (1-\lambda) \vec{t}_1 + \lambda \vec{t}_2\,,
\end{align}
where $\lambda \in [0,1]$ and $\vec{t}_1$ and $\vec{t}_2$ are K\"ahler parameters on the facets of the extended K\"ahler cone where the classes $[D_6]$ and $[D_A]$ shrink, respectively. Note that the most precise bound arises along the wall of $\mathcal{K}_{\text{ext}}$ where $[D_A]$ shrinks, because this class belongs to $\mathcal{B}$, while we lack a tight bound when $[D_6]$ shrinks. This figure is styled analogously to \cref{fig:volbounds}, with one additional feature: the element of $\mathcal{B}$ which sets the tightest upper bound varies as a function of $\lambda$, so we demarcate the crossover points with vertical gray dashed lines. In particular, from left to right, the best bound is set by $[D_2]$ and then $[D_A]$. We conclude by noting that in the limit $\lambda \to 0$, we see a regime where the volume of the hole could attain the value needed for one to expect it to provide a leading contribution.

\section{Conclusions and Future Directions} \label{sec:conclusions}

In this work, we have studied the holomorphicity of divisor classes in Calabi--Yau threefolds. This property is vital for determining the spectra of string compactifications, including corrections to four-dimensional compactifications of string theory. This analysis therefore has bearing on which compactifications may give rise to effective theories that resemble our own universe.

In particular, we have illustrated the existence of holes and analyzed their properties,
combining general mathematical results with empirical evidence gathered from a large number of toric Calabi--Yau threefolds. Important examples of holes are furnished by 
so-called non-trivial Hilbert basis elements in toric hypersurface Calabi--Yau effective cones. 
In \cref{fig:additional elements}, we showed that many Calabi--Yau threefolds from the Kreuzer--Skarke database contain such divisors, and furthermore that the number of these divisors can be extremely large, sometimes in the thousands, strongly motivating their study. 

We now summarize two major through lines of this work, both of which feature formal and empirical mathematical aspects.

First, we compared non-trivial Hilbert basis elements arising on the boundary and in the interior of $\mathcal{E}_V$ (i.e., descending from non-big and big ambient divisor classes, respectively). By applying the Kawamata--Viehweg vanishing theorem for klt pairs, in \cref{prop:big} we proved that every non-trivial Hilbert basis element that is interior to $\mathcal{E}_V$ is a hole. In \cref{prop:torsion}, we also related non-movable holes that belong not just to $\partial \mathcal{E}_V$ but to $\partial \mathcal{E}_X$ to the existence of class group torsion on a birational variety. After this, we analyzed all non-trivial Hilbert basis elements in an ensemble of Calabi--Yau threefolds: namely, all favorable birational classes in the Kreuzer--Skarke databse with $h^{1,1}\leq 6$ and $\sim 10^5$ such birational classes in the range of $6<h^{1,1} \leq 17$.
In \cref{fig:interior fraction}, we showed that the majority of such elements appear on the boundary of the toric effective cone, with statistics surprisingly independent of $h^{1,1}$. While \textit{a priori} boundary elements need not be holes, we empirically encountered in \cref{sec:ksholes} that within our ensemble, \textit{all} non-trivial Hilbert basis elements are holes. This has motivated us to pose \cref{conj:non_trivial_hb_are_holes},
which implies that none of these potentially dangerous divisors in practice ever contribute to the superpotential in type IIB Calabi--Yau orientifold compactifications --- they only can contribute to the K\"ahler potential.
In \cref{sec:vol_bounds}, we explained how to bound the volumes of the non-holomorphic cycles representing holes, thereby estimating their contribution to unprotected quantities such as the K\"ahler potential.

Second, we studied how holes can arise in families. Inspired by ideas from the minimal model program in algebraic geometry, in \cref{cor:semigroup} we demonstrated that non-movable holes in Calabi--Yau threefold effective cones come in semigroups (and that perhaps all holes are non-movable: see \cref{conj:nonbig_movable_CY}). We also explained in \cref{sec:friends} why we expect non-trivial Hilbert basis elements that descend from non-big divisors on the ambient variety to result in semigroups contained in the boundary of $\mathcal{E}_V$. More generally, in \cref{th:CY_MMP_holes} we showed how non-movable holes correspond to nef holes on singular, birationally related varieties: for example, nef divisors that are torsion (\cref{prop:torsion}). These expectations regarding semigroups were largely verified in our ensemble of toric hypersurfaces, but the ensemble also furnished novel structure. In particular, we found that holes could come in families larger than that naively predicted by \cref{cor:semigroup}. In particular, these families were most often codimension-one. However, in an example presented in \cref{sec:h113 interior}, we found a family of holes arising from a non-big non-trivial Hilbert basis element which was not on the boundary of $\mathcal{E}_V$ and instead populated a solid region in the effective cone of dimension $h^{1,1}$. We related this to the presence of countably many nef holes on birational contractions of the original Calabi--Yau threefold. An automorphism of the effective cone in this example also allowed us to identify the presence of holes in the non-inherited region $\mathcal{E}_X \setminus \mathcal{E}_V$ of the effective cone, alongside many autochthonous divisors.

There are many exciting implications of our present work to which we will outline three particular directions which would be interesting to explore.
\begin{itemize}
    \item \textbf{The weak gravity conjecture:} In~\cite{Gendler:2022ztv}, it was observed that all electric BPS particles engineered via M2-branes wrapping holomorphic curves in a Calabi--Yau threefold satisfy the weak gravity conjecture~\cite{Arkani-Hamed:2006emk} and variants thereof~\cite{Heidenreich:2015nta,Andriolo:2018lvp,Heidenreich:2016aqi}. Additionally, the weak gravity conjecture remains satisfied among non-BPS particles appearing in type IIB compactifications on Calabi--Yau threefolds~\cite{Gendler:2020dfp}. However, the lattice weak gravity conjecture is known to be violated, see e.g., \cite{Etheredge:2025rkn,Reece:2026hfk}. An analogous study of magnetic monopole strings in Calabi--Yau compactifications of M-theory has not been performed, and the structure of holes would directly interface with the magnetic weak gravity conjecture --- especially its refinements that make predictions about individual charges. 
    Such a study could fruitfully sharpen refinements of the weak gravity conjecture.
    \item \textbf{Completeness of spectrum:} It is hypothesized that consistent theories of quantum gravity must contain a complete spectrum of gauge charges~\cite{Polchinski:2003bq}. In recent years, increasingly compelling arguments and evidence for this hypothesis has been presented~\cite{Banks:2010zn,Harlow:2018tng,Rudelius:2020orz,Heidenreich:2021xpr,McNamara:2021cuo}. 
    Holes naively represent magnetic charges in M-theory without BPS representatives where one would've otherwise expected them, thereby clarifying what completeness of the BPS spectrum could look like in quantum gravity. In particular, it would be interesting to understand the fundamental object in quantum gravity that carries complete spectrum of charges in the context of Calabi--Yau compactifications.
    \item \textbf{Probing enumerative geometry:} With the results of our present work, we can efficiently determine whether for a given charge, the resulting configuration of magnetic monopole string is supersymmetric. However, this information is potentially too coarse for carrying out a detailed study along the aforementioned two directions. In particular, one would ideally count the number of BPS magnetic monopole string states arising for each charge, which should be captured by an enumerative geometric invariant for divisor classes, such as the Donaldson-Thomas invariant~\cite{maulik2006gromovi,maulik2006gromovii}, in a manner resembling the Gopakumar-Vafa invariant for curves~\cite{Gopakumar:1998jq} in M-theory compactifications. There is not a systematic theory of such an enumerative invariant for surfaces, but the invariant should at least vanish for holes in Calabi--Yau threefolds. Hence, our results on the structure of holes could inform future studies of BPS state counting and enumerative geometry in the context of M-theory compactifications on Calabi--Yau threefolds, in a manner similar to the conjectures of~\cite{Reece:2025zva}.
\end{itemize}

\noindent While the above list is certainly not exhaustive, we hope that by highlighting the connections our work has to quantum gravity and enumerative geometry, we may inspire further applications of divisor cohomology calculations.

\subsubsection*{Acknowledgements}

We are grateful to Sebastian Vander Ploeg Fallon, Daniel Halpern-Leistner, Liam McAllister, Jakob Moritz, Houri Tarazi, and Cumrun Vafa for helpful conversations. We thank Chen Jiang for guidance regarding the proof of \cref{th:cy_log_minimal_model}. This material is based upon work supported by the National Science Foundation Graduate Research Fellowship under Grant No. 2139899. The work of NG and DW is supported in part by a grant from the Simons Foundation (602883, CV) and the DellaPietra Foundation. The work of MS was supported by NSF DMS-2001367. ES is supported in part by NSF grant PHY-2309456. NG and DW thank the organizers of the 2024 and 2025 Summer Workshops at the Simons Center for Geometry and Physics for kind hospitality.

\bibliographystyle{JHEP}
\bibliography{refs}

@article{Demirtas:2019lfi,
    author = "Demirtas, Mehmet and Long, Cody and McAllister, Liam and Stillman, Mike",
    title = "{Minimal Surfaces and Weak Gravity}",
    eprint = "1906.08262",
    archivePrefix = "arXiv",
    primaryClass = "hep-th",
    doi = "10.1007/JHEP03(2020)021",
    journal = "JHEP",
    volume = "03",
    pages = "021",
    year = "2020"
}

@article{Gendler:2022ztv,
    author = "Gendler, Naomi and Heidenreich, Ben and McAllister, Liam and Moritz, Jakob and Rudelius, Tom",
    title = "{Moduli space reconstruction and Weak Gravity}",
    eprint = "2212.10573",
    archivePrefix = "arXiv",
    primaryClass = "hep-th",
    reportNumber = "ACFI-T22-10",
    doi = "10.1007/JHEP12(2023)134",
    journal = "JHEP",
    volume = "12",
    pages = "134",
    year = "2023"
}

@article{Demirtas:2018akl,
    author = "Demirtas, Mehmet and Long, Cody and McAllister, Liam and Stillman, Mike",
    title = "{The Kreuzer-Skarke Axiverse}",
    eprint = "1808.01282",
    archivePrefix = "arXiv",
    primaryClass = "hep-th",
    doi = "10.1007/JHEP04(2020)138",
    journal = "JHEP",
    volume = "04",
    pages = "138",
    year = "2020"
}

@Inbook{Hirzebruch1966,
    author="Hirzebruch, F.",
    title="The Riemann-Roch theorem for algebraic manifolds",
    bookTitle="Topological Methods in Algebraic Geometry",
    year="1966",
    publisher="Springer Berlin Heidelberg",
    address="Berlin, Heidelberg",
    pages="114--158",
    isbn="978-3-662-30697-0",
    doi="10.1007/978-3-662-30697-0_5",
    url="https://doi.org/10.1007/978-3-662-30697-0_5"
}

@incollection{miyaoka1987chern,
      title={The Chern classes and Kodaira dimension of a minimal variety},
      author={Miyaoka, Yoichi},
      booktitle={Algebraic geometry, Sendai, 1985},
      volume={10},
      pages={449--477},
      year={1987},
      publisher={Mathematical Society of Japan}
}

@article{Batyrev:2005jc,
    author = "Batyrev, Victor and Kreuzer, Maximilian",
    title = "{Integral cohomology and mirror symmetry for Calabi-Yau 3-folds}",
    eprint = "math/0505432",
    archivePrefix = "arXiv",
    month = "5",
    year = "2005"
}

@article{Kreuzer:2000xy,
    author = "Kreuzer, Maximilian and Skarke, Harald",
    title = "{Complete classification of reflexive polyhedra in four-dimensions}",
    eprint = "hep-th/0002240",
    archivePrefix = "arXiv",
    reportNumber = "HUB-EP-00-13, TUW-00-07",
    doi = "10.4310/ATMP.2000.v4.n6.a2",
    journal = "Adv. Theor. Math. Phys.",
    volume = "4",
    pages = "1209--1230",
    year = "2000"
}

@article{batyrev1993dualpolyhedramirrorsymmetry,
      title={Dual Polyhedra and Mirror Symmetry for Calabi-Yau Hypersurfaces in Toric Varieties}, 
      author={Victor V. Batyrev},
      year={1993},
      eprint={alg-geom/9310003},
      archivePrefix={arXiv},
      primaryClass={alg-geom},
      url={https://arxiv.org/abs/alg-geom/9310003}, 
}

@article{Long:2021lon,
    author = "Long, Cody and Sheshmani, Artan and Vafa, Cumrun and Yau, Shing-Tung",
    title = "{Non-Holomorphic Cycles and Non-BPS Black Branes}",
    eprint = "2104.06420",
    archivePrefix = "arXiv",
    primaryClass = "hep-th",
    reportNumber = "CIMP-D-21-00618R0",
    doi = "10.1007/s00220-022-04587-4",
    journal = "Commun. Math. Phys.",
    volume = "399",
    number = "3",
    pages = "1991--2043",
    year = "2023"
}

@article{Katz:2020ewz,
    author = "Katz, Sheldon and Kim, Hee-Cheol and Tarazi, Houri-Christina and Vafa, Cumrun",
    title = "{Swampland Constraints on 5d $\mathcal{N}=1$ Supergravity}",
    eprint = "2004.14401",
    archivePrefix = "arXiv",
    primaryClass = "hep-th",
    doi = "10.1007/JHEP07(2020)080",
    journal = "JHEP",
    volume = "07",
    pages = "080",
    year = "2020"
}

@article{Oguiso,
    author = {OGUISO, KEIJI},
    title = {ON ALGEBRAIC FIBER SPACE STRUCTURES ON A CALABI-YAU 3-FOLD},
    journal = {International Journal of Mathematics},
    volume = {04},
    number = {03},
    pages = {439-465},
    year = {1993},
    doi = {10.1142/S0129167X93000248},
    URL = {https://doi.org/10.1142/S0129167X93000248},
}

@article{Blumenhagen:2010pv,
    author = "Blumenhagen, Ralph and Jurke, Benjamin and Rahn, Thorsten and Roschy, Helmut",
    title = "{Cohomology of Line Bundles: A Computational Algorithm}",
    eprint = "1003.5217",
    archivePrefix = "arXiv",
    primaryClass = "hep-th",
    reportNumber = "MPP-2010-32, NSF-KITP-10-031",
    doi = "10.1063/1.3501132",
    journal = "J. Math. Phys.",
    volume = "51",
    pages = "103525",
    year = "2010"
}

@Misc{M2,
          author = {Grayson, Daniel R. and Stillman, Michael E.},
          title = {Macaulay2, a software system for research in algebraic geometry},
          howpublished = {Available at \url{http://www2.macaulay2.com}}
        }

@book{cls,
  title={Toric Varieties},
  author={Cox, D.A. and Little, J.B. and Schenck, H.K.},
  isbn={9780821884263},
  series={Graduate studies in mathematics},
  url={https://books.google.com/books?id=eXLGwYD4pmAC},
  year={2011},
  publisher={American Mathematical Soc.}
}

@book{lazarsfeld2017positivity,
  title={Positivity in algebraic geometry I: Classical setting: line bundles and linear series},
  author={Lazarsfeld, Robert K},
  volume={48},
  year={2017},
  publisher={Springer}
}

@article{tanakakawamata,
  title={Kawamata-Viehweg vanishing for toric varieties},
  author={Tanaka, Hiromu},
  eprint = "2208.09680",
  archivePrefix = "arXiv",
  journal={Pure and Applied Mathematics Quarterly},
  publisher={International Press of Boston},
  year={2025},
}

@article{Witten:1993yc,
    author = "Witten, Edward",
    editor = "Greene, B. and Yau, Shing-Tung",
    title = "{Phases of N=2 theories in two-dimensions}",
    eprint = "hep-th/9301042",
    archivePrefix = "arXiv",
    reportNumber = "IASSNS-HEP-93-3",
    doi = "10.1016/0550-3213(93)90033-L",
    journal = "Nucl. Phys. B",
    volume = "403",
    pages = "159--222",
    year = "1993"
}

@article{Maldacena:1997de,
    author = "Maldacena, Juan Martin and Strominger, Andrew and Witten, Edward",
    title = "{Black hole entropy in M theory}",
    eprint = "hep-th/9711053",
    archivePrefix = "arXiv",
    doi = "10.1088/1126-6708/1997/12/002",
    journal = "JHEP",
    volume = "12",
    pages = "002",
    year = "1997"
}

@article{Demirtas:2022hqf,
    author = "Demirtas, Mehmet and Rios-Tascon, Andres and McAllister, Liam",
    title = "{CYTools: A Software Package for Analyzing Calabi-Yau Manifolds}",
    eprint = "2211.03823",
    archivePrefix = "arXiv",
    primaryClass = "hep-th",
    month = "11",
    year = "2022"
}

@article{mustata_local_2000,
	title = {Local {Cohomology} at {Monomial} {Ideals}},
	volume = {29},
	issn = {0747-7171},
	url = {https://www.sciencedirect.com/science/article/pii/S0747717199903024},
	doi = {10.1006/jsco.1999.0302},
	number = {4},
	urldate = {2025-10-21},
	journal = {Journal of Symbolic Computation},
	author = {Musta\c{t}\u{a}, Mircea},
    eprint = "math/0001153",
    archivePrefix = "arXiv",
	month = may,
	year = {2000},
	pages = {709--720},
}

@article{eisenbud_coho,
author = {Eisenbud, David and Musta\c{t}\u{a}, Mircea and Stillman, Mike},
title = {Cohomology on toric varieties and local cohomology with monomial supports},
eprint = "math/0001159",
archivePrefix = "arXiv",
year = {2000},
issue_date = {April/May 2000},
publisher = {Academic Press, Inc.},
address = {USA},
volume = {29},
number = {4–5},
issn = {0747-7171},
url = {https://doi.org/10.1006/jsco.1999.0326},
doi = {10.1006/jsco.1999.0326},
journal = {J. Symb. Comput.},
month = apr,
pages = {583–600},
numpages = {18}
}

@article{maclagan_multigraded_2004,
	title = {Multigraded {Castelnuovo}-{Mumford} regularity},
    author = {Maclagan, Diane and Smith, Gregory G.},
    eprint = "math/0305214",
    archivePrefix = "arXiv",
	volume = {2004},
	url = {https://doi.org/10.1515/crll.2004.040},
	doi = {doi:10.1515/crll.2004.040},
	number = {571},
	urldate = {2025-10-21},
	journal = {Journal für die reine und angewandte Mathematik},
	year = {2004},
	pages = {179--212},
}

@article{Blumenhagen:2010ed,
    author = "Blumenhagen, Ralph and Jurke, Benjamin and Rahn, Thorsten and Roschy, Helmut",
    title = "{Cohomology of Line Bundles: Applications}",
    eprint = "1010.3717",
    archivePrefix = "arXiv",
    primaryClass = "hep-th",
    reportNumber = "MPP-2010-134, NSF-KITP-10-131",
    doi = "10.1063/1.3677646",
    journal = "J. Math. Phys.",
    volume = "53",
    pages = "012302",
    year = "2012"
}

@article{Rahn:2010fm,
    author = "Rahn, Thorsten and Roschy, Helmut",
    title = "{Cohomology of Line Bundles: Proof of the Algorithm}",
    eprint = "1006.2392",
    archivePrefix = "arXiv",
    primaryClass = "hep-th",
    reportNumber = "MPP-2010-64",
    doi = "10.1063/1.3523318",
    journal = "J. Math. Phys.",
    volume = "51",
    pages = "103520",
    year = "2010"
}

@article{jow_cohomology_2011,
	title = {Cohomology of toric line bundles via simplicial {Alexander} duality},
    eprint = "1006.0780",
    archivePrefix = "arXiv",
	volume = {52},
	issn = {0022-2488},
	url = {https://doi.org/10.1063/1.3562523},
	doi = {10.1063/1.3562523},
	number = {3},
	journal = {Journal of Mathematical Physics},
	author = {Jow, Shin-Yao},
	month = mar,
	year = {2011},
	pages = {033506},
}

@Misc{cohomCalg:Implementation,
   title     = "{cohomCalg package}",
   howpublished  = "Download link",
   url       = "https://github.com/BenjaminJurke/cohomCalg",
   note      = "High-performance line bundle cohomology computation based on \cite{Blumenhagen:2010pv}",
   year      = "2010"}

@article{Arkani-Hamed:2006emk,
    author = "Arkani-Hamed, Nima and Motl, Lubos and Nicolis, Alberto and Vafa, Cumrun",
    title = "{The String landscape, black holes and gravity as the weakest force}",
    eprint = "hep-th/0601001",
    archivePrefix = "arXiv",
    reportNumber = "HUTP-05-A0057",
    doi = "10.1088/1126-6708/2007/06/060",
    journal = "JHEP",
    volume = "06",
    pages = "060",
    year = "2007"
}

@article{Lee:2019wij,
    author = "Lee, Seung-Joo and Lerche, Wolfgang and Weigand, Timo",
    title = "{Emergent strings from infinite distance limits}",
    eprint = "1910.01135",
    archivePrefix = "arXiv",
    primaryClass = "hep-th",
    reportNumber = "CERN-TH-2019-159",
    doi = "10.1007/JHEP02(2022)190",
    journal = "JHEP",
    volume = "02",
    pages = "190",
    year = "2022"
}

@article{vandeHeisteeg:2023dlw,
    author = "van de Heisteeg, Damian and Vafa, Cumrun and Wiesner, Max and Wu, David H.",
    title = "{Species scale in diverse dimensions}",
    eprint = "2310.07213",
    archivePrefix = "arXiv",
    primaryClass = "hep-th",
    doi = "10.1007/JHEP05(2024)112",
    journal = "JHEP",
    volume = "05",
    pages = "112",
    year = "2024"
}

@article{Demirtas:2021gsq,
    author = "Demirtas, Mehmet and Gendler, Naomi and Long, Cody and McAllister, Liam and Moritz, Jakob",
    title = "{PQ axiverse}",
    eprint = "2112.04503",
    archivePrefix = "arXiv",
    primaryClass = "hep-th",
    doi = "10.1007/JHEP06(2023)092",
    journal = "JHEP",
    volume = "06",
    pages = "092",
    year = "2023"
}

@article{Gendler:2023hwg,
    author = "Gendler, Naomi and Janssen, Oliver and Kleban, Matthew and La Madrid, Joan and Mehta, Viraf M.",
    title = "{Axion minima in string theory}",
    eprint = "2309.01831",
    archivePrefix = "arXiv",
    primaryClass = "hep-th",
    doi = "10.1007/JHEP02(2025)134",
    journal = "JHEP",
    volume = "02",
    pages = "134",
    year = "2025"
}

@article{Gendler:2023kjt,
    author = "Gendler, Naomi and Marsh, David J. E. and McAllister, Liam and Moritz, Jakob",
    title = "{Glimmers from the axiverse}",
    eprint = "2309.13145",
    archivePrefix = "arXiv",
    primaryClass = "hep-th",
    reportNumber = "KCL-PH-TH/2023-49",
    doi = "10.1088/1475-7516/2024/09/071",
    journal = "JCAP",
    volume = "09",
    pages = "071",
    year = "2024"
}

@article{Gendler:2024adn,
    author = "Gendler, Naomi and Marsh, David J. E.",
    title = "{Possible Implications of QCD Axion Dark Matter Constraints from Helioscopes and Haloscopes for the String Theory Landscape}",
    eprint = "2407.07143",
    archivePrefix = "arXiv",
    primaryClass = "hep-th",
    doi = "10.1103/PhysRevLett.134.081602",
    journal = "Phys. Rev. Lett.",
    volume = "134",
    number = "8",
    pages = "081602",
    year = "2025"
}

@article{Sheridan:2024vtt,
    author = "Sheridan, Elijah and Carta, Federico and Gendler, Naomi and Jain, Mudit and Marsh, David J. E. and McAllister, Liam and Righi, Nicole and Rogers, Keir K. and Schachner, Andreas",
    title = "{Fuzzy axions and associated relics}",
    eprint = "2412.12012",
    archivePrefix = "arXiv",
    primaryClass = "hep-th",
    reportNumber = "KCL-PH-TH/2024-75, KCL-PH-TH/2024-75",
    doi = "10.1007/JHEP09(2025)016",
    journal = "JHEP",
    volume = "09",
    pages = "016",
    year = "2025"
}

@article{Cheng:2025ggf,
    author = "Cheng, Junyi and Gendler, Naomi",
    title = "{Universality in the axiverse}",
    eprint = "2507.12516",
    archivePrefix = "arXiv",
    primaryClass = "hep-th",
    doi = "10.1007/JHEP11(2025)012",
    journal = "JHEP",
    volume = "11",
    pages = "012",
    year = "2025"
}

@article{Yin:2025amn,
    author = "Yin, Ziwen and Cheng, Hanyu and Di Valentino, Eleonora and Gendler, Naomi and Marsh, David J. E. and Visinelli, Luca",
    title = "{Constraining the axiverse with reionization}",
    eprint = "2507.03535",
    archivePrefix = "arXiv",
    primaryClass = "hep-ph",
    reportNumber = "CA21106; CA21136",
    month = "7",
    year = "2025"
}

@article{McAllister:2024lnt,
    author = "McAllister, Liam and Moritz, Jakob and Nally, Richard and Schachner, Andreas",
    title = "{Candidate de Sitter vacua}",
    eprint = "2406.13751",
    archivePrefix = "arXiv",
    primaryClass = "hep-th",
    reportNumber = "CERN-TH-2024-090",
    doi = "10.1103/PhysRevD.111.086015",
    journal = "Phys. Rev. D",
    volume = "111",
    number = "8",
    pages = "086015",
    year = "2025"
}

@article{Demirtas:2021ote,
    author = "Demirtas, Mehmet and Kim, Manki and McAllister, Liam and Moritz, Jakob and Rios-Tascon, Andres",
    title = "{Exponentially Small Cosmological Constant in String Theory}",
    eprint = "2107.09065",
    archivePrefix = "arXiv",
    primaryClass = "hep-th",
    doi = "10.1103/PhysRevLett.128.011602",
    journal = "Phys. Rev. Lett.",
    volume = "128",
    number = "1",
    pages = "011602",
    year = "2022"
}

@article{Demirtas:2021nlu,
    author = "Demirtas, Mehmet and Kim, Manki and McAllister, Liam and Moritz, Jakob and Rios-Tascon, Andres",
    title = "{Small cosmological constants in string theory}",
    eprint = "2107.09064",
    archivePrefix = "arXiv",
    primaryClass = "hep-th",
    doi = "10.1007/JHEP12(2021)136",
    journal = "JHEP",
    volume = "12",
    pages = "136",
    year = "2021"
}

@article{Mehta:2021pwf,
    author = "Mehta, Viraf M. and Demirtas, Mehmet and Long, Cody and Marsh, David J. E. and McAllister, Liam and Stott, Matthew J.",
    title = "{Superradiance in string theory}",
    eprint = "2103.06812",
    archivePrefix = "arXiv",
    primaryClass = "hep-th",
    doi = "10.1088/1475-7516/2021/07/033",
    journal = "JCAP",
    volume = "07",
    pages = "033",
    year = "2021"
}

@book{weibel1994introduction,
  title={An introduction to homological algebra},
  author={Weibel, Charles A},
  number={38},
  year={1994},
  publisher={Cambridge university press}
}

@article{Cicoli:2012sz,
    author = "Cicoli, Michele and Goodsell, Mark and Ringwald, Andreas",
    title = "{The type IIB string axiverse and its low-energy phenomenology}",
    eprint = "1206.0819",
    archivePrefix = "arXiv",
    primaryClass = "hep-th",
    reportNumber = "DESY-12-058, CERN-PH-TH-2012-153",
    doi = "10.1007/JHEP10(2012)146",
    journal = "JHEP",
    volume = "10",
    pages = "146",
    year = "2012"
}

@article{hu2000mori,
  title={Mori dream spaces and GIT.},
  author={Hu, Yi and Keel, Sean},
  journal={Michigan Mathematical Journal},
  volume={48},
  number={1},
  pages={331--348},
  year={2000},
  publisher={University of Michigan, Department of Mathematics},
  eprint = {math/0004017},
  arxivPrefix = {arXiv},
}

@article{kawamataCone,
author = {Kawamata, Yujiro},
title = {On the Cone of Divisors of Calabi–Yau Fiber Spaces},
journal = {International Journal of Mathematics},
volume = {08},
number = {05},
pages = {665-687},
year = {1997},
doi = {10.1142/S0129167X97000354},

URL = { 
    
        https://doi.org/10.1142/S0129167X97000354
    
    

},
eprint = {alg-geom/9701006},
arxivPrefix = {arXiv},
}

@misc{NormalToricVarietiesSource,
  title = {{NormalToricVarieties: routines for working with normal toric varieties and related objects. Version~1.9}},
  author = {Gregory G. Smith},
  howpublished = {A \emph{Macaulay2} package available at
    \url{https://github.com/Macaulay2/M2/tree/master/M2/Macaulay2/packages}}
}

@misc{StringToricsSource,
  title = {{StringTorics: toric variety functions useful for investigations in string theory}},
  author = {Michael Stillman},
  howpublished = {A \emph{Macaulay2} package available at
    \url{https://github.com/Macaulay2/Workshop-2024-Utah/tree/CYTools/CYToolsM2/StringTorics}}
}

@article{Reece:2025thc,
    author = "Reece, Matthew",
    title = "{Extra-dimensional axion expectations}",
    eprint = "2406.08543",
    archivePrefix = "arXiv",
    primaryClass = "hep-ph",
    doi = "10.1007/JHEP07(2025)130",
    journal = "JHEP",
    volume = "07",
    pages = "130",
    year = "2025"
}

@article{Maharana:2012tu,
    author = "Maharana, Anshuman and Palti, Eran",
    title = "{Models of Particle Physics from Type IIB String Theory and F-theory: A Review}",
    eprint = "1212.0555",
    archivePrefix = "arXiv",
    primaryClass = "hep-th",
    reportNumber = "HRI-ST-1212, CPHT-RR088.1112",
    doi = "10.1142/S0217751X13300056",
    journal = "Int. J. Mod. Phys. A",
    volume = "28",
    pages = "1330005",
    year = "2013"
}

@article{Marchesano:2024gul,
    author = "Marchesano, Fernando and Shiu, Gary and Weigand, Timo",
    title = "{The Standard Model from String Theory: What Have We Learned?}",
    eprint = "2401.01939",
    archivePrefix = "arXiv",
    primaryClass = "hep-th",
    reportNumber = "IFT-UAM/CSIC-24-01, ZMP-HH/24-01",
    doi = "10.1146/annurev-nucl-102622-012235",
    journal = "Ann. Rev. Nucl. Part. Sci.",
    volume = "74",
    pages = "113--140",
    year = "2024"
}

@article{vandeHeisteeg:2023uxj,
    author = "van de Heisteeg, Damian and Vafa, Cumrun and Wiesner, Max and Wu, David H.",
    title = "{Bounds on field range for slowly varying positive potentials}",
    eprint = "2305.07701",
    archivePrefix = "arXiv",
    primaryClass = "hep-th",
    doi = "10.1007/JHEP02(2024)175",
    journal = "JHEP",
    volume = "02",
    pages = "175",
    year = "2024"
}

@article{Banks:2010zn,
    author = "Banks, Tom and Seiberg, Nathan",
    title = "{Symmetries and Strings in Field Theory and Gravity}",
    eprint = "1011.5120",
    archivePrefix = "arXiv",
    primaryClass = "hep-th",
    doi = "10.1103/PhysRevD.83.084019",
    journal = "Phys. Rev. D",
    volume = "83",
    pages = "084019",
    year = "2011"
}

@article{Polchinski:2003bq,
    author = "Polchinski, Joseph",
    editor = "Baer, H. and Belyaev, A.",
    title = "{Monopoles, duality, and string theory}",
    eprint = "hep-th/0304042",
    archivePrefix = "arXiv",
    doi = "10.1142/S0217751X0401866X",
    journal = "Int. J. Mod. Phys. A",
    volume = "19S1",
    pages = "145--156",
    year = "2004"
}

@article{Alim:2021vhs,
    author = "Alim, Murad and Heidenreich, Ben and Rudelius, Tom",
    title = "{The Weak Gravity Conjecture and BPS Particles}",
    eprint = "2108.08309",
    archivePrefix = "arXiv",
    primaryClass = "hep-th",
    reportNumber = "ACFI-T21-09",
    doi = "10.1002/prop.202100125",
    journal = "Fortsch. Phys.",
    volume = "69",
    number = "11-12",
    pages = "2100125",
    year = "2021"
}

@article{Douglas:2020hpv,
    author = "Douglas, Michael R. and Lakshminarasimhan, Subramanian and Qi, Yidi",
    title = "{Numerical Calabi-Yau metrics from holomorphic networks}",
    eprint = "2012.04797",
    archivePrefix = "arXiv",
    primaryClass = "hep-th",
    month = "12",
    year = "2020"
}

@article{Larfors:2021pbb,
    author = "Larfors, Magdalena and Lukas, Andre and Ruehle, Fabian and Schneider, Robin",
    title = "{Learning Size and Shape of Calabi-Yau Spaces}",
    eprint = "2111.01436",
    archivePrefix = "arXiv",
    primaryClass = "hep-th",
    reportNumber = "UUITP-53/21",
    month = "11",
    year = "2021"
}

@article{Gerdes:2022nzr,
    author = "Gerdes, Mathis and Krippendorf, Sven",
    title = "{CYJAX: A package for Calabi-Yau metrics with JAX}",
    eprint = "2211.12520",
    archivePrefix = "arXiv",
    primaryClass = "hep-th",
    doi = "10.1088/2632-2153/acdc84",
    journal = "Mach. Learn. Sci. Tech.",
    volume = "4",
    number = "2",
    pages = "025031",
    year = "2023"
}

@article{Ashmore:2019wzb,
    author = "Ashmore, Anthony and He, Yang-Hui and Ovrut, Burt A.",
    title = "{Machine Learning Calabi{\textendash}Yau Metrics}",
    eprint = "1910.08605",
    archivePrefix = "arXiv",
    primaryClass = "hep-th",
    doi = "10.1002/prop.202000068",
    journal = "Fortsch. Phys.",
    volume = "68",
    number = "9",
    pages = "2000068",
    year = "2020"
}

@article{Butbaia:2024xgj,
    author = {Butbaia, Giorgi and Mayorga Pe{\~n}a, Dami{\'a}n and Tan, Justin and Berglund, Per and H{\"u}bsch, Tristan and Jejjala, Vishnu and Mishra, Challenger},
    title = "{cymyc: Calabi-Yau Metrics, Yukawas, and Curvature}",
    eprint = "2410.19728",
    archivePrefix = "arXiv",
    primaryClass = "hep-th",
    doi = "10.1007/JHEP03(2025)028",
    journal = "JHEP",
    volume = "03",
    pages = "028",
    year = "2025"
}

@article{Witten:1996bn,
    author = "Witten, Edward",
    title = "{Nonperturbative superpotentials in string theory}",
    eprint = "hep-th/9604030",
    archivePrefix = "arXiv",
    reportNumber = "IASSNS-HEP-96-29",
    doi = "10.1016/0550-3213(96)00283-0",
    journal = "Nucl. Phys. B",
    volume = "474",
    pages = "343--360",
    year = "1996"
}

@article{mckernan2010mori,
  title={Mori dream spaces},
  author={McKernan, James},
  doi = "10.1007/s11537-010-0944-7",
  journal = "Jpn. J. Math.",
  volume = "5",
  pages = "127--151",
  year={2010},
  publisher={Springer},
}

@article{birkar2010existence,
  title={On existence of log minimal models},
  author={Birkar, Caucher},
  eprint = "0706.1792",
  archivePrefix = "arXiv",
  journal={Compositio Mathematica},
  volume={146},
  number={4},
  pages={919--928},
  year={2010},
  publisher={London Mathematical Society}
}

@incollection{AST_1993__218__243_0,
     author = {Morrison, David R.},
     title = {Compactifications of moduli spaces inspired by mirror symmetry},
     booktitle = {Journ\'ees de g\'eom\'etrie alg\'ebrique d'Orsay - Juillet 1992},
     series = {Ast\'erisque},
     pages = {243--271},
     year = {1993},
     publisher = {Soci\'et\'e math\'ematique de France},
     number = {218},
     language = {en},
     url = {https://www.numdam.org/item/AST_1993__218__243_0/},
     eprint = {alg-geom/9304007},
     arxivPrefix = "arXiv",
}

@article{Heidenreich:2016aqi,
    author = "Heidenreich, Ben and Reece, Matthew and Rudelius, Tom",
    title = "{Evidence for a sublattice weak gravity conjecture}",
    eprint = "1606.08437",
    archivePrefix = "arXiv",
    primaryClass = "hep-th",
    doi = "10.1007/JHEP08(2017)025",
    journal = "JHEP",
    volume = "08",
    pages = "025",
    year = "2017"
}

@article{Andriolo:2018lvp,
    author = "Andriolo, Stefano and Junghans, Daniel and Noumi, Toshifumi and Shiu, Gary",
    title = "{A Tower Weak Gravity Conjecture from Infrared Consistency}",
    eprint = "1802.04287",
    archivePrefix = "arXiv",
    primaryClass = "hep-th",
    reportNumber = "KOBE-COSMO-18-01, MAD-TH-17-07",
    doi = "10.1002/prop.201800020",
    journal = "Fortsch. Phys.",
    volume = "66",
    number = "5",
    pages = "1800020",
    year = "2018"
}

@article{Heidenreich:2015nta,
    author = "Heidenreich, Ben and Reece, Matthew and Rudelius, Tom",
    title = "{Sharpening the Weak Gravity Conjecture with Dimensional Reduction}",
    eprint = "1509.06374",
    archivePrefix = "arXiv",
    primaryClass = "hep-th",
    doi = "10.1007/JHEP02(2016)140",
    journal = "JHEP",
    volume = "02",
    pages = "140",
    year = "2016"
}

@article{Reece:2025zva,
    author = "Reece, Matthew and Rudelius, Tom and Tudball, Christopher",
    title = "{Co-scaling and alignment of electric and magnetic towers}",
    eprint = "2505.22713",
    archivePrefix = "arXiv",
    primaryClass = "hep-th",
    doi = "10.1007/JHEP09(2025)146",
    journal = "JHEP",
    volume = "09",
    pages = "146",
    year = "2025"
}

@article{Etheredge:2025rkn,
    author = "Etheredge, Muldrow and Heidenreich, Ben and Pittman, Nicholas and Rauch, Sebastian and Reece, Matthew and Rudelius, Tom",
    title = "{Confined monopoles and failure of the Lattice Weak Gravity Conjecture}",
    eprint = "2502.14951",
    archivePrefix = "arXiv",
    primaryClass = "hep-th",
    reportNumber = "ACFI-T25-01",
    doi = "10.1007/JHEP10(2025)186",
    journal = "JHEP",
    volume = "10",
    pages = "186",
    year = "2025"
}

@article{Harlow:2018tng,
    author = "Harlow, Daniel and Ooguri, Hirosi",
    title = "{Symmetries in quantum field theory and quantum gravity}",
    eprint = "1810.05338",
    archivePrefix = "arXiv",
    primaryClass = "hep-th",
    doi = "10.1007/s00220-021-04040-y",
    journal = "Commun. Math. Phys.",
    volume = "383",
    number = "3",
    pages = "1669--1804",
    year = "2021"
}

@article{Heidenreich:2021xpr,
    author = "Heidenreich, Ben and McNamara, Jacob and Montero, Miguel and Reece, Matthew and Rudelius, Tom and Valenzuela, Irene",
    title = "{Non-invertible global symmetries and completeness of the spectrum}",
    eprint = "2104.07036",
    archivePrefix = "arXiv",
    primaryClass = "hep-th",
    reportNumber = "ACFI-T21-03",
    doi = "10.1007/JHEP09(2021)203",
    journal = "JHEP",
    volume = "09",
    pages = "203",
    year = "2021"
}

@article{Rudelius:2020orz,
    author = "Rudelius, Tom and Shao, Shu-Heng",
    title = "{Topological Operators and Completeness of Spectrum in Discrete Gauge Theories}",
    eprint = "2006.10052",
    archivePrefix = "arXiv",
    primaryClass = "hep-th",
    doi = "10.1007/JHEP12(2020)172",
    journal = "JHEP",
    volume = "12",
    pages = "172",
    year = "2020"
}

@article{McNamara:2021cuo,
    author = "McNamara, Jacob",
    title = "{Gravitational Solitons and Completeness}",
    eprint = "2108.02228",
    archivePrefix = "arXiv",
    primaryClass = "hep-th",
    month = "8",
    year = "2021"
}

@article{Gopakumar:1998jq,
    author = "Gopakumar, Rajesh and Vafa, Cumrun",
    title = "{M theory and topological strings. 2.}",
    eprint = "hep-th/9812127",
    archivePrefix = "arXiv",
    reportNumber = "HUTP-98-A070",
    month = "12",
    year = "1998"
}

@article{Gendler:2020dfp,
    author = "Gendler, Naomi and Valenzuela, Irene",
    title = "{Merging the weak gravity and distance conjectures using BPS extremal black holes}",
    eprint = "2004.10768",
    archivePrefix = "arXiv",
    primaryClass = "hep-th",
    doi = "10.1007/JHEP01(2021)176",
    journal = "JHEP",
    volume = "01",
    pages = "176",
    year = "2021"
}

@article{maulik2006gromovi,
  title={Gromov--witten theory and donaldson--thomas theory, i},
  author={Maulik, Davesh and Nekrasov, Nikita and Okounkov, Andrei and Pandharipande, Rahul},
  eprint = "math/0312059",
  archivePrefix = "arXiv",
  journal={Compositio Mathematica},
  volume={142},
  number={5},
  pages={1263--1285},
  year={2006},
  publisher={London Mathematical Society}
}

@article{maulik2006gromovii,
  title={Gromov--witten theory and donaldson--thomas theory, ii},
  author={Maulik, Davesh and Nekrasov, Nikita and Okounkov, Andrei and Pandharipande, Rahul},
  eprint = "math/0406092",
  archivePrefix = "arXiv",
  journal={Compositio Mathematica},
  volume={142},
  number={5},
  pages={1286--1304},
  year={2006},
  publisher={London Mathematical Society}
}

@inpreparation{constantin_inprep,
    author = {Constantin, Andrei and Lukas, Andre and Sheridan, Elijah},
    title = "{Line Bundle Cohomology Formulae via the Minimal Model Program: Mori Dream Spaces and Calabi--Yau Threefolds}",
    year="\emph{in preparation}"
}

@article{Vafa:2005ui,
    author = "Vafa, Cumrun",
    title = "{The String landscape and the swampland}",
    eprint = "hep-th/0509212",
    archivePrefix = "arXiv",
    reportNumber = "HUTP-05-A043",
    month = "9",
    year = "2005"
}

@article{Grimm:2020cda,
    author = "Grimm, Thomas W.",
    title = "{Moduli space holography and the finiteness of flux vacua}",
    eprint = "2010.15838",
    archivePrefix = "arXiv",
    primaryClass = "hep-th",
    doi = "10.1007/JHEP10(2021)153",
    journal = "JHEP",
    volume = "10",
    pages = "153",
    year = "2021"
}

@article{Hamada:2021yxy,
    author = "Hamada, Yuta and Montero, Miguel and Vafa, Cumrun and Valenzuela, Irene",
    title = "{Finiteness and the swampland}",
    eprint = "2111.00015",
    archivePrefix = "arXiv",
    primaryClass = "hep-th",
    doi = "10.1088/1751-8121/ac6404",
    journal = "J. Phys. A",
    volume = "55",
    number = "22",
    pages = "224005",
    year = "2022"
}

@article{Tarazi:2021duw,
    author = "Tarazi, Houri-Christina and Vafa, Cumrun",
    title = "{On The Finiteness of 6d Supergravity Landscape}",
    eprint = "2106.10839",
    archivePrefix = "arXiv",
    primaryClass = "hep-th",
    month = "6",
    year = "2021"
}

@article{Kim:2024eoa,
    author = "Kim, Hee-Cheol and Vafa, Cumrun and Xu, Kai",
    title = "{Finite Landscape of 6d N=(1,0) Supergravity}",
    eprint = "2411.19155",
    archivePrefix = "arXiv",
    primaryClass = "hep-th",
    doi = "10.21468/SciPostPhys.20.1.016",
    journal = "SciPost Phys.",
    volume = "20",
    pages = "016",
    year = "2026"
}

@article{Delgado:2024skw,
    author = "Delgado, Matilda and van de Heisteeg, Damian and Raman, Sanjay and Torres, Ethan and Vafa, Cumrun and Xu, Kai",
    title = "{Finiteness and the emergence of dualities}",
    eprint = "2412.03640",
    archivePrefix = "arXiv",
    primaryClass = "hep-th",
    reportNumber = "MPP-2024-224, CERN-TH-2024-204",
    doi = "10.21468/SciPostPhys.19.2.047",
    journal = "SciPost Phys.",
    volume = "19",
    number = "2",
    pages = "047",
    year = "2025"
}

@article{Agmon:2022thq,
    author = "Agmon, Nathan Benjamin and Bedroya, Alek and Kang, Monica Jinwoo and Vafa, Cumrun",
    title = "{Lectures on the string landscape and the Swampland}",
    eprint = "2212.06187",
    archivePrefix = "arXiv",
    primaryClass = "hep-th",
    month = "12",
    year = "2022"
}

@article{Yau1990,
     ISSN = {08940347, 10886834},
     URL = {http://www.jstor.org/stable/1990928},
     author = {G. Tian and Shing Tung Yau},
     journal = {Journal of the American Mathematical Society},
     number = {3},
     pages = {579--609},
     publisher = {American Mathematical Society},
     title = {Complete Kähler Manifolds with Zero Ricci Curvature. I},
     urldate = {2025-12-01},
     volume = {3},
     year = {1990}
}

@article{Yau1991,
    author = {Yau, Shing Tung and Tian, Gang},
    journal = {Inventiones mathematicae},
    number = {1},
    pages = {27-60},
    title = {Complete Kähler manifolds with zero Ricci curvature. II.},
    url = {http://eudml.org/doc/143931},
    volume = {106},
    year = {1991},
}

@article{Yau1978,
author = {Yau, Shing-Tung},
title = {On the ricci curvature of a compact kähler manifold and the complex monge-ampére equation, I},
journal = {Communications on Pure and Applied Mathematics},
volume = {31},
number = {3},
pages = {339-411},
doi = {https://doi.org/10.1002/cpa.3160310304},
year = {1978}
}

@article{Jain:2025vfh,
    author = "Jain, Mudit and Sheridan, Elijah and Marsh, David J. E. and Heyes, Elli and Rogers, Keir K. and Schachner, Andreas",
    title = "{Bayesian inference on Calabi--Yau moduli spaces and the axiverse: experimental data meets string theory}",
    eprint = "2512.00144",
    archivePrefix = "arXiv",
    primaryClass = "hep-th",
    month = "11",
    year = "2025"
}

@article{katz1992gorenstein,
  title={Gorenstein threefold singularities with small resolutions via invariant theory for Weyl groups},
  author={Katz, Sheldon and Morrison, David R},
  eprint = "alg-geom/9202002",
  archivePrefix = "arXiv",
  primaryClass = "alg-geom",
  year={1992}
}

@book{Hori:2003ic,
    author = "Hori, K. and Katz, S. and Klemm, A. and Pandharipande, R. and Thomas, R. and Vafa, C. and Vakil, R. and Zaslow, E.",
    title = "{Mirror symmetry}",
    publisher = "AMS",
    address = "Providence, USA",
    series = "Clay mathematics monographs",
    volume = "1",
    year = "2003"
}

@article{Grimm:2023lrf,
    author = "Grimm, Thomas W. and Monnee, Jeroen",
    title = "{Finiteness theorems and counting conjectures for the flux landscape}",
    eprint = "2311.09295",
    archivePrefix = "arXiv",
    primaryClass = "hep-th",
    doi = "10.1007/JHEP08(2024)039",
    journal = "JHEP",
    volume = "08",
    pages = "039",
    year = "2024"
}

@article{Grimm:2021vpn,
    author = "Grimm, Thomas W.",
    title = "{Taming the landscape of effective theories}",
    eprint = "2112.08383",
    archivePrefix = "arXiv",
    primaryClass = "hep-th",
    doi = "10.1007/JHEP11(2022)003",
    journal = "JHEP",
    volume = "11",
    pages = "003",
    year = "2022"
}

@article{kawamata1982generalization,
  title={A generalization of Kodaira-Ramanujam's vanishing theorem},
  author={Kawamata, Yujiro},
  journal={Mathematische Annalen},
  volume={261},
  number={1},
  pages={43--46},
  year={1982},
  publisher={Springer}
}

@article{viehweg1982vanishing,
  title={Vanishing theorems.},
  author={Viehweg, Eckart},
  journal = "Journal für die reine und angewandte Mathematik",
  volume = "335",
  pages = "1--8",
  year={1982},
  publisher={Walter de Gruyter, Berlin/New York Berlin, New York}
}

@Article{BagnaraHZ08SCP, 
  Author = "R. Bagnara and P. M. Hill and E. Zaffanella", 
  Title = "The {Parma Polyhedra Library}: Toward a Complete Set of Numerical 
             Abstractions for the Analysis and Verification 
             of Hardware and Software Systems", 
  Journal = "Science of Computer Programming", 
  Volume = 72, Number = "1--2", 
  Pages = "3--21", 
  Year = 2008,
}

@article{mitchell2011pulp,
  title={Pulp: a linear programming toolkit for python},
  author={Mitchell, Stuart and OSullivan, Michael and Dunning, Iain},
  journal={The University of Auckland, Auckland, New Zealand},
  volume={65},
  pages={25},
  year={2011}
}

@article{kaloghiros2016finite,
  title={Finite generation and geography of models},
  author={Kaloghiros, Anne-Sophie and K{\"u}ronya, Alex and Lazic, Vladimir},
  eprint = "1202.1164",
  archivePrefix = "arXiv",
  journal={Minimal models and extremal rays (Kyoto, 2011)},
  volume={70},
  pages={215--245},
  year={2016}
}

@article{MacFadden:2025ssx,
    author = "MacFadden, Nate and Sheridan, Elijah",
    title = "{Calabi-Yau Threefolds from Vex Triangulations}",
    eprint = "2512.14817",
    archivePrefix = "arXiv",
    primaryClass = "hep-th",
    month = "12",
    year = "2025"
}

@article{keel1994log,
  title={Log abundance theorem for threefolds},
  author={Keel, Sean and Matsuki, Kenji and McKernan, James},
  doi = "10.1215/S0012-7094-94-07504-2",
  journal = "Duke Math. J.",
  volume = "75",
  number = "1",
  pages = "99--119",
  year={1994}
}

@book{KM,
  title={Birational geometry of algebraic varieties},
  author={Koll{\'a}r, J{\'a}nos and Mori, Shigefumi},
  volume={134},
  year={1998},
  publisher={Cambridge university press}
}

@inproceedings{McAllister:2025qwq,
    author = "McAllister, Liam and Schachner, Andreas",
    title = "{TASI Lectures on de Sitter Vacua}",
    eprint = "2512.17095",
    archivePrefix = "arXiv",
    primaryClass = "hep-th",
    month = "12",
    year = "2025"
}

@article{Chen:2025rkb,
    author = "Chen, Shi and van de Heisteeg, Damian and Vafa, Cumrun",
    title = "{Symmetries and M-theory-like vacua in four dimensions}",
    eprint = "2503.16599",
    archivePrefix = "arXiv",
    primaryClass = "hep-th",
    doi = "10.1007/JHEP07(2025)258",
    journal = "JHEP",
    volume = "07",
    pages = "258",
    year = "2025"
}

@article{Mohseni:2025tig,
    author = "Mohseni, Amineh and Vafa, Cumrun",
    title = "{Symmetry points of $ \mathcal{N}=1 $ modular geometry}",
    eprint = "2510.19927",
    archivePrefix = "arXiv",
    primaryClass = "hep-th",
    doi = "10.1007/JHEP02(2026)202",
    journal = "JHEP",
    volume = "02",
    pages = "202",
    year = "2026"
}

@article{Becker:2022hse,
    author = "Becker, Katrin and Gonzalo, Eduardo and Walcher, Johannes and Wrase, Timm",
    title = "{Fluxes, vacua, and tadpoles meet Landau-Ginzburg and Fermat}",
    eprint = "2210.03706",
    archivePrefix = "arXiv",
    primaryClass = "hep-th",
    doi = "10.1007/JHEP12(2022)083",
    journal = "JHEP",
    volume = "12",
    pages = "083",
    year = "2022"
}

@article{Becker:2024ijy,
    author = "Becker, Katrin and Rajaguru, Muthusamy and Sengupta, Anindya and Walcher, Johannes and Wrase, Timm",
    title = "{Stabilizing massless fields with fluxes in Landau-Ginzburg models}",
    eprint = "2406.03435",
    archivePrefix = "arXiv",
    primaryClass = "hep-th",
    doi = "10.1007/JHEP08(2024)069",
    journal = "JHEP",
    volume = "08",
    pages = "069",
    year = "2024"
}

@article{Rajaguru:2024emw,
    author = "Rajaguru, Muthusamy and Sengupta, Anindya and Wrase, Timm",
    title = "{Fully stabilized Minkowski vacua in the 2$^{6}$ Landau-Ginzburg model}",
    eprint = "2407.16756",
    archivePrefix = "arXiv",
    primaryClass = "hep-th",
    doi = "10.1007/JHEP10(2024)095",
    journal = "JHEP",
    volume = "10",
    pages = "095",
    year = "2024"
}

@article{Bardzell:2022jfh,
    author = "Bardzell, Jacob and Gonzalo, Eduardo and Rajaguru, Muthusamy and Smith, Danielle and Wrase, Timm",
    title = "{Type IIB flux compactifications with h$^{1,1}$ = 0}",
    eprint = "2203.15818",
    archivePrefix = "arXiv",
    primaryClass = "hep-th",
    doi = "10.1007/JHEP06(2022)166",
    journal = "JHEP",
    volume = "06",
    pages = "166",
    year = "2022"
}

@article{Reece:2026hfk,
    author = "Reece, Matthew and Rudelius, Tom",
    title = "{Nonabelian Lattice Weak Gravity Conjecture and Monopole Confinement}",
    eprint = "2603.04494",
    archivePrefix = "arXiv",
    primaryClass = "hep-th",
    month = "3",
    year = "2026"
}

@article{Bedroya:2025fie,
    author = "Bedroya, Alek and Steinhardt, Paul J.",
    title = "{Holographic Constraints on the String Landscape}",
    eprint = "2511.15784",
    archivePrefix = "arXiv",
    primaryClass = "hep-th",
    month = "11",
    year = "2025"
}

@article{Bedroya:2025ltj,
    author = "Bedroya, Alek and Steinhardt, Paul J.",
    title = "{Holography vs. Scale Separation}",
    eprint = "2509.25313",
    archivePrefix = "arXiv",
    primaryClass = "hep-th",
    month = "9",
    year = "2025"
}

\end{document}